\documentclass[11pt, numbers, sort&compress]{elsarticle}
\usepackage{epsfig,graphics}
\usepackage{amssymb,amsmath,amsthm,bm, amscd, amsfonts}
\usepackage{booktabs}  
\usepackage{threeparttable}  
\usepackage{multirow}
\usepackage{mathabx}
\usepackage{siunitx}
\usepackage{booktabs}
\usepackage{algorithm}
\usepackage{algpseudocode}
\usepackage{subcaption}
\usepackage[hidelinks]{hyperref}
\usepackage{float}
\usepackage[title]{appendix}
\usepackage{orcidlink}

\theoremstyle{plain}

\newtheorem{theo}{Theorem}
\newtheorem{lem}{Lemma}
\newtheorem{pro}{Proposition}

\newtheorem{mrem}{Remark}
\newtheorem{mdef}{Definition}

\def\RR{{\mathbb R}}

\def\gO{\Omega}

\def\gl{\lambda}
\def\gL{\Lambda}
\def\wt{\widetilde}

\def\b0{{\bf 0}}
\def\1{{\bf 1}}

\def\cF{\mathcal{F}}

\def\wh{\widehat}
\def\wt{\widetilde}

\makeatletter
\def\widebreve{\mathpalette\wide@breve}
\def\wide@breve#1#2{\sbox\z@{$#1#2$}%
     \mathop{\vbox{\m@th\ialign{##\crcr
\kern0.08em\brevefill#1{0.99\wd\z@}\crcr\noalign{\nointerlineskip}%
                    $\hss#1#2\hss$\crcr}}}\limits}
\def\brevefill#1#2{$\m@th\sbox\tw@{$#1($}%
  \hss\resizebox{#2}{\wd\tw@}{\rotatebox[origin=c]{90}{\upshape(}}\hss$}
\makeatletter

\begin{document}
\begin{frontmatter}

 \title{High-order synchrosqueezed wavelet-chirplet transform       
 for instantaneous frequency and chirprate estimation\\
 {\small Original version: September 20, 2025}\\
 {\small Revised version: April 10, 2026}\\
{\small Accepted by JCAM on May 31, 2026}}   

\author[university1]{Shuixin Li}
\ead{lishuixin@zjnu.edu.cn}

\author[university1]{Jiecheng Chen}
\ead{jcchen@zjnu.edu.cn}

\author[university1]{Qingtang Jiang}
\ead{jiangq@zjnu.edu.cn}

\author[university2]{Gang Yu}
\ead{cse_yug@ujn.edu.cn}

\address[university1]{School of Mathematical Sciences, Zhejiang Normal University, Jinhua 321004, China}
\address[university2]{School of Electrical Engineering, University of Jinan, Jinan 250022, China}

\begin{abstract}
The separation of multicomponent signals with crossing instantaneous frequency (IF) curves remains a fundamental challenge in time-frequency analysis. 
Although the synchrosqueezed wavelet-chirplet transform (SWCT) enhances time-frequency readability by introducing a chirprate variable, its effectiveness is constrained by the underlying assumption of local linear chirp. Consequently, this method does not perform well when analyzing signals characterized by strong frequency modulation.
This paper extends the SWCT framework by relaxing the linear chirp assumption. We model signal components as having polynomial phase behavior over short intervals and derive compact expressions for high-order IF and chirprate reassignment operators. The proposed high-order synchrosqueezed wavelet-chirplet transform (HSWCT) enables accurate estimation of both IF and chirprate, and supports robust mode retrieval even with intersecting IF curves.
Another key contribution is a rigorous mathematical analysis of the approximation errors of arbitrary-order reassignment operators for IF and chirprate estimation. When the chirprate vanishes, HSWCT simplifies to the traditional high-order synchrosqueezed wavelet transform. To our best knowledge, no theoretical analysis exists in the literature on the approximation of arbitrary-order SST IF reassignment operators to the  IF. As a by-product of this work, our established theorem provides such an analysis, thereby filling a gap in the theoretical framework of high-order SSTs.

\end{abstract}


\end{frontmatter}

\section{Introduction}
In real-world applications, signals from fields such as seismology \cite{liu2023sparse,fomel2013seismic}, biomedicine \cite{park2011time,kinoshita2020sleep}, sonar \cite{mercuri2019vital,amin2016radar}, and mechanical systems \cite{yu2020time,wang2018matching} typically appear as a superposition of amplitude-modulated and frequency-modulated (AM-FM) modes:
\begin{equation}
     \label{AHM0}
    x(t)= \sum_{k=1}^{K} x_k(t) = \sum_{k=1}^{K} A_k(t) e^{2 \pi i \phi_k(t)},
\end{equation}
where $A_k(t) > 0$ denotes the amplitude of the $k$-th component, with $\phi_k(t)$ its instantaneous phase, and $\phi_k'(t) > 0$ its instantaneous frequency (IF).
The amplitude $A_k(t)$ is assumed to vary slowly. 
To gain a deeper understanding of multicomponent signals, a fundamental challenge in signal processing is mode decomposition--recovering the constituent components $x_k(t)$ from composite signals. Despite its conceptually simple formulation, this problem remains challenging in real-world applications.

Time-frequency analysis (TFA) is essential for processing non-stationary signals of the form Eq.~\eqref{AHM0}. Linear methods like the short-time Fourier transform (STFT) \cite{stankovic2013time} and continuous wavelet transform (CWT) \cite{daubechies1992ten, mallat1999wavelet} are inherently limited by the uncertainty principle, preventing optimal joint time-frequency resolution. 
In contrast, quadratic representations, such as the Wigner--Ville distribution \cite{hlawatsch1992linear} and its Cohen class generalizations \cite{cohen1995time}, offer improved resolution but suffer from interference terms.
To improve the readability of time-frequency representations (TFRs), post-processing methods such as the reassignment technique \cite{kodera1978analysis, auger1995improving} have been developed. Reassignment sharpens TFRs by moving spectrogram energy to the centroid of the signal's energy distribution, but at the expense of losing signal reconstruction capability.
Another well-established post-processing technique is the synchrosqueezing transform (SST).
Originally proposed in \cite{daubechies1996nonlinear} and further developed in \cite{daubechies2011synchrosqueezed}, SST has since attracted considerable research attention in time-frequency analysis and non-stationary signal processing. This seminal work has inspired numerous extensions, including the STFT-based SST \cite{oberlin2014fourier}, the synchroextracting transform \cite{yu2017synchroextracting}, and the local maximum synchrosqueezing transform \cite{yu2019local}, to improve energy concentration and mode retrieval.
For signals with strong frequency modulation, second-order SST \cite{oberlin2017secondwavelet} and its STFT-based counterpart \cite{oberlin2015second, behera2018theoretical, lu2021second} were introduced, later extended to higher orders \cite{pham2017high, hu2019high}.
More recently, the time-reassigned synchrosqueezing transform \cite{he2019time,  fourer2019second, he2020gaussian, li2022theoretical} has been proposed to address signals with rapid frequency variations or transient components.

For multicomponent signals with overlapping IF curves, conventional TFA methods based on the 2D time-frequency plane inevitably suffer from blurring and smearing near intersection points. To overcome this limitation, a widely adopted strategy is to incorporate an additional chirprate dimension, thereby enabling analysis in 3D spaces.
The chirplet transform (CT) \cite{mann1995chirplet} pioneered this approach by extending the STFT with a chirprate parameter, effectively lifting the 2D time-frequency plane into a 3D time-frequency-chirprate (TFC) space. (We use the contraction ``chirprate'' throughout) This allows signal components that are entangled in the 2D plane to be well separated in the 3D space.
The wavelet-chirplet transform (WCT) \cite{chui2021time} follows a similar philosophy, introducing the chirprate parameter into the time-scale plane to form a 3D time-scale-chirprate (TSC) representation.
For mode retrieval, the signal separation operator (SSO) \cite{chui2016signal} was extended to TFC and TSC spaces \cite{chui2021time,li2022chirplet,chui2023analysis}, enabling accurate component reconstruction from the estimated IF and chirprate curves.
Building on these foundations, several methods have been developed to further enhance analysis in 3D spaces. 
For example, the synchrosqueezed chirplet transform (SCT) \cite{zhu2020frequency,chen2023disentangling} and the synchrosqueezed windowed linear canonical transform (SWLCT) \cite{li2026synchrosqueezed} were developed to sharpen TFC representations, thereby enabling more accurate IF and chirprate estimation.
The local maximum frequency-chirprate SCT \cite{zhang2022local} further refines this idea by concentrating energy around local maxima in the TFC space.
Meanwhile, the 3D extracting transform \cite{zhu2021three} offers an alternative approach to simultaneously extract IF and chirprate information.
More recently, multiple synchrosqueezing techniques have been extended to 3D spaces, as exemplified by the multiple enhanced synchrosqueezing chirplet transform \cite{chen2024multiple} and the high-order multiple synchrosqueezing wavelet transform within the WCT framework \cite{chen2024composite}.


Although these methods have shown effectiveness in estimating IFs and retrieving modes for signals with overlapping IFs, most rely on the assumption that signals can be locally approximated as linear chirps. For signals with strongly modulated AM-FM components, however, their performance degrades significantly and often proves unsatisfactory.
In this paper, we focus on the WCT framework and, by assuming that signal modes behave locally as high-order polynomial phase signals (PPSs) \cite{wang2008generalized, jiang2020novel} over short intervals, we propose a novel method for deriving high-order IF and chirprate estimators in a concise general form. For signals with $N$th-degree polynomial phase, we derive the corresponding $N$th-order IF and chirprate reassignment operators.
Building on these high-order estimators, we introduce the high-order synchrosqueezed wavelet chirplet transform (HSWCT), which enables significantly more precise extraction of IF and chirprate information.

Another critical contribution of this paper is the development of mathematically rigorous analyses for two key approximation errors: those of arbitrary-order IF reassignment operators for IF estimation, and arbitrary-order chirprate reassignment operators for chirprate estimation.
When the chirprate is zero, the WCT simplifies to the CWT, and the high-order SWCT reduces to a high-order SST. To the best of our knowledge, the literature currently contains no theoretical analysis of how IF reassignment operators of an arbitrary-order SST approximate the IF. Most existing works focus on first- or second-order approximations, leaving higher-order cases largely unexplored from a theoretical perspective.
As a valuable by-product of this research, the theorem we establish provides such an analysis. Specifically, we derive explicit bounds on the approximation errors for arbitrary-order IF reassignment operators under sufficient smoothness conditions on the signal's phase. This result not only fills a significant gap in the theoretical foundation of high-order SSTs but also validates the practical effectiveness of higher-order reassignment techniques.

The remainder of this paper is organized as follows. Section 2 reviews the WCT-based scheme and derives high-order IF and chirprate reassignment operators for the HSWCT. 
Theoretical error estimates for the proposed arbitrary-order IF and chirprate reassignment operators are provided in Section 3. 
Section 4 then presents the methodological framework, covering algorithmic implementation, parameter selection, and mode retrieval. 
Numerical experiments are reported in Section 5, and Section 6 concludes the paper with a discussion of future work.

\section{High-order synchrosqueezed wavelet-chirplet transform} 
In this paper, we denote the Schwartz space by \(\mathcal{S}(\mathbb{R})\) and the space of tempered distributions by \(\mathcal{S'}(\mathbb{R})\) \cite{grafakos2008classical}.
The Fourier transform of the wavelet function \(\psi(t) \in \mathcal{S}(\mathbb{R})\) is defined as
\[
\widehat{\psi}(\eta) := \int_{\mathbb{R}} \psi(t) e^{-i 2\pi \eta t} \, dt.
\]
In the references \cite{katkovnik1995new, li2011local}, the $N$th-order polynomial Fourier transform is given by 
\begin{equation}
\label{def_PFT}
\cF^N(\psi)(\eta_1, \eta_2, \cdots, \eta_N) := \int_{\mathbb{R}} \psi(t) e^{-i 2\pi ( \eta_1 t + \frac{1}{2!}\eta_2 t^2 + \cdots + \frac{1}{N!} \eta_N t^N )} dt, 
\end{equation}
where $\eta_1, \eta_2, \cdots, \eta_N \in \RR$. 

Given a signal \( x(t) \in L^2(\mathbb{R}) \) (or \( x(t) \in \mathcal{S}'(\mathbb{R}) \) in the distributional sense), the wavelet-chirplet transform (WCT), which is also called the time-scale-chirprate transform/operator in \cite{chui2021time}, with a wavelet function \( \psi(t) \in \mathcal{S}(\mathbb{R}) \) is given by
\begin{align}
     U_x^\psi(a,b,\lambda) : &=\int_{\mathbb{R}} x(t)\frac 1a \psi\big(\frac{t-b}{a}\big) e^{-i \pi \lambda (t-b)^2} dt \label{def_WCT_wavelet} \\
     &=\int_{\mathbb{R}} x(b+at) \psi(t) e^{-i \pi \lambda a^2 t^2} dt \label{def_WCT}.
\end{align}

The notation \( U_x^{t^m \psi^{(j)}}(a,b,\lambda) \) generalizes the definition in Eq.~\eqref{def_WCT_wavelet}. Specifically, it is obtained by replacing the original wavelet \( \psi(t) \) with its modified form \( t^m \psi^{(j)}(t) \), where \( \psi^{(j)}(t) \) denotes the \( j \)-th derivative of \( \psi(t) \), and \( j, m \) are non-negative integers. 
For brevity, we may omit the variables \((a,b,\lambda)\) and denote it simply as \( U_x^{t^m {\psi^{(j)}}} \) when no ambiguity arises.

\begin{mrem}
In the original WCT definition \cite{chui2021time}, the complex conjugate of the wavelet function, \( \overline{\psi(t)} \), is used in Eq.~\eqref{def_WCT_wavelet} and Eq.~\eqref{def_WCT}. 
For simplicity, we omit the conjugate operation and work directly with \( \psi(t) \) throughout this paper. Besides, the WCT framework does not require \(\psi(t)\) to belong to the Schwartz space \(\mathcal{S}(\mathbb{R})\). 
While conventional second-order reassignment operators involve the first two derivatives \(\psi'(t)\) and \(\psi''(t)\), the \(N\)th-order reassignment operators derived here are expressed in terms of \(U_x^{t^m \psi}\) and \(U_x^{t^n \psi'}\) for \(0 \le m \le 2N-2\) and \(0\le n\le N-1\). 
To support these expressions, it is sufficient that \(\psi(t)\) is twice continuously differentiable and satisfies certain decay conditions, specifically that \(t^{m}\psi(t), t^{n}\psi'(t) \in L^1(\mathbb{R})\). 
In this paper, we take \(\psi(t) \in \mathcal{S}(\mathbb{R})\) as a convenient sufficient condition that encompasses these requirements, thereby streamlining the presentation.
\end{mrem}

The definitions of the HSWCT rely on both IF and chirprate reassignment operators. The second-order versions of these operators can be derived using methods analogous to those for second-order SST \cite{oberlin2015second, oberlin2017secondwavelet, behera2018theoretical, li2020adaptive, li2020adaptivestft} and for SCT \cite{zhu2020frequency, chen2023disentangling}. 
In Subsection \ref{second_order_reassignment}, we present a conventional approach to obtain these second-order operators, following the methodology in the references above. This serves as a baseline for comparison with the novel method proposed in this paper, which derives high-order IF and chirprate reassignment operators and is presented in Subsection \ref{high_order_reassignment}.

\subsection{  Second-order reassignment operators derived via conventional method}\label{second_order_reassignment}
Building upon the WCT framework, we now  derive the second-order IF and chirprate reassignment operators.
Taking the partial derivative of both sides of Eq.~\eqref{def_WCT} with respect to \( b \), we obtain the following relation for a signal \( x(t) \):
\begin{eqnarray}
\nonumber {\partial _b} U^\psi_x(a, b, \gl)\hskip -0.6cm &&= \int_{\RR} x'(b+at)  \psi(t) e^{-i\pi \gl a^2 t^2}dt\\
\nonumber &&= -\frac 1a \int_{\RR} x(b+at) \big (\psi'(t)-i2\pi \gl a^2 t\psi(t) \big)e^{-i\pi \gl a^2 t^2} dt. 
\end{eqnarray}
Thus, 
\begin{equation}
\label{partial_b}
{\partial _b} U^\psi_x(a, b, \gl) =-\frac 1a  U^{\psi'}_x(a, b, \gl) +i2\pi \gl a  U^{t\psi}_x(a, b, \gl).   
\end{equation}
Similarly, by taking the partial derivatives of both sides of Eq.~\eqref{def_WCT} with respect to \( a \) and \( \lambda \), we obtain the  results:
\begin{align}     
   \label{partial_a} 
\partial_a U_x^\psi(a,b,\lambda) &=-\frac 1 a  U^{\psi}_x(a, b, \gl)-\frac 1 a  U^{t\psi'}_x(a, b, \gl),    \\
   \label{partial_lambda} 
\partial_\lambda U_x^\psi(a,b,\lambda) &=   -i \pi a^2 U_x^{t^2 \psi}(a,b,\lambda).  
\end{align}
Furthermore, we introduce a new differential operator
\begin{align}
\label{Dalambda}
D_{a,\lambda} U_x^\psi(a,b,\lambda): = \partial_a U_x^\psi(a,b,\lambda) -\frac{2\lambda}{a} \partial_\lambda U_x^\psi(a,b,\lambda) .
\end{align}
Note that the operator \(D_{a,\lambda}\) is a linear combination of the partial derivatives  \(\partial_a\) and \(\partial_\lambda\). Specifically, it satisfies the following relation:
\begin{align}
   \label{partial_tpsi} D_{a,\lambda} U_x^\psi(a,b,\lambda) = {\partial_b} U_x^{t\psi}(a, b, \lambda).
\end{align}

 We say $x(t)$ is a (generalized) linear chirp if
 \begin{equation}
 \label{def_chirp}
 x(t)=A(t) e^{i2\pi \phi(t)}= e^{pt+\frac 12 qt^2} e^{i2\pi (ct +\frac 12 r t^2)},   
 \end{equation}
 where $ p,~  q,~  c$ and $r$ are real-valued constants. 
 Based on \( x(t) \) given in Eq.~\eqref{def_chirp}, we take the partial derivative of both sides of Eq.~\eqref{def_WCT} with respect to \(b\),
 \begin{eqnarray*}
 \nonumber {\partial _b} U^\psi_x(a, b, \gl)\hskip -0.6cm &&= \int_{\RR} x'(b+at)  g(t) e^{ -i\pi \gl a^2 t^2}dt\\
 &&= \big(p+qb+i2\pi(c+rb)\big) U^\psi_x(a, b, \gl)+ (q+i2\pi r) a U^{t \psi}_x(a, b, \gl). 
 \end{eqnarray*}
 Hence, when \(U^\psi_x(a, b, \gl)\neq 0\), then
 \begin{equation}
   \label{eqution1stdiff}  \frac{\partial_b U^\psi_x(a, b, \gl)}{U^\psi_x(a, b, \gl)}=\big(p+qb+i2\pi(c+rb)\big) + (q+i2\pi r) a \frac{U^{t\psi}_x(a, b, \gl)}{U^{\psi}_x(a, b, \gl)}. 
 \end{equation}
 Differentiating both sides of Eq.~\eqref{eqution1stdiff} with respect to \( b \), 
 \begin{equation}
\label{derivation1}
     \partial_b\Big(\frac{\partial_b U^\psi_x}{U^\psi_x}\Big)  = (q + i 2 \pi r) + (q + i 2 \pi r) a \partial_b\Big(\frac{U^{t\psi}_x}{U^\psi_x}\Big).  
  \end{equation}

We introduce the notation
\begin{equation*}
D_x^\psi(a,b,\lambda) := U^{\psi'}_x U_x^{t\psi} - U^{\psi}_x U_x^{t \psi'} + i 2\pi \lambda a^2 \big(U_x^{t^2 \psi} U_x^{\psi} - U_x^{t \psi} U_x^{t \psi}\big).
\end{equation*}
For any \((a, b, \lambda)\) such that \(D_x^\psi(a,b,\lambda) \neq 0\), 
Eq.~\eqref{derivation1} together with Eq.~\eqref{partial_b}, yields
\begin{align}
\frac q{2\pi} + i r 
&= \frac 1{2\pi} \frac{\partial_b\Big(\frac{\partial_b U^\psi_x}{U^\psi_x}\Big)}{1 + a \partial_b\Big(\frac{U^{t\psi}_x}{U^\psi_x}\Big)} \nonumber \\
&= i\gl+\frac 1{2 \pi a^2} \frac {{U_x^{\psi}U_x^{\psi''}-U_x^{\psi'}U_x^{\psi'}+i 2\pi \lambda a^2 \big( U_x^{t\psi}U_x^{\psi'}-U_x^{t\psi'}U_x^{\psi}-U_x^{\psi}U_x^{\psi}\big)}}{ U^{\psi'}_xU_x^{t\psi}-U^{\psi}_xU_x^{t\psi'}+i 2\pi \lambda a^2\big(U_x^{t^2\psi}U_x^{\psi}-U_x^{t\psi}U_x^{t\psi}\big)}. \label{derivation2}
\end{align}
Combining this with Eq.~\eqref{eqution1stdiff} gives
\begin{align}
\frac {p+qb}{2\pi}+i(c+rb) 
&= \frac{\partial_b U^\psi_x}{2\pi U^\psi_x} - \bigl(\frac q{2\pi}+ir\bigr) a \frac{U^{t\psi}_x}{U^{\psi}_x} \;= \frac{\partial_b U^\psi_x} {2\pi U^\psi_x} - \frac{a U^{t \psi}_x\,\partial_b\left(\frac{\partial_b U^\psi_x}{U^\psi_x}\right)}{ 2\pi U^\psi_x \big(1 + a \partial_b\bigl(\frac{U^{t\psi}_x}{U^\psi_x}\bigr)\big)} \nonumber \\
&= -\frac{1}{2\pi a}  \frac {U_x^{t \psi}U_x^{\psi''}-U_x^{t\psi'}U_x^{\psi'}+i 2\pi \lambda a^2 \big( U_x^{t^2 \psi}U_x^{\psi'}-U_x^{t\psi'}U_x^{t\psi}-U_x^{t \psi}U_x^\psi \big)}{ U^{\psi'}_xU_x^{t\psi}-U^{\psi}_xU_x^{t\psi'}+i 2\pi \lambda a^2 \big( U_x^{t^2 \psi}U_x^{\psi}-U_x^{t\psi}U_x^{t\psi}\big)}. \label{derivation3}
\end{align}

Taking the imaginary parts of Eq.~\eqref{derivation3} and Eq.~\eqref{derivation2}, we obtain the second-order IF reassignment operator and the chirprate reassignment operator, respectively:
\begin{align}
\label{def_gO}
\gO^{\psi}_x(a, b, \lambda) &:= -\frac{1}{2\pi a} \operatorname{Im}\Big(\frac{U_x^{t \psi}U_x^{\psi''} - U_x^{t\psi'}U_x^{\psi'} + i 2\pi \lambda a^2 \bigl(U_x^{t^2 \psi}U_x^{\psi'} - U_x^{t\psi'}U_x^{t\psi} - U_x^{t \psi}U_x^\psi\bigr)}{U^{\psi'}_xU_x^{t\psi} - U^{\psi}_xU_x^{t\psi'} + i 2\pi \lambda a^2 \bigl(U_x^{t^2 \psi}U_x^{\psi} - U_x^{t\psi}U_x^{t\psi}\bigr)}\Big), \\
\label{def_gL}
\gL^{\psi}_x(a, b, \lambda) &:= \lambda + \frac{1}{2\pi a^2} \operatorname{Im}\Big(\frac{U_x^{\psi}U_x^{\psi''} - U_x^{\psi'}U_x^{\psi'} + i 2\pi \lambda a^2 \bigl(U_x^{t\psi}U_x^{\psi'} - U_x^{t\psi'}U_x^{\psi} - U_x^{\psi}U_x^{\psi}\bigr)}{U^{\psi'}_xU_x^{t\psi} - U^{\psi}_xU_x^{t\psi'} + i 2\pi \lambda a^2 \bigl(U_x^{t^2\psi}U_x^{\psi} - U_x^{t\psi}U_x^{t\psi}\bigr)}\Big).
\end{align}

As shown above, the reassignment operators Eq.~\eqref{def_gO} and Eq.~\eqref{def_gL} rely on the assumption Eq.~\eqref{def_chirp}, yet their derivation is rather involved.
In \cite{chen2024composite}, this approach is extended to third-order PPS approximations within the WCT framework.
Our recent work \cite{jiang2025synchrosqueezed} provides explicit expressions for the resulting third-order reassignment operators and shows their superiority over the second-order ones.
Nevertheless, extending this methodology to higher orders remains challenging, as the complexity of the derivation increases rapidly with the order.

\subsection{ High-order IF and chirprate  reassignment operators} \label{high_order_reassignment}

In this subsection, we develop a new approach for constructing high-order  IF and chirprate reassignment operators. We begin by considering \(x(t)\) 
of the form
\begin{equation}
\label{def_generalized_PPS}
x(t)=A(t) e^{i 2\pi \phi(t)}=e^{d(t)} e^{i 2\pi \phi(t)},  
\end{equation}
where 
\(\phi(t)\) is a real-valued polynomial of degree \(N\). 
When $d(t)$ is constant, the signal $x(t)$ defined in \eqref{def_generalized_PPS} is referred to as a polynomial phase signal (PPS). Our modeling assumption for \eqref{AHM0} requires that $A(t)$ varies slowly, or equivalently, that $d(t)$ behaves as a constant. Such a constraint on $A(t)$ is also necessary for the theoretical derivation of error analysis for the approximation of the proposed high-order reassignment operators to IFs and chirprates. 

Nevertheless, the result established in Proposition 1 of this subsection concerning high-order reassignment operators remains valid when $d(t)$ is a real-valued polynomial of degree $N$. Accordingly, we assume $d(t)$ to be a real-valued polynomial of degree $N$ throughout this and the following subsections, so that Proposition 1 applies to $d(t)$ in its most general form. When both $d(t)$ and $\phi(t)$ are real-valued polynomials, we refer to the signal $x(t)$ in \eqref{def_generalized_PPS} as a generalized PPS.

It is natural to perform a simple Taylor expansion 
\begin{align}  
  \label{Taylor_expansion_WCT}  x(b+at)=e^{\sum_{j=0}^{N}\frac 1{j!}\big( d^{(j)}(b) + i 2\pi \phi^{(j)}(b) \big)(at)^j}.
\end{align} 
Substituting Eq.~\eqref{Taylor_expansion_WCT} into Eq.~\eqref{def_WCT}, we have
\begin{align}
U_x^\psi(a,b,\lambda) 
\label{WCT_polynomials}                     
&= \int_{\mathbb{R}} e^{\sum_{j=0}^{N} \frac1{j!}\big(d^{(j)}(b) + i 2\pi \phi^{(j)}(b) \big)(at)^j} \psi(t) e^{-i \pi \lambda a^2 t^2} dt.
\end{align}
Taking the partial derivative of Eq.~\eqref{WCT_polynomials} with respect to the variable time $b$, one has 
\begin{align}
\label{partial_b_WCT_polynomials}  \partial_b U_x^\psi(a,b,\lambda) &= {\sum_{j=1}^{N} \frac{d^{(j)}(b) +  i 2\pi \phi^{(j)}(b) }{(j-1)!} }a^{j-1}   U_x^{ {t}^{j-1} \psi }(a,b,\lambda). 
\end{align}
With $D_{a,\lambda}$ defined in Eq.~\eqref{Dalambda}, substituting it into both sides of Eq.~\eqref{WCT_polynomials} yields
\begin{align}
\label{D_a_WCT_polynomials}  D_{a,\lambda} U_x^{\psi}(a,b,\lambda) &= {\sum_{j=1}^{N} \frac{d^{(j)}(b) +i 2\pi \phi^{(j)}(b) }{(j-1)!} }a^{j-1}   U_x^{ {t}^{j} \psi }(a,b,\lambda). 
\end{align}
From Eq.~\eqref{partial_tpsi}, we know Eq.~\eqref{D_a_WCT_polynomials} is in fact
\begin{align}
    \label{partial_tb_WCT_polynomials}  \partial_b U_x^{t\psi}(a,b,\lambda) &= {\sum_{j=1}^{N} \frac{d^{(j)}(b) +i 2\pi \phi^{(j)}(b) }{(j-1)!} }a^{j-1}   U_x^{ {t}^{j} \psi }(a,b,\lambda). 
    \end{align}

If we consider Eq.~\eqref{partial_b_WCT_polynomials} and Eq.~\eqref{partial_tb_WCT_polynomials} as a system of equations in the variables \(\frac{d^{(j)}(b) + i 2\pi \phi^{(j)}(b) }{(j-1)!} a^{j-1}\) for \(j=1,2,\dots,N\), then solving this system requires  at least \(N\) equations. 
As we currently have only two, we need to derive additional equations.   
Differentiating both sides of Eq.~\eqref{partial_b_WCT_polynomials} with respect to the variable \(\lambda\), we  obtain
 \begin{align}
    \label{partial_lambdab_WCT_polynomials}  \partial_\lambda \partial_b U_x^\psi(a,b,\lambda) &= {\sum_{j=1}^{N} \frac{d^{(j)}(b) +  i 2\pi \phi^{(j)}(b) }{(j-1)!} }a^{j-1}  \partial_\lambda U_x^{ {t}^{j-1} \psi }(a,b,\lambda). 
    \end{align}
    Since $\psi(t) \in \mathcal{S}(\mathbb{R})$, the transform $U_x^{t^n \psi}(a, b, \lambda)$ is always well-defined, and the interchange of differentiation order is justified.  
    Therefore, combining Eq.~\eqref{partial_lambda}, we can rewrite Eq.~\eqref{partial_lambdab_WCT_polynomials} as 
    \begin{align}
\label{partial_t2b_WCT_polynomials}
        \partial_b  U_x^{t^2\psi}(a,b,\lambda) &= {\sum_{j=1}^{N} \frac{d^{(j)}(b) +  i 2\pi \phi^{(j)}(b)}{(j-1)!} }a^{j-1}   U_x^{ {t}^{j+1} \psi }(a,b,\lambda). 
    \end{align}
    Similarly, differentiating Eq.~\eqref{partial_tb_WCT_polynomials} with respect to \(\lambda\) and  applying Eq.~\eqref{partial_lambda}, we obtain 
    \begin{align*}
       \partial_b  U_x^{t^3\psi}(a,b,\lambda) &= {\sum_{j=1}^{N} \frac{d^{(j)}(b) +  i 2\pi \phi^{(j)}(b) }{(j-1)!} }a^{j-1} U_x^{ {t}^{j+2} \psi }(a,b,\lambda). 
    \end{align*}
Repeat this process until we arrive at the following:
\begin{align*}
    \partial_b U_x^{t^{N-1} \psi} =  \sum_{j=1}^{N} \frac{d^{(j)}(b) +  i 2\pi \phi^{(j)}(b)}{(j-1)!}a^{j-1} U_x^{t^{j+N-2} \psi}.
\end{align*}
We can also express the system in matrix form as:
\begin{align}
    \label{system_of_equations_2}
    \mathbf{y}^N = \mathbf{D}^N \; \mathbf{x}^N,
\end{align}
where 
\begin{align}
    \label{matrix2}
& \mathbf{D}^N:=
    \begin{bmatrix}
        U_x^{\psi} & U_x^{t \psi} & \cdots & U_x^{t^{N-1} \psi} \\[0.5em]
        U_x^{t \psi} & U_x^{t^2 \psi} & \cdots & U_x^{t^{N} \psi} \\[0.5em]
        U_x^{t^2 \psi} & U_x^{t^3 \psi} & \cdots & U_x^{t^{N+1} \psi} \\[0.5em]
        \vdots & \vdots & \ddots & \vdots \\[0.5em]
        U_x^{t^{N-1} \psi} & U_x^{t^{N} \psi} & \cdots & U_x^{t^{2N-2} \psi}
    \end{bmatrix}, 
 \mathbf{y}^N:=\begin{bmatrix}
        \partial_b U_x^{\psi} \\[0.5em]
        \partial_b U_x^{t \psi} \\[0.5em]
        \partial_b U_x^{t^2 \psi} \\[0.5em]
        \vdots \\[0.5em]
        \partial_b U_x^{t^{N-1} \psi}
    \end{bmatrix}, 
  \mathbf{x}^N:=  \begin{bmatrix}
        d'(b) + i 2\pi \phi'(b) \\[0.5em]
        \left(d''(b) + i 2\pi \phi''(b)\right) a \\[0.5em]
        \frac{d'''(b) + i 2\pi \phi'''(b)}{2!} a^2 \\[0.5em]
        \vdots \\[0.5em]
        \frac{d^{(N)}(b) + i 2\pi \phi^{(N)}(b)}{(N-1)!} a^{N-1}
    \end{bmatrix}. 
\end{align}

Let  $D^N_0$ denote the determinant of $\mathbf{D}^N$, that is 
\begin{align}
    D^N_0:=\left|\mathbf{D}^N\right|.
\label{D^N_0}
\end{align}
By replacing the first and second columns of $\mathbf{D}^N$ with \(\mathbf{y}^N\), respectively, we  obtain the determinants of the resulting matrices:  
      \begin{align}
\label{DN_1_2}
        D^N_1 = \begin{vmatrix}
            \partial_b U_x^{\psi}  & U_x^{t \psi} & \cdots & U_x^{t^{N-1} \psi} \\[0.5em]
            \partial_b U_x^{t \psi} & U_x^{t^2 \psi}  & \cdots & U_x^{t^{N} \psi} \\[0.5em]
            \partial_b U_x^{t^2 \psi} & U_x^{t^3 \psi} & \cdots & U_x^{t^{N+1} \psi} \\[0.5em]
            \vdots & \vdots & \ddots & \vdots \\[0.5em]
            \partial_b U_x^{t^{N-1} \psi} & U_x^{t^{N} \psi} & \cdots & U_x^{t^{2N-2} \psi} \\[0.5em]
            \end{vmatrix}, \quad
            D^N_2 = \begin{vmatrix}
                U_x^{\psi}  & \partial_b U_x^{ \psi} & \cdots & U_x^{t^{N-1} \psi} \\[0.5em]
                U_x^{t \psi} & \partial_b U_x^{t \psi}  & \cdots & U_x^{t^{N} \psi} \\[0.5em]
                U_x^{t^2 \psi} &\partial_b  U_x^{t^2 \psi} & \cdots & U_x^{t^{N+1}\psi}\\[0.5em]
                \vdots & \vdots & \ddots & \vdots \\[0.5em]
                U_x^{t^{N-1} \psi} & \partial_b U_x^{t^{N-1} \psi} & \cdots & U_x^{t^{2N-2} \psi} \\[0.5em]
               \end{vmatrix}.
      \end{align} 
    With a threshold \(\epsilon > 0\), we define the region \(D_\epsilon\) as 
       \begin{equation}
         \label{Deps}  D_\epsilon:=\left\{ {(a,b,\lambda): \left|D^N_0(a,b,\lambda)\right| >\epsilon} \right\}.
        \end{equation} 
    Then, according to Cramer's rule, when $(a,b,\lambda) \in D_\epsilon$, 
    \begin{align*}
        d'(b)+ i 2\pi \phi'(b) =\frac{D^N_1}{D^N_0},\quad  a d''(b)+ i 2\pi a \phi''(b) =\frac{D^N_2}{D^N_0}.
    \end{align*}
   By utilizing Eq.~\eqref{partial_b}, $\mathbf{y}^N$ can be represented as 
    \begin{align}
        \label{matrix3}
        \begin{bmatrix}
            \partial_b U_x^{ \psi} \\[0.5em]
            \partial_b U_x^{t \psi} \\[0.5em]
            \partial_b U_x^{t^2 \psi}  \\[0.5em]
            \vdots \\[0.5em]
            \partial_b U_x^{t^{N-1} \psi} \\[0.5em]
        \end{bmatrix}
        &=
        -\frac{1}{a}
        \begin{bmatrix}
             U_x^{\psi'}  \\[0.5em]
            U_x^{t \psi'}   \\[0.5em]
             U_x^{t^2 \psi'}  \\[0.5em]
            \vdots \\[0.5em]
            U_x^{t^{N-1} \psi'}  \\[0.5em]
        \end{bmatrix}
        -\frac{1}{a}
        \begin{bmatrix}
              0  \\[0.5em]
             U_x^{\psi}  \\[0.5em]
             2U_x^{t \psi}  \\[0.5em]
            \vdots \\[0.5em]
           (N-1)U_x^{t^{N-2} \psi}  \\[0.5em]
        \end{bmatrix}
        +
        i 2\pi a \lambda
        \begin{bmatrix}
            U_x^ {t \psi} \\[0.5em]
            U_x^{t^2 \psi} \\[0.5em]
            U_x^{t^3 \psi}  \\[0.5em]
            \vdots \\[0.5em]
            U_x^{t^{N} \psi} \\[0.5em]
        \end{bmatrix}.
    \end{align}
 When we substitute the first column of determinant $D^N_1$ with Eq.~\eqref{matrix3}, we notice that the rightmost part of Eq.~\eqref{matrix3} coincides with the second column of $D^N_1$. 
 Similarly, for determinant $D^N_2$, when we substitute its second column with Eq.~\eqref{matrix3}, the rightmost part of Eq.~\eqref{matrix3} coincides with the second column of $D^N_0$ multiplied by $i2 \pi a \lambda$. 
 In reality, what needs to be calculated is the determinant formed by the first and second columns on the right of Eq.~\eqref{matrix3}.

To construct high-order reassignment operators, we first define the following determinants:
\begin{align}
\label{def_D11}
D^N_{11} &:= \begin{vmatrix}
    \scriptstyle U_x^{\psi'} & \scriptstyle U_x^{t \psi} & \cdots & \scriptstyle U_x^{t^{N-1} \psi} \\
    \scriptstyle U_x^{t \psi'} & \scriptstyle U_x^{t^2 \psi} & \cdots & \scriptstyle U_x^{t^{N} \psi} \\
    \scriptstyle U_x^{t^2 \psi'} & \scriptstyle U_x^{t^3 \psi} & \cdots & \scriptstyle U_x^{t^{N+1} \psi} \\
    \scriptstyle \vdots & \scriptstyle \vdots & \ddots & \scriptstyle \vdots \\
    \scriptstyle U_x^{t^{N-1} \psi'} & \scriptstyle U_x^{t^{N} \psi} & \cdots & \scriptstyle U_x^{t^{2N-2} \psi}
\end{vmatrix}, \quad
D^N_{12} := \begin{vmatrix}
    \scriptstyle 0 & \scriptstyle U_x^{t \psi} & \cdots & \scriptstyle U_x^{t^{N-1} \psi} \\
    \scriptstyle U_x^{\psi} & \scriptstyle U_x^{t^2 \psi} & \cdots & \scriptstyle U_x^{t^{N} \psi} \\
    \scriptstyle 2U_x^{t \psi} & \scriptstyle U_x^{t^3 \psi} & \cdots & \scriptstyle U_x^{t^{N+1} \psi} \\
    \scriptstyle \vdots & \scriptstyle \vdots & \ddots & \scriptstyle \vdots \\
    \scriptstyle (N-1)U_x^{t^{N-2}\psi} & \scriptstyle U_x^{t^{N} \psi} & \cdots & \scriptstyle U_x^{t^{2N-2} \psi}
\end{vmatrix},
\end{align}
and
\begin{align*}
D^N_{21} &:= \begin{vmatrix}
    \scriptstyle U_x^{\psi} & \scriptstyle U_x^{\psi'} & \cdots & \scriptstyle U_x^{t^{N-1} \psi} \\
    \scriptstyle U_x^{t \psi} & \scriptstyle U_x^{t \psi'} & \cdots & \scriptstyle U_x^{t^{N} \psi} \\
    \scriptstyle U_x^{t^2 \psi} & \scriptstyle U_x^{t^2 \psi'} & \cdots & \scriptstyle U_x^{t^{N+1} \psi} \\
    \scriptstyle \vdots & \scriptstyle \vdots & \ddots & \scriptstyle \vdots \\
    \scriptstyle U_x^{t^{N-1} \psi} & \scriptstyle U_x^{t^{N-1} \psi'} & \cdots & \scriptstyle U_x^{t^{2N-2} \psi}
\end{vmatrix}, \quad
D^N_{22} := \begin{vmatrix}
    \scriptstyle U_x^{\psi} & \scriptstyle 0 & \cdots & \scriptstyle U_x^{t^{N-1} \psi} \\
    \scriptstyle U_x^{t \psi} & \scriptstyle U_x^{\psi} & \cdots & \scriptstyle U_x^{t^{N} \psi} \\
    \scriptstyle U_x^{t^2 \psi} & \scriptstyle 2U_x^{t \psi} & \cdots & \scriptstyle U_x^{t^{N+1} \psi} \\
    \scriptstyle \vdots & \scriptstyle \vdots & \ddots & \scriptstyle \vdots \\
    \scriptstyle U_x^{t^{N-1} \psi} & \scriptstyle (N-1)U_x^{t^{N-2} \psi} & \cdots & \scriptstyle U_x^{t^{2N-2} \psi}
\end{vmatrix}.
\end{align*}
Using these determinants, we define the following auxiliary quantities:
\begin{align*}    
\tilde{\omega}_{N}^{\psi}(a,b,\lambda) &:= \frac{D^N_1}{2 \pi D^N_0} = -\frac{1}{2 \pi a} \frac{D^N_{11} + D^N_{12}}{D^N_0}, \\
\tilde{q}_{N}^{\psi}(a,b,\lambda) &:= \frac{1}{2 \pi a}\frac{D^N_2}{D^N_0} = i \lambda -\frac{1}{2 \pi a^2} \frac{D^N_{21} + D^N_{22}}{D^N_0}.
\end{align*}
Taking the imaginary parts then yields the desired $N$th-order IF and chirprate estimators:
\begin{align}
\label{omega1}
\widehat{\omega}_{N}^{\psi}(a,b,\lambda) &= \operatorname{Im}\bigl(\tilde{\omega}_{N}^{\psi}(a,b,\lambda)\bigr) = -\frac{1}{2 \pi a} \operatorname{Im}\Big(\frac{D^N_{11} + D^N_{12}}{D^N_0}\Big), \\
\label{lambda1}
\widehat{q}_{N}^{\psi}(a,b,\lambda) &= \operatorname{Im}\bigl(\tilde{q}_{N}^{\psi}(a,b,\lambda)\bigr) = \lambda - \frac{1}{2 \pi a^2} \operatorname{Im}\Big(\frac{D^N_{21} + D^N_{22}}{D^N_0}\Big).
\end{align}

We now summarize the main result of this subsection in the following proposition.
\begin{pro}
Suppose $x(t)$ is a signal given by Eq.~\eqref{def_generalized_PPS}. Let $\widehat{\omega}_{N}^{\psi}(a,b,\lambda)$ and $\widehat{q}_{N}^{\psi}(a,b,\lambda)$ be the IF and chirprate reassignment operators defined by Eq.~\eqref{omega1} and Eq.~\eqref{lambda1}, respectively.
Then for any $a, b, \lambda$ with $D_0^N(a, b, \lambda) \neq 0$, the following holds:
\begin{align}
\hat{\omega}_N^\psi(a,b,\lambda) = \phi'(b),\quad \hat{q}_N^\psi(a,b,\lambda) = \phi''(b). \label{reassignment_operators_condition}
\end{align}
\end{pro}

\begin{proof}
The result follows directly from Cramer's rule applied to the linear system Eq.~\eqref{matrix2}. A detailed derivation is provided in the preceding derivations.
\end{proof}

By replacing the $j$-th column of the determinant $D^N_0$ with $\mathbf{y}^N$, we obtain the determinant $D^N_j$, which satisfies
\begin{align*}
\frac{d^{(j)}(b)+ i 2\pi \phi^{(j)}(b)}{(j-1)!}a^{j-1} = \frac{D^N_j}{D^N_0},\quad 3 \leq j\leq N.
\end{align*}
This allows us to estimate the high-order phase derivatives as
\[
\phi^{(j)}(b)=\frac{(j-1)!}{2 \pi a^{j-1}}  \operatorname{Im}\Big(\frac{D^N_j}{D^N_0}\Big).
\]
Note that the rightmost columns in Eq.~\eqref{matrix3} correspond to the second column of the determinant $D^N_j$; consequently, we only need to compute the determinant formed by the first and second columns on the right-hand side of Eq.~\eqref{matrix3}.

\begin{mrem}
Notice that replacing the wavelet \(\psi(t)\) in Eq.~\eqref{partial_b_WCT_polynomials} with \(t\psi(t)\) immediately yields Eq.~\eqref{partial_tb_WCT_polynomials}. The authors of \cite{li2022self} derived the second-order IF 
reassignment operator for their proposed self-matched extracting wavelet transform by replacing \(\psi(t)\) with \(t\psi(t)\) and \(t^2\psi(t)\), thereby forming a linear system of two equations. In fact, one can replace \(\psi(t)\) in Eq.~\eqref{partial_b_WCT_polynomials} with any alternative wavelet--for instance, \(\psi'(t)\)--to obtain a new equation. 
Specifically, replacing \(\psi(t)\) with \(N-1\) distinct wavelets, denoted \(\psi_1(t), \cdots, \psi_{N-1}(t)\), produces a system of N equations analogous to Eq.~\eqref{system_of_equations_2}. 
Under appropriate conditions, the components of \(\mathbf{x}^N\), which include the Nth-order IF and chirprate reassignment operators that satisfy Eq.~\eqref{reassignment_operators_condition}, can be solved from this linear system using Cramer's rule.

However, this approach raises concerns about the properties of the associated SWCT. On one hand, the SWCT is built upon the WCT \(U^\psi_x(a, b, \gl)\), which is entirely determined by the single wavelet \(\psi(t)\). On the other hand, the Nth-order IF and chirprate reassignment operators obtained through the aforementioned method depend on the auxiliary wavelets \(\psi_1(t), \cdots, \psi_{N-1}(t)\). The role of these auxiliary wavelets in shaping the properties of the resulting SWCT remains unclear, even in the specialized case where \(\psi_j(t) = t^j\psi(t)\) for \(j = 1, \cdots, N-1\).

Our discussion from Eq.~\eqref{partial_b_WCT_polynomials} to Eq.~\eqref{system_of_equations_2} demonstrates that the linear system Eq.~\eqref{system_of_equations_2} arises entirely from the WCT \(U^\psi_x(a, b, \lambda)\) defined by a single wavelet \(\psi(t)\). 
More precisely, Eq.~\eqref{partial_tb_WCT_polynomials}, which is equivalent to Eq.~\eqref{D_a_WCT_polynomials}, is a property of \(U^\psi_x(a, b, \lambda)\), not merely a variation of Eq.~\eqref{partial_b_WCT_polynomials} with \(\psi(t)\) replaced by a different wavelet \(\psi_1(t) = t\psi(t)\).
Similarly, Eq.~\eqref{partial_t2b_WCT_polynomials} (which corresponds to Eq.~\eqref{partial_lambdab_WCT_polynomials}) is another property of \(U^\psi_x(a, b, \lambda)\).
Analogously, all remaining equations  in the linear system Eq.~\eqref{system_of_equations_2} are properties of \(U^\psi_x(a, b, \lambda)\).
 \hfill  \scalebox{1}{$\square$}
\end{mrem}

\subsection{Special cases and reduction to CWT}

We now illustrate the proposed framework with the concrete cases \(N=2\) and \(N=3\), which lead to simplified expressions for the reassignment operators.

For \(N=2\), the new second-order IF and chirprate reassignment operators in Eq.~\eqref{omega1} and Eq.~\eqref{lambda1} simplify to
\begin{align}
\label{def_gO1}
\widehat{\omega}_{2}^{\psi} &:= -\frac{1}{2 \pi a} \operatorname{Im} \Big( \frac{U_x^{t^2\psi}U_x^{\psi'} - U_x^{t\psi}U_x^{\psi} - U_x^{t\psi}U_x^{t\psi'}}{U_x^{t^2\psi}U_x^{\psi} - U_x^{t\psi}U_x^{t\psi}} \Big), \\
\label{def_gL1}
\widehat{q}_{2}^{\psi} &:= \lambda + \frac{1}{2 \pi a^2} \operatorname{Im} \Big( \frac{U_x^{t\psi}U_x^{\psi'} - U_x^{\psi}U_x^{\psi} - U_x^{\psi}U_x^{t\psi'}}{U_x^{t^2\psi}U_x^{\psi} - U_x^{t\psi}U_x^{t\psi}} \Big).
\end{align}
Compared with Eq.~\eqref{def_gO} and Eq.~\eqref{def_gL}, the expressions Eq.~\eqref{def_gO1} and Eq.~\eqref{def_gL1} are much simpler and do not require \(U_x^{\psi''}(a, b, \lambda)\).

For \(N=3\), the third-order IF and chirprate reassignment operators are given by
\begin{align*}
\widehat{\omega}_{3}^{\psi}(a,b,\lambda) = -\frac{1}{2 \pi a} \operatorname{Im}\Big(\frac{D^3_{11} + D^3_{12}}{D^3_0}\Big), \;
\widehat{q}_{3}^{\psi}(a,b,\lambda) = \lambda - \frac{1}{2 \pi a^2} \operatorname{Im}\Big(\frac{D^3_{21} + D^3_{22}}{D^3_0}\Big),
\end{align*}
where
\begin{align*}
    & D^3_0:={U_x^{\psi}(U_x^{t^2 \psi}U_x^{t^4 \psi}-U_x^{t^3 \psi}U_x^{t^3 \psi})-U_x^{t \psi}(U_x^{t \psi}U_x^{t^4 \psi}-U_x^{t^2 \psi}U_x^{t^3 \psi})+U_x^{t^2\psi}(U_x^{t \psi}U_x^{t^3 \psi}-U_x^{t^2 \psi}U_x^{t^2 \psi})}; \\
     &D^3_{11}:={U_x^{\psi'}(U_x^{t^2 \psi}U_x^{t^4 \psi}-U_x^{t^3 \psi}U_x^{t^3 \psi})-U_x^{t \psi'}(U_x^{t \psi}U_x^{t^4 \psi}-U_x^{t^2 \psi}U_x^{t^3 \psi})+U_x^{t^2\psi'}(U_x^{t \psi}U_x^{t^3 \psi}-U_x^{t^2 \psi}U_x^{t^2 \psi})}; \\
     &D^3_{12}:={-U_x^{ \psi}(U_x^{t \psi}U_x^{t^4 \psi}-U_x^{t^2 \psi}U_x^{t^3 \psi})+2 U_x^{t\psi}(U_x^{t \psi}U_x^{t^3 \psi}-U_x^{t^2 \psi}U_x^{t^2 \psi})}; \\
    &D^3_{21} :=-{U_x^{\psi'}(U_x^{t \psi}U_x^{t^4 \psi}-U_x^{t^2 \psi}U_x^{t^3 \psi})+U_x^{t \psi'}(U_x^{ \psi}U_x^{t^4 \psi}-U_x^{t^2 \psi}U_x^{t^2 \psi})-U_x^{t^2\psi'}(U_x^{ \psi}U_x^{t^3 \psi}-U_x^{t \psi}U_x^{t^2 \psi})}; \\
    & D^3_{22} :={U_x^{\psi}(U_x^{ \psi}U_x^{t^4 \psi}-U_x^{t^2 \psi}U_x^{t^2 \psi})-2U_x^{t\psi}(U_x^{ \psi}U_x^{t^3 \psi}-U_x^{t^2 \psi}U_x^{t \psi})}. 
\end{align*}

In contrast to the third-order reassignment operators in \cite{chen2024composite, jiang2025synchrosqueezed}, our high-order IF and chirprate reassignment operators achieve both a simplified representation and reduced computational overhead, leading to significantly improved implementation efficiency.

Another special case is when \(\lambda = 0\). In this case, the WCT reduces to the continuous wavelet transform (CWT) of \(x(t)\) with wavelet \(\psi(t)\):
\begin{align*}
W_x^\psi(a,b) &:=\int_{\mathbb{R}} x(t)\frac 1a \psi\Big(\frac{t-b}{a}\Big)dt 
= \int_{\mathbb{R}} x(b+at) \psi(t) dt. 
\end{align*}
Let \(E^N_0(a, b)\), \(E^N_{11}(a, b)\), and \(E^N_{12}(a, b)\) be the quantities obtained from 
\(D^N_0(a, b, \lambda)\), \(D^N_{11}(a, b, \lambda)\), and \(D^N_{12}(a, b, \lambda)\) in Eq.~\eqref{D^N_0} and Eq.~\eqref{def_D11} 
by replacing \(U_x^{t^m\psi}(a, b, \lambda)\) and \(U_x^{t^m\psi'}(a, b, \lambda)\) with 
\(W_x^{t^m\psi}(a, b)\) and \(W_x^{t^m\psi'}(a, b)\), respectively. Define
\begin{align}
\label{omega1_2D}
\widehat{\Omega}_{N}^{\psi}(a,b): = -\frac{1}{2 \pi a}\operatorname{Im}\Big(\frac{E^N_{11}(a, b) + E^N_{12}(a, b)}{E^N_0(a, b)}\Big).  
\end{align}
Then \(\widehat{\Omega}_{N}^{\psi}(a,b)\) is an \(N\)th-order IF reassignment operator in the time-scale plane satisfying
\[
\widehat{\Omega}_{N}^{\psi}(a,b) = \phi'(b)
\]
for any \(x(t)\) of the form Eq.~\eqref{def_generalized_PPS} and any \((a, b)\) with \(E^N_0(a, b) \neq 0\).

Although the authors in \cite{hu2019high} proposed a method to obtain the \(N\)th-order IF reassignment operator based on the CWT, their derivation is quite complicated, heavily reliant on iteration, and computationally expensive.
In contrast, the reassignment operator in Eq.~\eqref{omega1_2D} is more straightforward and intuitive.

\subsection{High-order synchrosqueezed wavelet-chirplet transform}

The high-order IF and chirprate reassignment operators derived above enable a sharpened representation of multicomponent signals. By reassigning the WCT coefficients to the estimated IF and chirprate curves, we define the high-order synchrosqueezed wavelet-chirplet transform (HSWCT) as follows.
\begin{mdef}
    \label{definition_HSWCT}
    Given \( x(t) \in \mathcal{S}'(\mathbb{R}) \) and \( \psi(t) \in \mathcal{S}(\mathbb{R}) \) with a threshold \(\epsilon > 0\), let \(\widehat{\omega}_{N}^{\psi}(a,b,\lambda)\) and \(\widehat{q}_{N}^{\psi}(a,b,\lambda)\) be 
    the reassignment operations defined by Eq.~\eqref{omega1} and Eq.~\eqref{lambda1}, respectively. Then, we define the HSWCT  as
    \begin{equation*}
        \mathcal{U}^{N,\psi}_x(\xi, b, \gamma) := \iint\limits_{\{(a, \lambda) : (a,b,\lambda) \in D_\epsilon\}} 
        U_x^\psi(a,b,\lambda) \, \delta\!\left(\xi - \widehat{\omega}_{N}^{\psi}(a,b,\lambda)\right) 
        \delta\!\left(\gamma - \widehat{q}_{N}^{\psi}(a,b,\lambda)\right) \frac{da}{a} \, d\lambda,
    \end{equation*}
    where \(D_\epsilon\) is a region defined in Eq.~\eqref{Deps}.   
\end{mdef}
In Definition \ref{definition_HSWCT}, \(\delta\) denotes the Dirac delta function.

Building upon the HSWCT framework, we further define the multiple HSWCT (MHSWCT). Since the HSWCT provides a TFC representation of the signal, we need to employ the scale-frequency relationship
\[
\xi = \frac{\mu_0}{a},
\]
where $\mu_0 := -\arg\max_{\eta \in \mathbb{R}} |\hat{\psi}(\eta)|$ and we assume that $\psi(t)$ is such a wavelet that $\mu_0 > 0$. 
Under the same conditions as in Definition~\ref{definition_HSWCT}, define the initial transform as $\mathcal{U}^{N,\psi}_{x,1}(\xi,b,\gamma) := \mathcal{U}^{N,\psi}_x(\xi,b,\gamma)$. The MHSWCT is then defined recursively for $j = 2, 3, \dots$ by
\begin{equation*}
\mathcal{U}^{N,\psi}_{x, j}(\xi,b,\gamma) := \iint\limits_{\{(\eta, \lambda): (\mu_0/\eta, b, \lambda) \in D_\epsilon\}}
\mathcal{U}^{N,\psi}_{x,j-1}(\eta, b, \lambda) 
\delta\!\big(\xi - \widehat{\omega}_{N}^{\psi}(\mu_0/\eta, b, \lambda)\big)
\delta\!\big(\gamma - \widehat{q}_{N}^{\psi}(\mu_0/\eta, b, \lambda)\big)
\,\frac{d\eta}{\eta} \, d\lambda.
\end{equation*}

\section{Theoretical error estimates}
In this section, we establish the error analysis for the approximation of the proposed high-order reassignment operators to IFs and chirprates. 
To commence this analysis, we need to impose some specific conditions on the  multicomponent signals given in Eq.~\eqref{AHM0}.

\begin{mdef}
    \label{define_set}
    Let $\epsilon_1, \epsilon_2$ be small positive constants, $N$ be a positive integer, and $x(t)$ be a multicomponent signal given in \eqref{AHM0}. The set $\mathcal{A}^{N}_{\epsilon_1, \epsilon_2}$ consists of multicomponent signals such that for all $t \in \mathbb{R}$, each mode $x_k(t)$ \((k = 1, 2, \cdots, K)\) satisfies:
    \begin{itemize}
        \item The amplitude function $A_k(t) \in L^{\infty}(\mathbb{R})$ is nonnegative; 
        \item The phase function $\phi_k(t) \in C^{N+1}(\mathbb{R})$ with $ \phi'_k(t) > 0$;
        \item The growth constraints: 
              $\lvert A'_k(t) \rvert \leq \epsilon_1$ and 
              $\lvert \phi_k^{(N+1)}(t) \rvert \le \epsilon_2$.
    \end{itemize}
\end{mdef}


In this paper, we also  assume the signal components \(x_k\) and \(x_l\) \((1 \leq k, l \leq K, k \neq l)\) satisfy
\begin{equation}
    \label{condition_separation}
    \text{either } \frac{|\phi'_k(b) - \phi'_l(b)|}{\phi'_k(b) + \phi'_l(b)} \geq \frac{\Delta_1}{\mu_0}, \, \text{ or } |\phi''_k(b) - \phi''_l(b)| \geq 2\Delta_2, \, b \in \mathbb{R},
\end{equation}
where \(0<\Delta_1 <\mu_0\) and \(\Delta_2 > 0\). There may exist time instants where $\phi'_k(b)$ and $\phi'_l(b)$ cross.

For \(1 \leq k \leq K\), define
\begin{equation} \label{def_zk}
Z_k := \left\{(a, b, \lambda) : |\mu_0 - a\phi'_k(b)| < \Delta_1 \text{ and } |\lambda - \phi''_k(b)| < \Delta_2, \, b \in \mathbb{R}\right\}.
\end{equation}
If the multicomponent signal \(x(t) \in \mathcal{A}^{N}_{\epsilon_1, \epsilon_2}\) satisfies the separation condition Eq.~\eqref{condition_separation}, then these sets are pairwise disjoint:
\begin{align}
Z_k \cap Z_l = \emptyset, \quad k \neq l. \label{separation_condition}
\end{align}

For signals satisfying the above conditions, the coefficient $|U^{t^m \psi}_{x_k}(a,b,\lambda)|$ exhibits significant concentration within the region $Z_k$, while its magnitude remains minor in $Z_l$ for $l \neq k$. 
This distinctive behavior arises from the localized nature of the WCT and the separation of signal components in the TSC domain. 
Consequently, this property not only enables the effective discrimination of individual signal components by analyzing the WCT coefficients but also provides a solid foundation for the subsequent error analysis of the reassignment operators.

\subsection{Lemmas}

Under the assumption that \(x(t) \in \mathcal{A}^{N}_{\epsilon_1,\epsilon_2}\) with sufficiently small \(\epsilon_1,\epsilon_2 > 0\),  then each component \(x_k(b+at)\) can be locally approximated as a 
polynomial phase signal within the time frame \(at\):
\begin{equation*}
x_k(b+at) = A_k(b+at) e^{i 2 \pi \phi_k(b+at)} = x_{N,k}(a,b,t) + x_{r,k}(a,b,t),
\end{equation*}
where
\begin{equation*}
x_{N, k}(a,b,t) = A_k(b) e^{i 2 \pi \phi_k(b)} e^{\sum_{j=1}^{N} \frac{i 2\pi}{j!} \phi^{(j)}_k(b)(at)^j}
\end{equation*}
and
\begin{equation*}
x_{r, k}(a,b,t) = x_k(b+at) - x_{N,k}(a,b,t).
\end{equation*}
Here, \(x_{N,k}(a,b,t)\) represents the main term, while \(x_{r,k}(a,b,t)\) denotes the remainder term of \(x_k(b+at)\).

To facilitate the subsequent discussion, for \(m=0,1,2,\cdots\), we introduce the following notations:
\begin{align}
    \label{notation1}  I_l :&=\int_{\mathbb{R}}|t^l \psi(t)|dt;\\  
\label{notation2} \mathcal{R}^{t^m\psi}_{x_k}(a,b,\lambda) :&= \int_\RR  x_{N,k}(a,b,t)  t^m \psi(t) e^{-i\pi \lambda a^2 t^2}dt. 
\end{align} 
Notice that 
\begin{equation}
\label{PFT1}
 \mathcal{R}^{t^m\psi}_{x_k}(a,b,\lambda)=x_k(b) {\cF^{N}}(t^m\psi)\big(-a\phi_k'(b), a^2(\gl-\phi_k''(b)), -a^3\phi_k'''(b), \cdots, -a^{N}\phi_k^{N}(b)\big),  
\end{equation}
where $\cF^{N}(t^m\psi)$ is the $N$th-order polynomial Fourier transform of $t^m\psi(t)$ defined by \eqref{def_PFT}. 

The following lemma provides a bound for this approximation error in terms of the WCT.

\begin{lem}\label{lem1}
Let \(x(t) \in \mathcal{A}^{N}_{\epsilon_1,\epsilon_2}\) for some \(\epsilon_1>0\) and \(\epsilon_2>0\) and \(U^{t^m\psi}_x(a, b, \lambda)\) denote the WCT of \(x(t)\) with wavelet function \(t^m \psi(t)\), where $m$ is a nonnegative integer. 
For each $k$,  let $\mathcal{R}^{t^m\psi}_{x_k}(a,b,\lambda)$ be the quantity defined by Eq.~\eqref{notation2}. Then 
\begin{equation}
\label{approx_Uxk}
 |U^{t^m \psi}_{x_k} (a,b,\lambda)- \mathcal{R}^{t^m\psi}_{x_k}(a,b,\lambda)|\le \Pi_{k, m}(a,b), 
\end{equation}
where 
\begin{equation}
\label{def_Pimk} 
\Pi_{k, m}(a, b):= \epsilon_1  a I_{m+1}+ \epsilon_2  A_k(b)\frac{2 \pi a^{N+1}}{(N+1)!}I_{N+m+1}.
\end{equation}
\end{lem} 

\begin{proof}
The proof is provided in Appendix~\ref{app:proof_lem1}.
\end{proof}

We next examine how the derivative \(\partial_b U_{x_k}^{t^m\psi}\) relates to the 
higher-order moments \(U_{x_k}^{t^{m+j-1}\psi}\). The following lemma provides a 
quantitative bound for this relationship.

\begin{lem} \label{lem2}
Assume all conditions remain the same as in Lemma \ref{lem1}. Then 
\begin{equation}
\label{bounded_by_Ckn} 
\Big|\partial_b U_{x_k}^{t^m \psi}(a,b,\lambda)-\sum_{j=1}^N \frac {i2\pi \phi_k^{(j)}(b)}{(j-1)!}a^{j-1}
U_{x_k}^{t^{m+j-1} \psi}(a,b,\lambda)\Big| \le C_{k, m}(a,b),     
\end{equation}
where 
\begin{equation}
\label{def_Ckn}
C_{k, m}(a,b):=\epsilon_1 I_m +\epsilon_2\frac{ 2\pi a^N}{N!}\big( \epsilon_1 a I_{m+N+1} +A_k(b) I_{m+N}\big). 
\end{equation}
\end{lem} 
\begin{proof}
The proof is provided in Appendix~\ref{app:proof_lem2}.
\end{proof}


For \( m = 0,\cdots,2N-1 \), let 
\(\Upsilon_{k, m}(b)\) be a function satisfying the following condition:
\begin{align}
  \label{Upsilonmk}
  \sup_{\{(a, \gl): (a,b,\lambda)\notin Z_k\}}  
\Big|{\cF^{N}}(t^m\psi)\Big(-a\phi_k'(b), a^2(\gl-\phi_k''(b)), -a^3\phi_k'''(b), \cdots, -a^N\phi_k^{(N)}(b)\Big)\Big|
&\leq \Upsilon_{k, m}(b).
\end{align}
Consider any distinct indices \(k \neq l\) and \((a,b,\lambda) \in Z_l\). Then, from Eq.~\eqref{approx_Uxk} and Eq.~\eqref{Upsilonmk} we obtain
\begin{align}
\big| U_{x_k}^{t^m \psi}(a,b,\lambda) \big| 
&\le \big| U_{x_k}^{t^m \psi}(a,b,\lambda) - \mathcal{R}_{x_k}^{t^m\psi}(a,b,\lambda) \big| 
+ \big| \mathcal{R}_{x_k}^{t^m\psi}(a,b,\lambda) \big| \nonumber \\
&\le \Pi_{k,m}(a,b) + A_k(b) \Upsilon_{k,m}(b). \label{eq:bound_Uxk_proof}
\end{align}

Subsequently, we present another pivotal lemma that quantifies the error term 
\begin{equation}
\label{Resnk}  
  \mathrm{Res}_{l, m}(a, b, \gl):=\partial_b U_x^{t^m \psi}(a,b,\lambda) - \sum_{j=1}^{N} \frac{i2\pi \phi_l^{(j)}(b)}{(j-1)!}a^{j-1} U_x^{t^{m+j-1}\psi}(a,b,\lambda),
\end{equation}
which  will play a fundamental role in the proof of the main theorem.

\begin{lem}\label{residual_estimate}
Let \(x(t) \in \mathcal{A}^{N}_{\epsilon_1,\epsilon_2}\) for some \(\epsilon_1>0\) and \(\epsilon_2>0\), 
with \(x(t)\) satisfying Eq.~\eqref{condition_separation}. 
Define \(\mathrm{Res}_{l, m}(a, b, \lambda)\), \(\Pi_{k, m}(a,b)\), \( C_{k,m}(a,b)\), 
and \(\Upsilon_{k, m}(b)\) by Eq.~\eqref{Resnk}, Eq.~\eqref{def_Pimk}, Eq.~\eqref{def_Ckn}, 
and Eq.~\eqref{Upsilonmk}, respectively.
Then for any \((a,b,\lambda) \in Z_l\),
\begin{equation*}
\left| \mathrm{Res}_{l, m}(a, b, \lambda) \right| \leq \Lambda_{l, m}(a, b),
\end{equation*}
where
    \begin{align}
\label{def_Lam_mk}
  \Lambda_{l, m}(a, b):=\sum_{k=1}^K C_{k, m}(a,b)+ \sum_{k\not=l}\sum_{j=1}^{N} \frac{2 \pi a^{j-1}}{(j-1)!} \Big|\phi_k^{(j)}(b)-\phi_l^{(j)}(b)\big | \big( \Pi_{k, m+j-1}(a,b) + A_k(b)  \Upsilon_{k, m+j-1}(b) \big).
    \end{align}
\end{lem}

\begin{proof}
The proof is provided in Appendix~\ref{app:proof_lem3}.
\end{proof}

This lemma provides a uniform bound for the residual \(\mathrm{Res}_{l,m}\) on the region \(Z_l\). 
The bound \(\Lambda_{l,m}\) comprises terms arising from both the interference from other 
components (via Lemma \ref{lem2} and the separation condition) and the intrinsic approximation 
errors of each component (via Lemma \ref{lem1}). This estimate will be directly used in the 
proof of Theorem \ref{error_reassignment} to control the errors of the high-order IF and 
chirprate reassignment operators.

\subsection{Main results}\label{error_reassignment}
We now present the main theorems establishing error bounds between the reassignment operators $\widehat{\omega}^{\psi}_{N}$ and $\widehat{q}^{\psi}_{N}$ (defined in Eq.~\eqref{omega1} and Eq.~\eqref{lambda1}, respectively) and the true IF and chirprate of a signal component.

\begin{theo}\label{theorem_eta_lambda}
Suppose \(x(t)\) is of the form Eq.~\eqref{AHM0} with \(x(t) \in \mathcal{A}^N_{\epsilon_1, \epsilon_2}\) and satisfies the separation condition Eq.~\eqref{condition_separation}. 
Let \(\epsilon_0 > 0\) be a threshold and \(D^N_0\) be the determinant of \(\mathbf{D}^N\) defined by Eq.~\eqref{matrix2}.
For each \(l\) with \(1 \le l \le K\), and for any \((a, b, \lambda) \in Z_l\) (see Eq.~\eqref{def_zk}) such that \(|D^N_0(a, b, \lambda)| > \epsilon_0^{-1}\), the following holds:
\begin{align*}
|\widehat{\omega}_{N}^{\psi}(a,b,\lambda) - \phi_l'(b)| \leq Bd_{1,l},\;
|\widehat{q}_{N}^{\psi}(a,b,\lambda) - \phi_l''(b)| \leq Bd_{2,l},
\end{align*}
where
\begin{align}
\label{def_Bd12}
Bd_{1,l} = \frac{\epsilon_0}{2\pi} \sum_{n=1}^{N} \Lambda_{l, n-1} \left|M_{n,1}\right|, \quad
Bd_{2,l} = \frac{\epsilon_0}{2\pi a} \sum_{n=1}^{N} \Lambda_{l, n-1} \left|M_{n,2}\right|,
\end{align}
with \(\Lambda_{l, n}\) defined by Eq.~\eqref{def_Lam_mk} and \(M_{n,j}\) denoting the algebraic cofactor of \(D^N_0\) for the element in the \(n\)-th row and \(j\)-th column.
\end{theo}

\begin{proof}   
 Since \(\widehat{\omega}_{N}^{\psi}(a,b,\lambda)\) and \(\widehat{q}_{N}^{\psi}(a,b,\lambda)\) are the imaginary  parts of \(\tilde{\omega}_{N}^{\psi}(a,b,\lambda)\) and \(\tilde{q}_{N}^{\psi}(a,b,\lambda)\) respectively, we have 
$$
|\wh {\omega}_{N}^{\psi}(a,b,\lambda) - \phi_l'(b)|\le 
|\tilde{\omega}_{N}^{\psi}(a,b,\lambda) - i\phi_l'(b)|, \quad |\wh{q}_{N}^{\psi}(a,b,\lambda) - \phi_l''(b)|
\le |\tilde{q}_{N}^{\psi}(a,b,\lambda) - i\phi_l''(b)|. 
$$
Thus we only need to demonstrate that
\[
|\tilde{\omega}_{N}^{\psi}(a,b,\lambda) - i\phi_l'(b)| \leq  Bd_{1,l}, \quad |\tilde{q}_{N}^{\psi}(a,b,\lambda) - i\phi_l''(b)| \leq   Bd_{2,l}. 
\]

Recall $D^N_0,  {D_1^N}$ are the determinants of matrices defined by \eqref{DN_1_2}. Clearly, we have   
 \begin{align*}
 {D_1^N}-   i 2 \pi \phi'_l(b) D^N_0 &=
    \begin{vmatrix}
        \partial_b U_x^{ \psi}  - i 2 \pi \phi'_l(b) U_x^{\psi} & U_x^{t \psi} & \cdots & U_x^{t^{N-1} \psi} \\
        \partial_b U_x^{t \psi} - i 2 \pi \phi'_l(b) U_x^{t \psi}  & U_x^{t^2 \psi}  & \cdots & U_x^{t^N \psi} \\
        \partial_b U_x^{t^2 \psi} - i 2 \pi \phi'_l(b) U_x^{t^2 \psi}  & U_x^{t^3 \psi} & \cdots & U_x^{t^{N+1} \psi} \\
        \vdots & \vdots & \ddots & \vdots \\
        \partial_b U_{x}^{t^{N-1} \psi} - i 2 \pi \phi'_l(b) U_x^{t^{N-1} \psi}  & U_x^{t^N \psi} & \cdots & U_x^{t^{2N-2} \psi} \\
    \end{vmatrix} \\
    &= 
    \begin{vmatrix}
  \partial_b U_x^{\psi} -       \sum_{j=1}^{N}\frac{i 2\pi \phi_l^{(j)}(b) }{(j-1)!} a^{j-1} U_x^{t^{j-1} \psi} & U_x^{t \psi} & \cdots & U_x^{t^{N-1} \psi} \\
 \partial_b U_x^{t\psi} -        \sum_{j=1}^{N}\frac{i 2\pi \phi_l^{(j)}(b) }{(j-1)!} a^{j-1}U_x^{t^{j} \psi}  & U_x^{t^2 \psi}  & \cdots & U_x^{t^N \psi} \\
  \partial_b U_x^{t^2 \psi} -       \sum_{j=1}^{N}\frac{i 2\pi \phi_l^{(j)}(b) }{(j-1)!} a^{j-1}U_x^{t^{j+1} \psi} & U_x^{t^3 \psi} & \cdots & U_x^{t^{N+1} \psi} \\
        \vdots & \vdots & \ddots & \vdots \\
 \partial_b U_x^{t^{N-1} \psi} -        \sum_{j=1}^{N}\frac{i 2\pi \phi_l^{(j)}(b) }{(j-1)!} a^{j-1}U_x^{t^{j+N-2} \psi} & U_x^{t^N \psi} & \cdots & U_x^{t^{2N-2} \psi} \\
    \end{vmatrix} \\ 
    &= 
    \begin{vmatrix}
         \mathrm{Res}_{l, 0} & U_x^{t \psi} & \cdots & U_x^{t^{N-1} \psi} \\
         \mathrm{Res}_{l, 1} & U_x^{t^2 \psi}  & \cdots & U_x^{t^N \psi} \\
         \mathrm{Res}_{l, 2} & U_x^{t^3 \psi} & \cdots & U_x^{t^{N+1} \psi} \\
        \vdots & \vdots & \ddots & \vdots \\
         \mathrm{Res}_{l, N-1} & U_x^{t^N \psi} & \cdots & U_x^{t^{2N-2} \psi} \\
    \end{vmatrix}=\sum_{m=1}^{N} \mathrm{Res}_{l, m-1} \; M_{m,1}.  
\end{align*}
By applying similar reasoning, we can also obtain: 
 \begin{align*}
{D_2^N} - i 2 \pi a \phi''_l(t) D^N_0 = \sum_{m=1}^{N} \mathrm{Res}_{l, m-1} \; M_{m,2}.  
\end{align*}

Thus,
\begin{align*}
 &    |\tilde{\omega }_{N}^{\psi}(a,b,\lambda) - i\phi_l'(b)| = \Big| \frac{D_1^N - i 2 \pi \phi'_l(b) D^N_0}{2 \pi D^N_0}\Big| \\
&\qquad \leq \frac{\epsilon_0}{2\pi} \sum_{m=1}^{N} \left| \mathrm{Res}_{l, m-1}\right| \; \left|M_{m,1}\right|
\le  \frac{\epsilon_0}{2\pi} \sum_{m=1}^{N} \Lambda_{l, m-1}\; \left|M_{m,1}\right|, \\
&  |\tilde{q }_{N}^{\psi}(a,b,\lambda) - i\phi_l''(b)| =    \Big| \frac{D_2^N - i 2 \pi \phi'_l(b) D^N_0}{2 \pi a D^N_0}\Big|\\
&\qquad\leq \frac{\epsilon_0}{2\pi a}
\; \sum_{m=1}^{N}\left| \mathrm{Res}_{l, m-1}\right|\; \left|M_{m,2}\right| \leq \frac{\epsilon_0}{2\pi a}
\; \sum_{m=1}^{N} \Lambda_{l, m-1} \; \left|M_{m,2}\right|.  
\end{align*}
\end{proof}

Note that the error bounds \(Bd_{1, l}\) and \(Bd_{2, l}\) in Eq.~\eqref{def_Bd12} for the approximation of the IF and chirprate reassignment operators depend on \(a\) and  \(b\). However, for \((a, b, \lambda) \in Z_l\), the scale variable \(a\) is confined to the interval
\begin{align*}
\frac{\mu_0 - \Delta_1}{\phi'_l(b)} < a < \frac{\mu_0 + \Delta_1}{\phi'_l(b)},
\end{align*}
which is determined entirely by \(b\). Moreover, by Eq.~\eqref{eq:bound_Uxk_proof}, for any \((a,b,\lambda) \in Z_l\), we have
\begin{align*}
|U^{t^m \psi}_{x} (a,b,\lambda)| &\le  |U^{t^m \psi}_{x_l} (a,b,\lambda)|+\sum_{k\neq l} |U^{t^m \psi}_{x_k} (a,b,\lambda)| \le  G_{l, m}(b), 
\end{align*}
where 
\[
G_{l, m}(b):= A_l(b) I_m + \sum_{k \neq l} \big(\Pi_{k, m}\big( \tfrac{\mu_0 + \Delta_1}{\phi'_l(b)}, b\big) + A_k(b) \Upsilon_{k, m}(b)\big). 
\]

Since \(|M_{m+1,1}|\) and \(|M_{m+1,2}|\) are finite sums of products of \(|U^{t^m \psi}_{x} (a,b,\lambda)|\), they are bounded by the corresponding sums of products of \(G_{l, m}(b)\). 
In addition, the quantity \(\Lambda_{l, m-1}(a, b)\) appearing in \(Bd_{1, l}\) and \(Bd_{2, l}\) is bounded by \(\Lambda_{l, m-1}\big(\frac{\mu_0 + \Delta_1}{\phi'_l(b)}, b\big)\). 
Consequently, for a given \(N\), one can derive bounds for \(Bd_{1, l}\) and \(Bd_{2, l}\) that depend only on \(b\). For general \(N\), we omit the explicit expressions due to the complexity of the determinant of an \((N-1) \times (N-1)\) matrix.

Theorem \ref{error_reassignment} establishes error bounds for the high-order IF and chirprate reassignment operators introduced in Section 2.3. On the region \(Z_l\), where the \(l\)-th component dominates, the estimates \(\hat{\omega}_N^\psi\) and \(\hat{q}_N^\psi\) approximate the true IF \(\phi_l'\) and chirprate \(\phi_l''\). 
The approximation errors are controlled by \(\epsilon_1\) (amplitude variation), \(\epsilon_2\) (higher-order phase variation), and the separation between components.
In each region \(Z_l\), a single component dominates the signal, whereas outside \(Z_l\), the component energies are relatively small. This can cause the denominator to approach zero and introduce significant singularities, so the error analysis may no longer hold in these regions.

\subsection{ Special case: IF approximation error in high-order synchrosqueezed wavelet transforms}
When \(\lambda = 0\), this result reduces to the first rigorous error analysis for 
high-order synchrosqueezed wavelet transforms, filling a theoretical gap in the literature.
We conclude this section by examining the approximation of the IF via the $N$th-order SST reassignment operator \(\widehat{\Omega}_{N}^{\psi}(a,b)\) in the wavelet domain.
For this purpose, we require that the signal components \(x_k\) and \(x_l\) \((1 \leq k, l \leq K, k \neq l)\) satisfy the following separation condition:
\begin{equation}
    \label{condition_separation_CWT}
  \frac{|\phi'_k(b) - \phi'_l(b)|}{\phi'_k(b) + \phi'_l(b)} 
\geq \frac{\Delta_1}{\mu_0},
\end{equation}
with \(0<\Delta_1<\mu_0\). This condition ensures that the components are sufficiently separated in the time-frequency plane, allowing us to define the region
\begin{equation*}
 O_k := \left\{(a, b) : |\mu_0 - a\phi'_k(b)| < \Delta_1, \; b \in \mathbb{R}\right\},
\end{equation*}
where the reassignment operator provides a reliable estimate of the IF.
Let \(\widetilde{\Upsilon}_{k, m}(b)\) be a function satisfying the following condition:
\begin{align*}
  \sup_{\{a: (a,b,\lambda)\notin O_k\}} 
\Big|{\cF^N}(t^m\psi)\Big(-a\phi_k'(b), -a^2\phi_k''(b)), -a^3\phi_k'''(b), \cdots, -a^N\phi_k^{(N)}(b)\Big)\Big|
&\leq \wt \Upsilon_{k, m}(b).
\end{align*}
Then for \(x(t) \in \mathcal{A}^N_{\epsilon_1, \epsilon_2}\) of the form Eq.~\eqref{AHM0} satisfying the separation condition Eq.~\eqref{condition_separation_CWT},
for any \((a, b)\) satisfying \((a, b) \in O_l\) and \(|E^N_0(a, b)| > \epsilon_0^{-1}\), we have 
\begin{align}
\label{HSST_approx}
 | \widehat{\Omega}_{N}^{\psi}(a,b)-\phi'(b)|\leq \widetilde{Bd}_{1, l},
\end{align}
where 
\begin{align*}
  \widetilde{Bd}_{1,l} = \frac{\epsilon_0}{2\pi} \sum_{n=1}^{N}\widetilde{\Lambda}_{l, n-1} \; \left|\widetilde{M}_{n,1}\right|,
\end{align*}
and \(\widetilde{\Lambda}_{l, n}\) is defined by Eq.~\eqref{def_Lam_mk} with \(\Upsilon_{k, n}(b)\) replaced by \(\widetilde{\Upsilon}_{k, n}(b)\); 
\(\widetilde{M}_{n,1}\) denotes the algebraic cofactor of \(E^N_0\) for the element in the \(n\)th row and first column.

 Eq. \eqref{HSST_approx} provides an approximation error between an arbitrary-order IF reassignment operator and the instantaneous frequency. To the best of our knowledge, no theoretical analysis exists in the literature concerning the approximation of arbitrary-order SST IF reassignment operators to the IF. As a by-product of this work, our derived theorem establishes such an analysis, thereby filling a gap in the theoretical framework of high-order SSTs.


\section{  Methodological Framework}
To implement the HSWCT numerically, this section details the key components of the computational framework. After a brief introduction to the Morlet-type wavelet, we present the numerical implementation, parameter selection strategy, and mode retrieval algorithm in the following subsections.

Throughout this work, we employ the Morlet-type wavelet
\begin{equation}
\psi_\sigma(t) = \frac{1}{\sigma \sqrt{2\pi}} e^{-\frac{t^2}{2\sigma^2}} e^{-i2\pi t}. \label{Morlet_type_wavelet}
\end{equation}
The Gaussian envelope ensures optimal time-frequency concentration, and the complex exponential makes the wavelet analytic--a desirable property for analyzing non-stationary signals. 
The parameter \(\sigma\) controls the width of the Gaussian envelope, balancing time and frequency resolution: smaller \(\sigma\) improves time resolution, while larger \(\sigma\) enhances frequency resolution.

For this wavelet, we have the simple relation
\begin{align*}
\psi'_\sigma(t) = -\frac{1}{\sigma^2} t \psi_\sigma(t) - 2\pi i \psi_\sigma(t),
\end{align*}
which, when substituted into Eq.~\eqref{omega1} and Eq.~\eqref{lambda1}, yields
\begin{align}
\hat{\omega}_{N}^{\psi}(a,b,\lambda) = \frac{1}{a} - \frac{1}{2 \pi a}\operatorname{Im}\left(\frac{D^N_{12}}{D^N_0}\right), \quad \hat{q}_{N}^{\psi}(a,b,\lambda) = \lambda - \frac{1}{2 \pi a^2}\operatorname{Im}\left(\frac{D^N_{22}}{D^N_0}\right),\label{Morlet_type_reassignment}
\end{align}
thus avoiding explicit numerical differentiation and the computation of the determinants of \(D^N_{11}\) and \(D^N_{21}\).
The Morlet-type wavelet therefore  provides a convenient and effective foundation for implementing the HSWCT.

\subsection{Numerical implementation}
For efficient computation, we reformulate the WCT of a signal \(x(t)\) into a structure amenable to FFT-based implementation.
 Using the Morlet-type wavelet \(\psi_\sigma\) defined in Eq.~\eqref{Morlet_type_wavelet}, we start from the definition Eq.~\eqref{def_WCT} and obtain
\begin{align*}
U_x^{t^m\psi_\sigma}(a, b, \lambda) &= \int_{\mathbb{R}} x(b + at) \, t^m \psi_\sigma(t) e^{- i\pi\lambda a^2t^2} \, dt  \nonumber\\
&= \int_{\mathbb{R}} \frac{1}{a} e^{-\frac{i2\pi\eta b}{a}} \hat{x}(\eta/a) \, \mathcal{F}^2(t^m \psi_\sigma)( - \eta, a^2\lambda) \, d\eta \nonumber \\
&= \int_{\mathbb{R}} \big( \hat{x}(\eta) \, \mathcal{F}^2(t^m\psi_\sigma)(-a\eta, a^2\lambda) \big) \, e^{i2\pi b\eta} \, d\eta.
\end{align*}

This formulation leads to an efficient three-step computational pipeline: first, compute the Fourier transform \(\hat{x}(\eta)\) of the signal via an FFT; second, for each scale \(a\) and chirprate \(\lambda\), multiply \(\hat{x}(\eta)\) by the precomputed function \(\mathcal{F}^2(t^m\psi_\sigma)(-a\eta, a^2\lambda)\); and third, apply an  IFFT to obtain the WCT coefficients.

The function \(\mathcal{F}^2(\psi_\sigma)(\eta_1, \eta_2)\) can be derived explicitly using the Gaussian integral formula \cite{cohen1995time}. For real \(\alpha\) and \(\beta\) with \(\alpha > 0\),
\begin{align*}
\int_{-\infty}^{\infty} e^{-(\alpha + i\beta)t^2 + i\omega t} \, dt = \frac{\sqrt{\pi}}{\sqrt{\alpha + i\beta}} e^{-\frac{\omega^2}{4(\alpha + i\beta)}}, \label{gaussian_integral}
\end{align*}
where \(\sqrt{\alpha + i\beta}\) denotes the principal square root. Applying this to \(\psi_\sigma\) gives
\begin{equation}
\mathcal{F}^2(\psi_\sigma)(\eta_1, \eta_2) =\frac{1}{\sqrt{1+i2\pi\sigma^2 \eta_2}} e^{-\frac{2\pi^2\sigma^2 (1+\eta_1)^2}{1+i2\pi \sigma^2\eta_2}}. \label{C_psi_explicit}
\end{equation}
For \(m \geq 1\), the functions \(\mathcal{F}^2(t^m\psi_\sigma)(\eta_1, \eta_2)\) can be obtained via straightforward integration techniques, such as differentiation with respect to parameters or integration by parts.

Let \(L_0\) be the signal length, \(J_0\) the number of scales, and \(L_c\) the number of chirprate samples.
This inverse FFT-based approach offers significant computational advantages. Specifically, for a fixed moment order \(m\), computing \(U_x^{t^m \psi_\sigma}\) requires \(J_0 \cdot L_c\) IFFTs, each of complexity \(O(L_0 \log L_0)\). Thus, the computational cost per moment order is \(O(J_0 \cdot L_c \cdot L_0 \log L_0)\).
For the \(N\)th-order HSWCT, the simplified reassignment formulae Eq.~\eqref{Morlet_type_reassignment} require computing \(U_x^{t^m\psi}\) for all moments \(m = 0, 1, \ldots, 2N-2\). This involves \(2N-1\) moment orders, i.e., \(2N-1\) times the number required for the WCT. Consequently, the total number of FFT/IFFT operations increases by this factor, but the asymptotic complexity per moment order remains \(O(J_0 \cdot L_c \cdot L_0 \log L_0)\).
Once the arrays \(U_x^{t^m\psi}\) are obtained, the reassignment operators are constructed via simple element-wise operations (additions and multiplications) on these precomputed arrays of size \(J_0 \times L_c \times L_0\). Since these operations are linear in the number of elements and require no additional FFTs, they do not affect the overall asymptotic complexity.

For a detailed discussion of the WCT implementation and the numerical aspects of the reassignment algorithms, we refer the readers to \cite{jiang2025synchrosqueezed}. A MATLAB implementation of the HSWCT is publicly available at \text{https://github.com/Lishuixin065/HSWCT} to facilitate reproducibility and further research.

\subsection{Parameter selection strategy}
To choose the parameter \(\sigma\) for the wavelet \(\psi_\sigma\), we employ R\'enyi entropy, which measures the concentration of 3D TSC representations. In the WCT framework, the R\'enyi entropy is given by
\begin{equation}
E_{\sigma}: =
 \frac{1}{{1 - \ell }}\log _2 
  \frac{\iint_{\mathbb{R}^2}\int_0^\infty \left|U^{\psi_\sigma}_x(a, b, \lambda) \right|^{2\ell}\frac {da}a db d\lambda} {\left(\iint_{\mathbb{R}^2}\int_0^\infty 
  \left| U^{\psi_\sigma}_x(a, b, \lambda) \right|^2\frac {da}a db d\lambda \right)^{\ell}}, \label{renyi_entropy}
\end{equation}
where \(\ell > 2\) is typically used. A lower R\'enyi entropy indicates a more concentrated TSC representation. Therefore, we select \(\sigma\) by minimizing the R\'enyi entropy:
\begin{align}
  \label{find_sigma}  \sigma_{1} = \arg\min_{\sigma>0} E_\sigma.
\end{align}
In this paper, we set \(\ell = 2.2\) for the R\'enyi entropy calculation.

This  strategy provides a reasonable range of suitable values and consistently yields an appropriate \(\sigma\) in many practical applications, leading to sharper and more concentrated TSC representations. 
Such an approach is widely adopted in the literature \cite{stankovic2001measure,baraniuk2002measuring,thakur2013synchrosqueezing}.

Another practical consideration is the selection of the optimal order \(N\) in HSWCT. In practice, the choice of \(N\) is often guided by empirical knowledge of the signal characteristics. 
While a higher order \(N\) can theoretically provide more accurate estimates of the IF and chirprate, it also increases computational complexity and sensitivity to noise, which may degrade performance.
A practical criterion for selecting the optimal order \(N\) is based on the R\'enyi entropy of the TFC representation. 
Define the R\'enyi entropy of order \(\ell\) (with \(\ell > 2\)) for the HSWCT in the TFC space as
\begin{equation*}
E_{\ell}(N): = \frac{1}{{1 - l}}\log _2 \frac{\iiint_{\mathbb{R}^3} \big |\mathcal{U}^{N,\psi_\sigma}_x(\xi, b, \gamma) \big|^{2 \ell}\, d\xi\, db\, d\gamma} {\Big(\iiint_{\mathbb{R}^3} \big| \mathcal{U}^{N,\psi_\sigma}_x(\xi, b, \gamma) \big|^2\, d\xi\, db\, d\gamma \Big)^{\ell}},
\end{equation*}
where \(\mathcal{U}^{N,\psi_\sigma}_x(\xi, b, \gamma)\) denotes the HSWCT coefficient of signal \(x\) in the TFC space. In this work, we set \(\ell = 2.2\) as in the previous subsection. 
The optimal order $N$ is then selected as the smallest value for which the decrement $\Delta E_{\ell}(N) := E_{\ell}(N) - E_{\ell}(N+1)$ drops below a given threshold, indicating that further increasing the order provides negligible gains in energy concentration. 
This entropy-based criterion ensures that the chosen order strikes an optimal balance between TFC representation concentration and computational complexity.

Additionally, cross-validation or other data-driven methods may assist in validating the selected order, balancing estimation accuracy with computational efficiency. 
However, it should be noted that this criterion does not constitute a fixed rule; rather, the optimal choice of \(N\) is application-dependent and should be fine-tuned according to the specific signal characteristics and analysis objectives.

\subsection{Mode retrieval }
We now turn our attention to mode retrieval. In the 3D TFC space, signal separation operator (SSO) reconstruction algorithms based on the CT and  WCT have been developed in \cite{chui2021time,chui2023analysis}. 
Suppose \((\check \xi_k(b), \check \gamma_k(b))\) are the ridge points extracted from the TFC representation \(\mathcal{U}^{N,\psi}_x(\xi,b,\gamma)\) or \(\mathcal{U}^{N,\psi}_{x, j}(\xi,b,\gamma)\). These ridge curves provide estimates of the IF \(\phi'_k(t)\) and chirprate \(\phi''_k(t)\).
We refer to \cite{chen2023disentangling,zhang2022local} for details on 3D ridge detection algorithms.

When a mode is well approximated by a linear chirp over local time intervals, i.e., for \(|at|\) sufficiently small,
\begin{align} \label{assume_chirp}
x_k(b + a t) \approx A_k(b) e^{i2\pi \big( \phi_k(b) + \phi'_k(b)(at) + \frac{\phi''_k(b)}{2}(at)^2 \big)},
\end{align}
the WCT of \(x_k\) then reduces to
\begin{align*}
U^{\psi_\sigma}_{x_k}(a, b, \lambda) &\approx \int_{\mathbb{R}} A_k(b) e^{i2\pi \phi_k(b)} e^{i2\pi \big( \phi'_k(b) a t + \frac{1}{2}\phi''_k(b) (at)^2 \big)} \psi(t) e^{-i\pi \lambda a^2 t^2} \, dt \\
&= x_k(b) \; \mathcal{F}^2(\psi_\sigma)\left( -a \phi'_k(b), \, a^2(\lambda - \phi''_k(b)) \right).
\end{align*}

Within the region \(Z_k\), the dominant contribution comes from the target component \(x_k(b)\), while interferences from other components are relatively small. 
Under this condition, one may approximate the mode directly via
\begin{align}\label{direct_SSO}
    x_k(b)\approx \frac{U^{\psi_\sigma}_{x}\big(1/\check{\xi}_k(b),\, b,\, \check{\gamma}_k(b)\big)}{\mathcal{F}^2(\psi_\sigma)\big(-1,\, 0\big)}.
\end{align}
However, as shown in Eq.~\eqref{C_psi_explicit}, the magnitude of the WCT coefficient \(U^{\psi_\sigma}_{x_k}(a, b, \lambda)\) is approximately
\[
|U^{\psi_\sigma}_{x_k}(a, b, \lambda)| \approx \frac{A_k(b)}{\big(1 + 4\pi^2 \sigma^4 a^4(\lambda - \phi''_k(b))^2\big)^{1/4}} 
e^{-\frac{2\pi^2 \sigma^2 (1-a \phi'_k(b))^2}{1 + 4\pi^2 \sigma^4 a^4(\lambda - \phi''_k(b))^2}},
\]
which shows that the decay along the chirprate direction is not exponential, but rather follows a power-law decay of the form \((1 + 4\pi^2 \sigma^4 a^4(\lambda - \phi''_k(b))^2)^{-1/4}\). 
This means that even when the region \(Z_k\) (where the target component dominates) is well-defined, the WCT coefficient may still be significantly influenced by interferences from other components, especially when the chirprates of those components are close to that of the target component.
Consequently, even when \(\check{\gamma}_k(b)\) is close to the true chirprate, the approximation error in Eq.~\eqref{direct_SSO} may become significant in practice.

To address this issue, a group WCT-based SSO scheme has been proposed in \cite{li2022chirplet,jiang2025synchrosqueezed}, following a similar idea to that in \cite{li2022chirplet}. This approach effectively mitigates inter-mode interference and yields more accurate mode reconstructions. The modes \(x_k(t), k=1, \cdots, K\), can then be recovered by solving the linear system
\begin{align}
    \label{solution_SSO}
	 \begin{bmatrix}
			x_1(b) \\
	    x_2(b) \\
	    \vdots \\
	    x_K(b)
	 \end{bmatrix}
	 &\approx 
	 \begin{bmatrix}
		 e_{1,1} & e_{1,2} & \cdots & e_{1,K} \\
		 e_{2,1} & e_{2,2} & \cdots & e_{2,K} \\
		 \vdots & \vdots & \ddots & \vdots \\
		 e_{K,1} & e_{K,2} & \cdots & e_{K,K}
	 \end{bmatrix}^{-1} 
	 \begin{bmatrix}
		U_x^{\psi_\sigma }(1/\check \xi_1(b),b,\check\gamma_1(b)) \\
	  U_x^{\psi_\sigma }(1/\check \xi_2(b),b,\check\gamma_2(b)) \\
	\vdots \\
	 U_x^{\psi_\sigma}(1/\check \xi_K(b),b,\check\gamma_K(b))
	 \end{bmatrix},
 \end{align}
where the entries of the mixing matrix are given by
\[
e_{l, k} = \mathcal{F}^2(\psi_\sigma)\big(-{\check \xi_k(b)}/{\check \xi_l(b)},\; (\check \gamma_l(b)-\check \gamma_k(b))/{\check \xi_l(b)}^2\big).
\]
The matrix \(E := [e_{l, k}]_{1 \leq l, k\leq K}\) is assumed to be nonsingular; otherwise, \(E^{-1}\) denotes its pseudo-inverse.

For the HSWCT, the choice of $\sigma$ involves a trade-off. On one hand, a smaller $\sigma$ helps preserve the validity of the local linear chirp assumption in Eq.~\eqref{assume_chirp}, especially when the phase of signal components varies rapidly. 
Moreover, a smaller $\sigma$  helps further reduce the recovery errors \cite{li2026time}.
On the other hand, a larger $\sigma$ ensures sufficient decay along the chirprate dimension, which is crucial for the accuracy of the reassignment estimates.

In light of these considerations, this paper adopts a two-parameter strategy. We first employ \(\sigma_1\) obtained from Eq.~\eqref{find_sigma} in the HSWCT to obtain accurate IF and chirprate estimates. 
Subsequently, for mode reconstruction, we use a smaller parameter \(\sigma_2\) in the recovery formula Eq.~\eqref{solution_SSO}. In this work, we empirically set \(\sigma_2\) to a value smaller than \(\sigma_1\). 
Based on our experiments, choosing \(\sigma_2 = \frac{1}{3}\sigma_1\) works well in practice, as it better satisfies the local linear chirp assumption Eq.~\eqref{assume_chirp} and yields more accurate mode reconstructions.

\section{Experimental results}
In this section, we  validate the proposed HSWCT and its associated mode retrieval scheme. Our primary goal is to demonstrate that HSWCT can produce sharply concentrated TFC representations, even for signals with strongly modulated IFs. 
Building on this, we then assess the effectiveness of the high-order reassignment operators in accurately extracting individual signal modes. 
Specifically, by systematically comparing HSWCTs of different orders, we highlight the superior resolution of higher-order transforms and quantify their accuracy.
\begin{figure}[H]
    \centering
    \begin{tabular}{cc}
        \subfloat[IFs of signal $x(t)$]{%
            \includegraphics[width=0.28\textwidth]{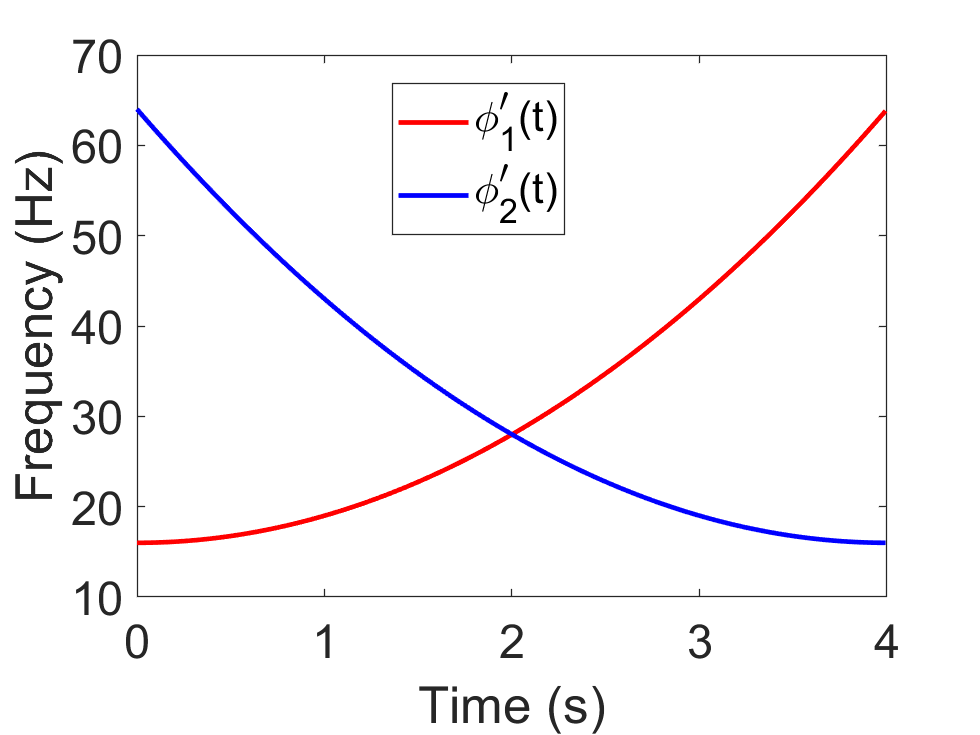}%
        } &
        \subfloat[Chirprates of signal $x(t)$]{%
            \includegraphics[width=0.28\textwidth]{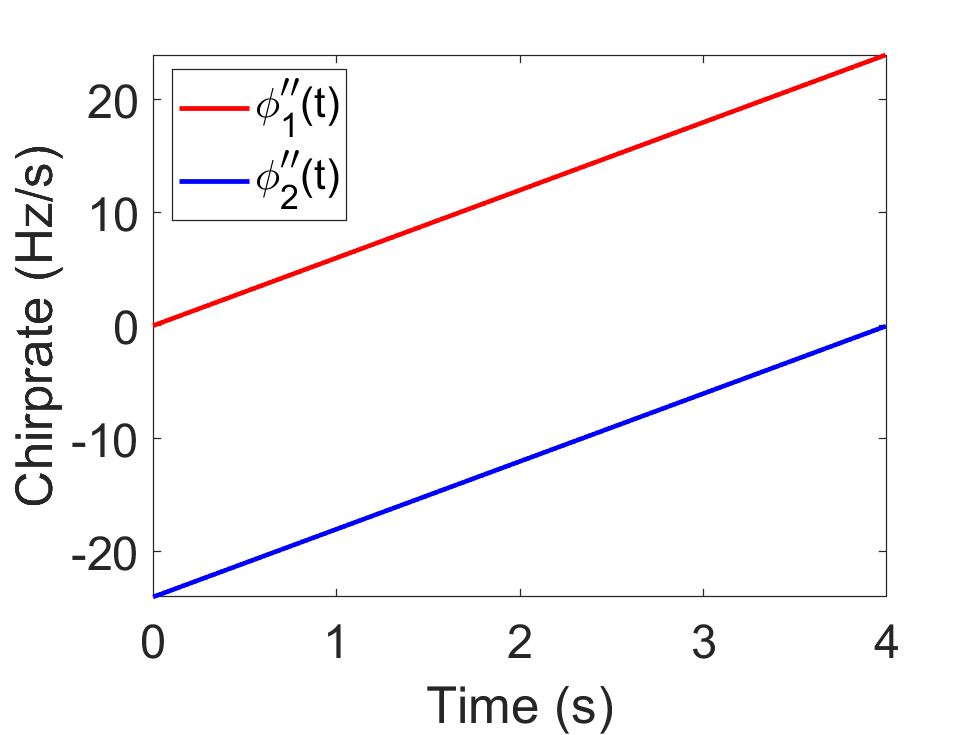}%
        } \\
    \end{tabular}
    \caption{\small IFs and chirprates of signal $x(t)$}
    \label{fig:IFs_CRs_cubic}
\end{figure}

Let \(x(t)\) be a signal with overlapping IFs, consisting of two cubic-phase signals:
\begin{align*}
x(t)=x_1(t) + x_2(t), \; 
    x_1(t) = e^{-0.01t^3 + 0.02t}  e^{2 \pi i (t^3 +16t) },\; 
    x_2(t) = e^{-0.02t^3 + 0.05t^2-0.1} e^{2\pi i ((t-4)^3 + 16t)}, 
\end{align*}
where \(0\le t\le 4\). 
The signal \(x(t)\) is sampled at a rate of \(\Delta t = 1/128\). 
The optimal parameter \(\sigma_1 = 4.4\) is obtained from Eq.~\eqref{find_sigma}.
The IFs and chirprates of \(x_1(t)\) and \(x_2(t)\) are shown in Fig.~\ref{fig:IFs_CRs_cubic}. Panel (a) shows the intersection of the IF curves at \(t = 2\) s and \(\omega = 28\) Hz, while Panel (b) shows that the corresponding chirprates are well separated.

Fig.~\ref{figure:time-frequency_slices_at_different_chirprates_of_cubic_signal} shows the time-frequency slices of the second-order HSWCT, \( \mathcal{U}_x^{2,\psi_{\sigma_1}}(\xi,b,\gamma) \), and the third-order HSWCT, \( \mathcal{U}_x^{3,\psi_{\sigma_1}}(\xi,b,\gamma) \), at different chirprate values \(\gamma = 0, -12, 12\). 
It can be observed that the third-order HSWCT provides a more concentrated 3D representation of \(x(t)\) than its second-order counterpart.
\begin{figure}[H]
    \centering
    \setlength{\tabcolsep}{3pt} 
    \begin{tabular}{ccc} 
        \begin{subfigure}[t]{0.28\textwidth} 
            \centering
            \includegraphics[width=\linewidth]{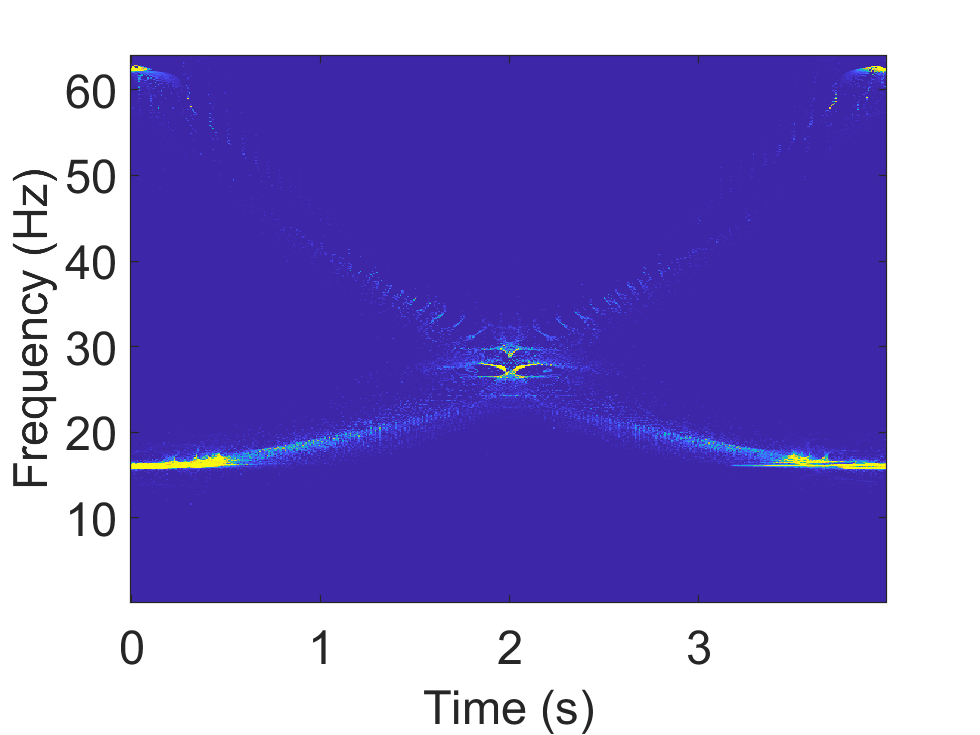}
            \caption{}
        \end{subfigure} &
        \begin{subfigure}[t]{0.28\textwidth}
            \centering
            \includegraphics[width=\linewidth]{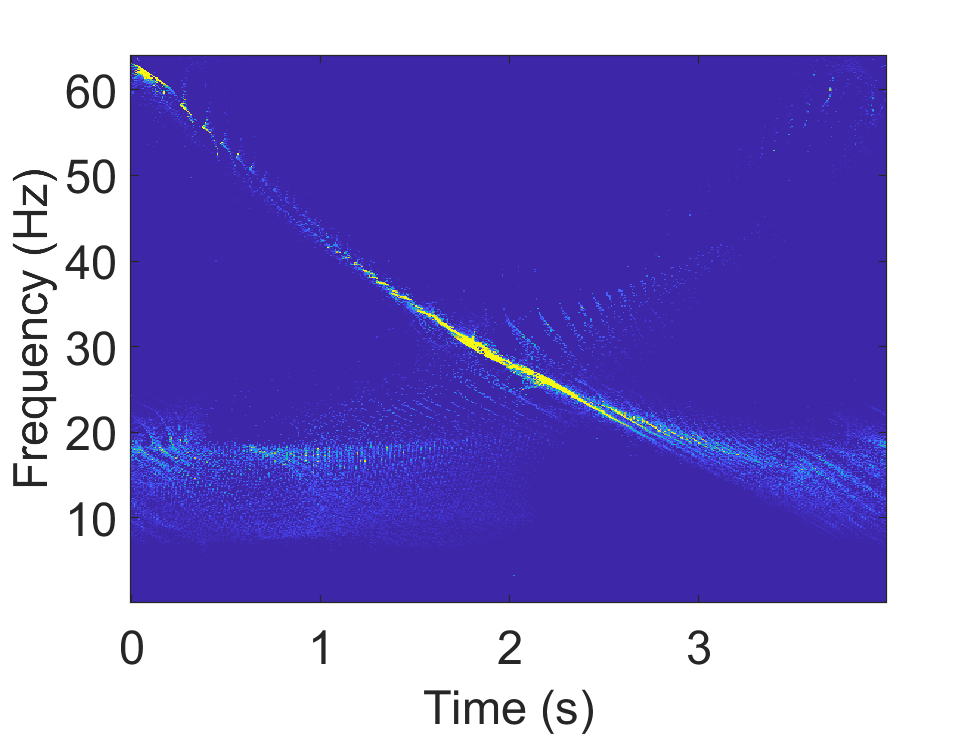}
            \caption{}
        \end{subfigure} &
        \begin{subfigure}[t]{0.28\textwidth}
            \centering
            \includegraphics[width=\linewidth]{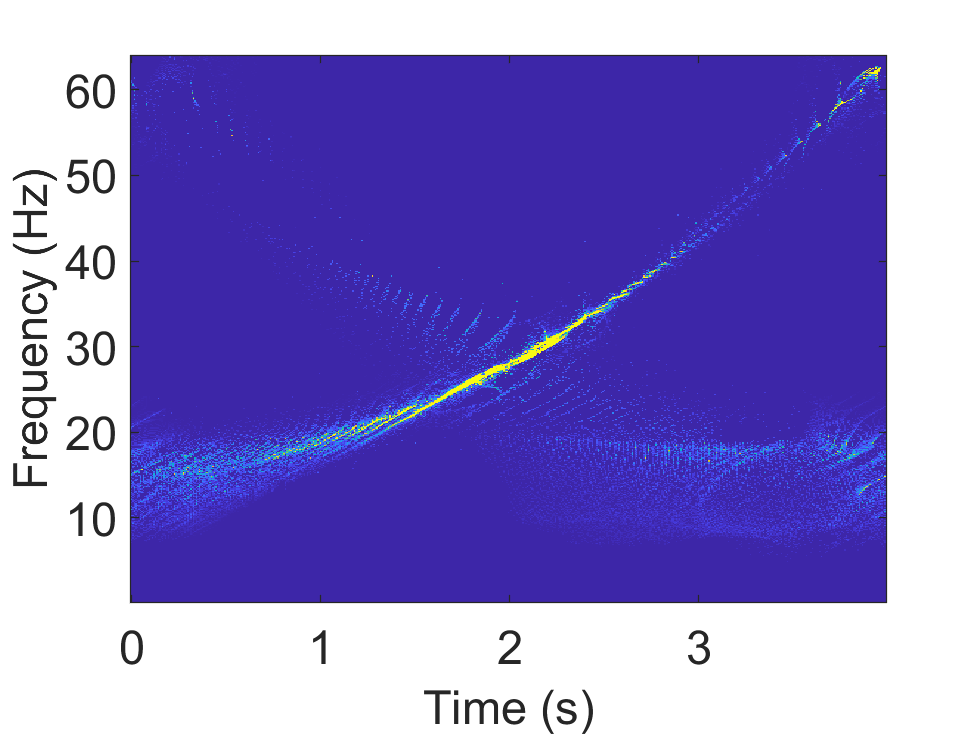}
            \caption{}
        \end{subfigure}\\
        \begin{subfigure}[t]{0.28\textwidth}
            \centering
            \includegraphics[width=\linewidth]{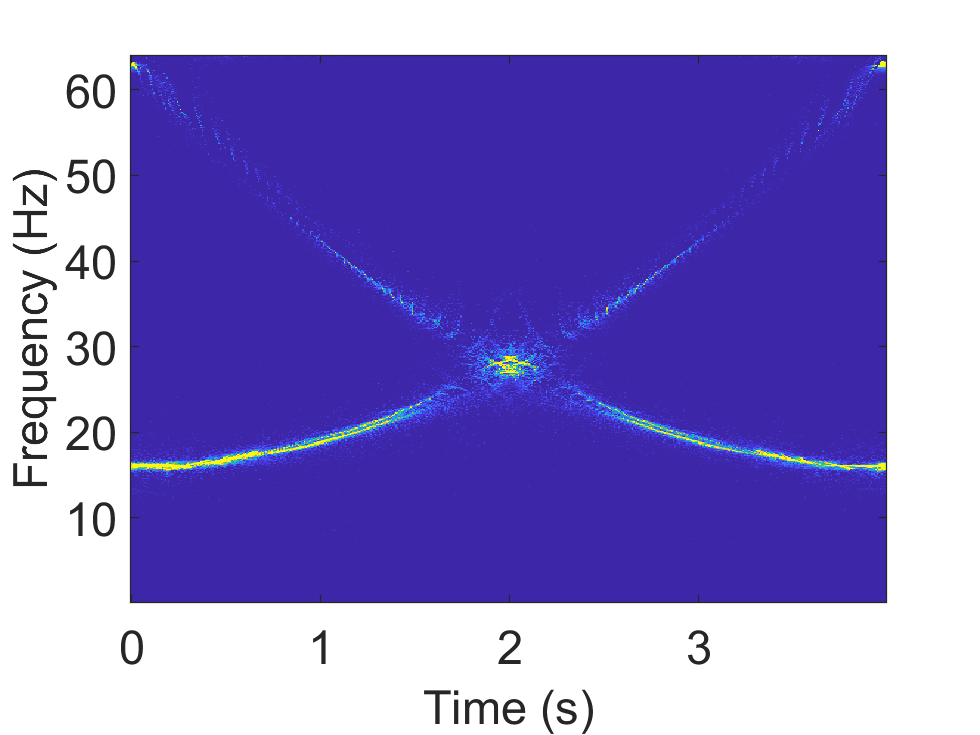}
            \caption{}
        \end{subfigure} &
        \begin{subfigure}[t]{0.28\textwidth}
            \centering
            \includegraphics[width=\linewidth]{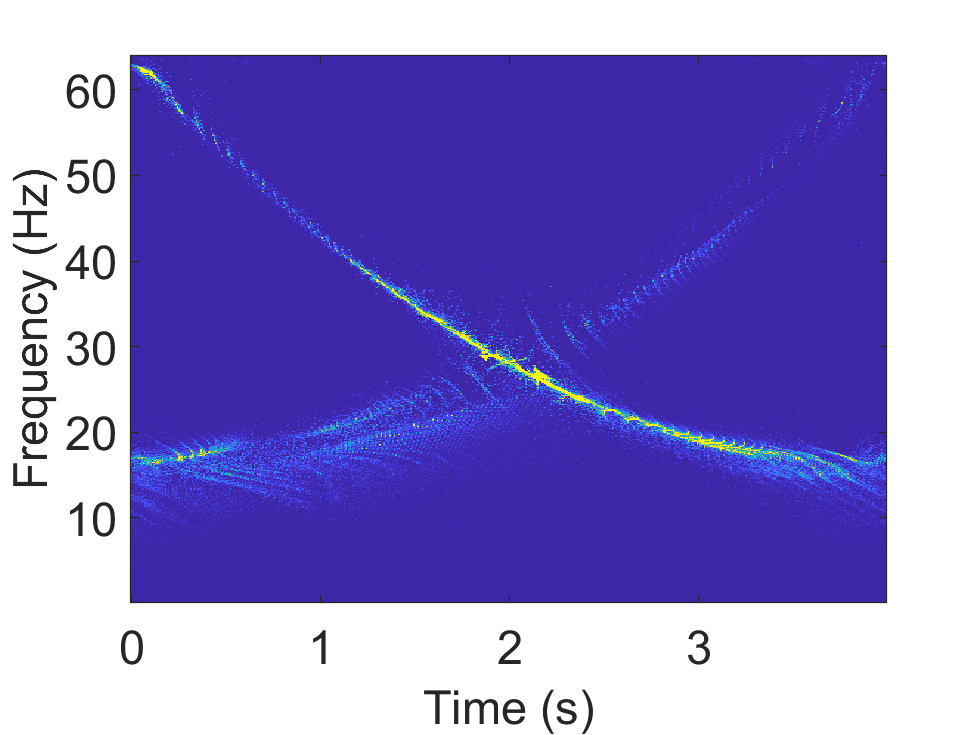}
            \caption{}
        \end{subfigure} &
        \begin{subfigure}[t]{0.28\textwidth}
            \centering
            \includegraphics[width=\linewidth]{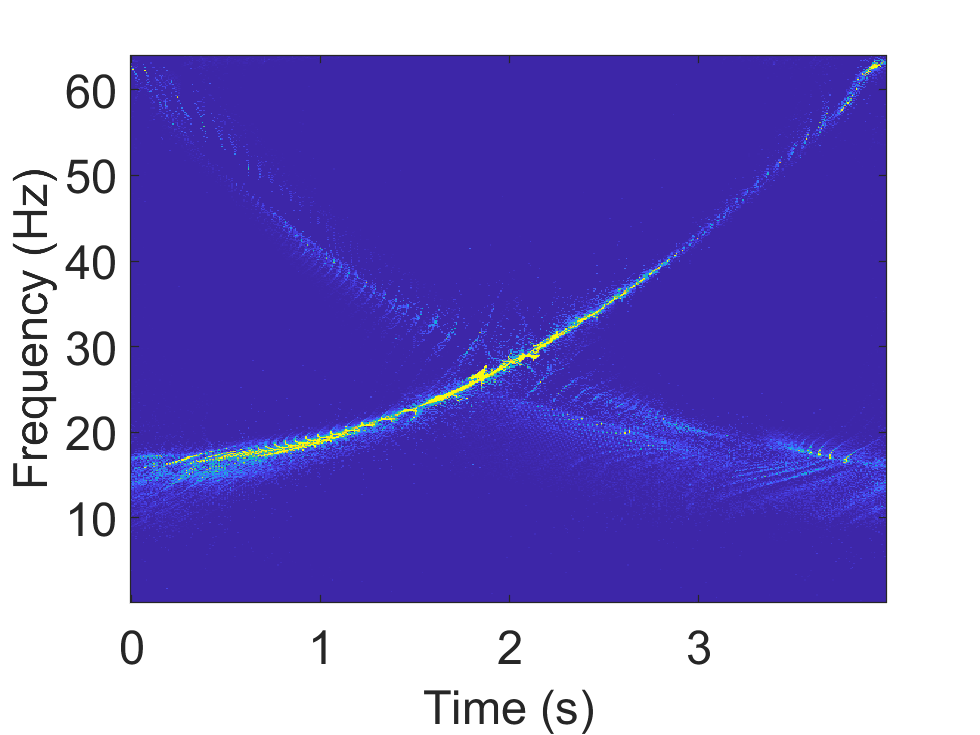}
            \caption{}
        \end{subfigure}
    \end{tabular}
    \caption{Time-frequency slices at different chirprate values. 
        (a)-(c) Second-order HSWCT \( \mathcal{U}_x^{2,\psi_{\sigma_1}}(\xi,b,\gamma) \) with \(\gamma=0, -12, 12\), respectively; 
        (d)-(f) Third-order HSWCT \( \mathcal{U}_x^{3,\psi_{\sigma_1}}(\xi,b,\gamma) \) with \(\gamma=0, -12, 12\), respectively.}
    \label{figure:time-frequency_slices_at_different_chirprates_of_cubic_signal}
\end{figure}

Fig.~\ref{figure:IFs_and_CRs_estimation_of_cubic_signal} displays the IF and chirprate ridges extracted from the second-order HSWCT \(\mathcal{U}_x^{2,\psi_{\sigma_1}}(\xi,b,\gamma)\) and the third-order HSWCT \(\mathcal{U}_x^{3,\psi_{\sigma_1}}(\xi,b,\gamma)\). 
Both methods provide accurate IF estimates, as shown in panels (a)-(b) and (e)-(f). However, their performance in chirprate estimation differs substantially: the third-order successfully captures the chirprate ridges in panels (g)-(h) (except near the boundaries), while the second-order suffers from significant estimation errors in panels (c)-(d).
\begin{figure}[H]
    \centering
    \setlength{\tabcolsep}{2pt} 
    \begin{tabular}{cccc} 
        \begin{subfigure}[t]{0.22\textwidth}
            \centering
            \includegraphics[width=\linewidth]{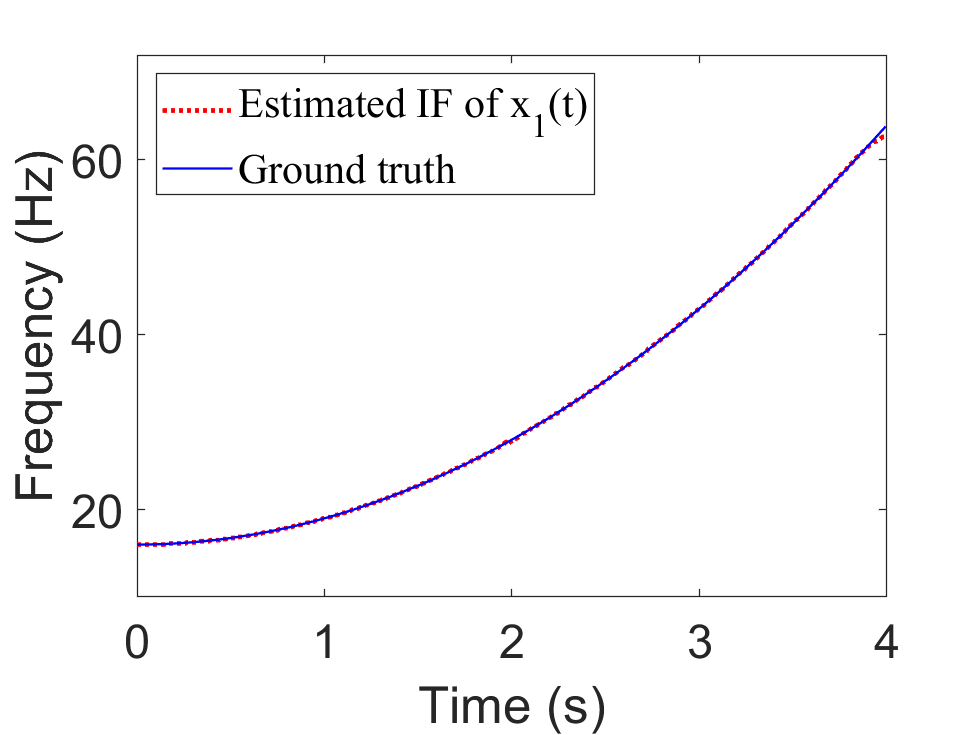}
            \caption{}  
        \end{subfigure} &
        \begin{subfigure}[t]{0.22\textwidth}
            \centering
            \includegraphics[width=\linewidth]{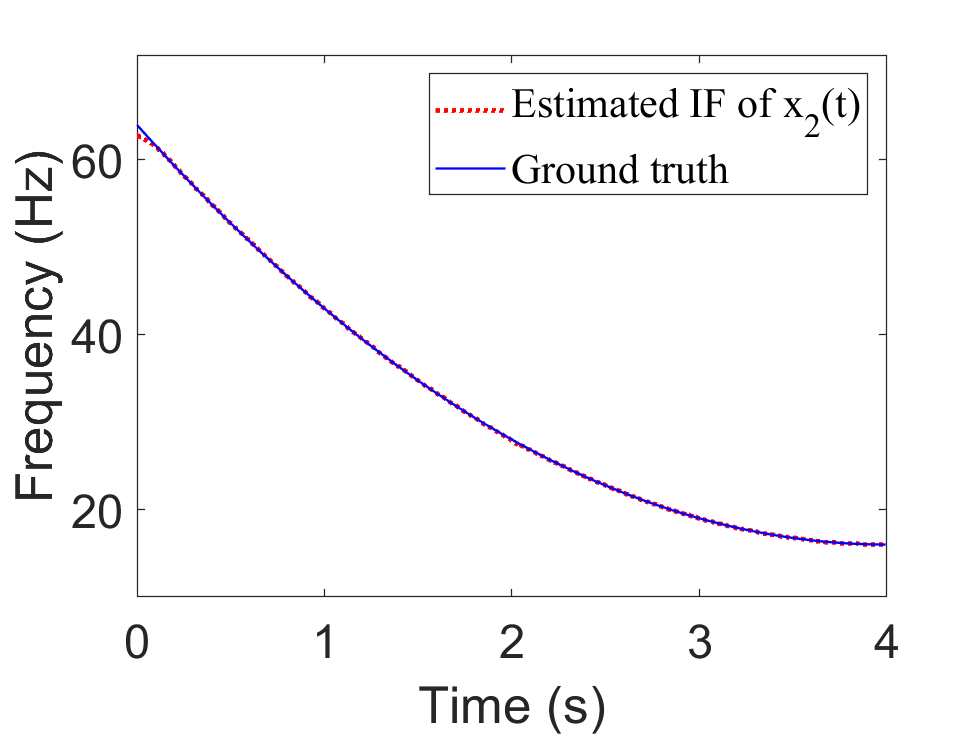}
            \caption{}
        \end{subfigure} &
        \begin{subfigure}[t]{0.22\textwidth}
            \centering
            \includegraphics[width=\linewidth]{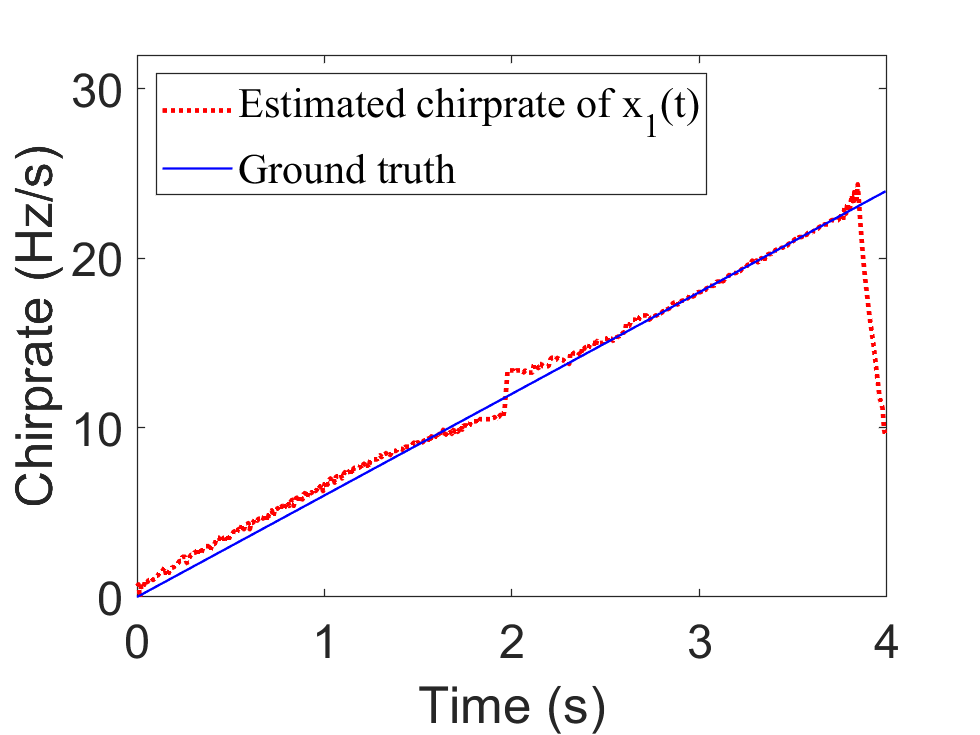}
            \caption{}
        \end{subfigure} &
        \begin{subfigure}[t]{0.22\textwidth}
            \centering
            \includegraphics[width=\linewidth]{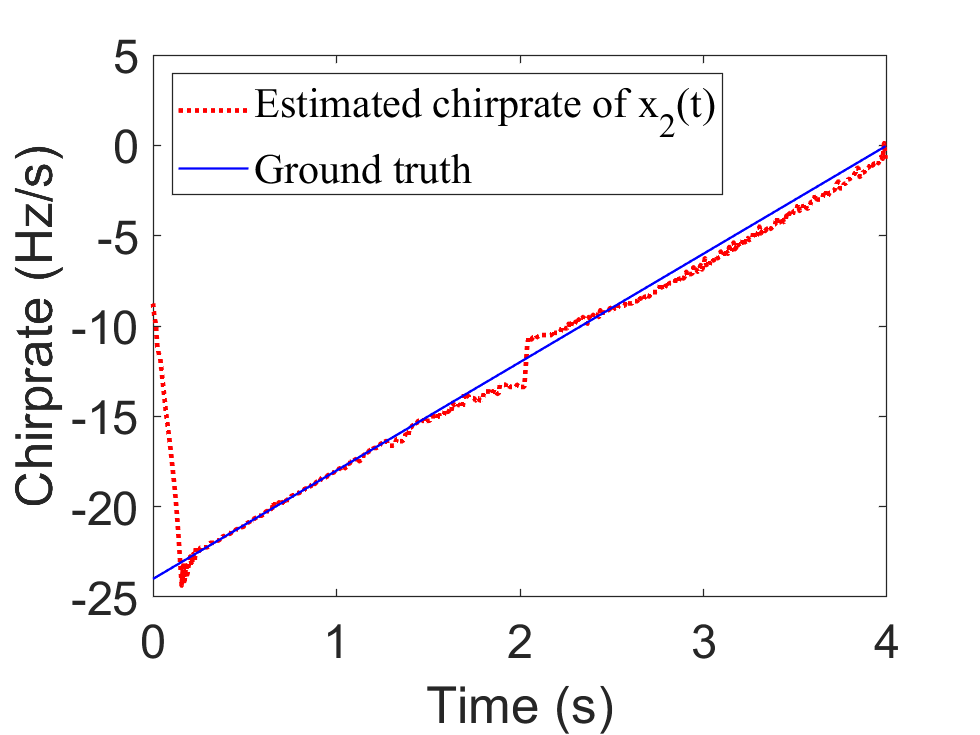}
            \caption{}
        \end{subfigure} \\
        \begin{subfigure}[t]{0.22\textwidth}
            \centering
            \includegraphics[width=\linewidth]{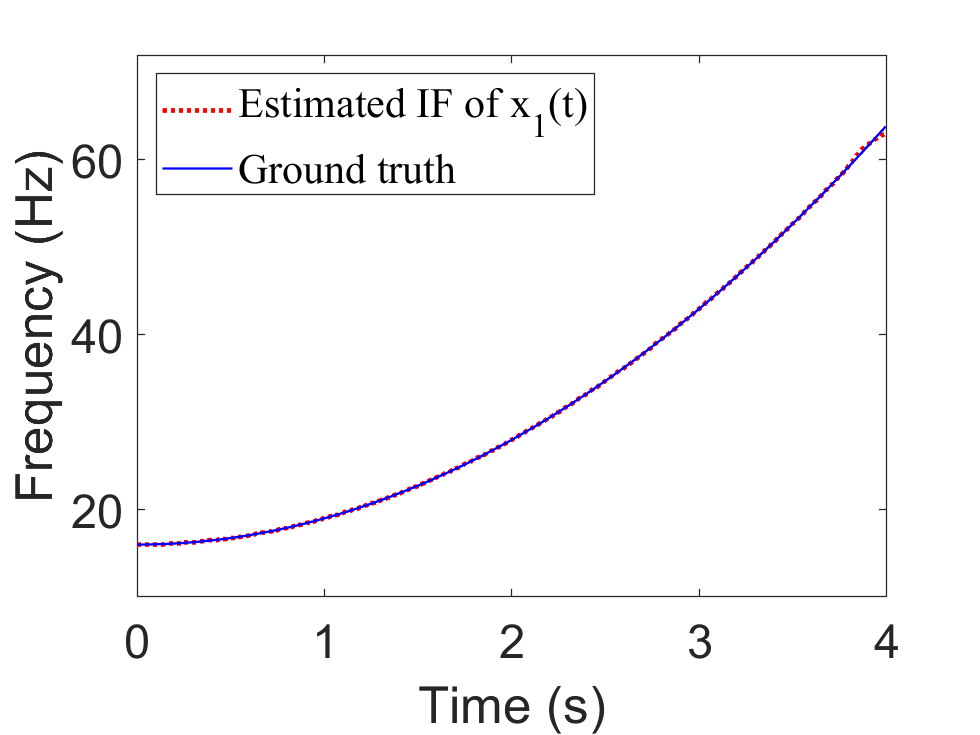}
            \caption{}
        \end{subfigure} &
        \begin{subfigure}[t]{0.22\textwidth}
            \centering
            \includegraphics[width=\linewidth]{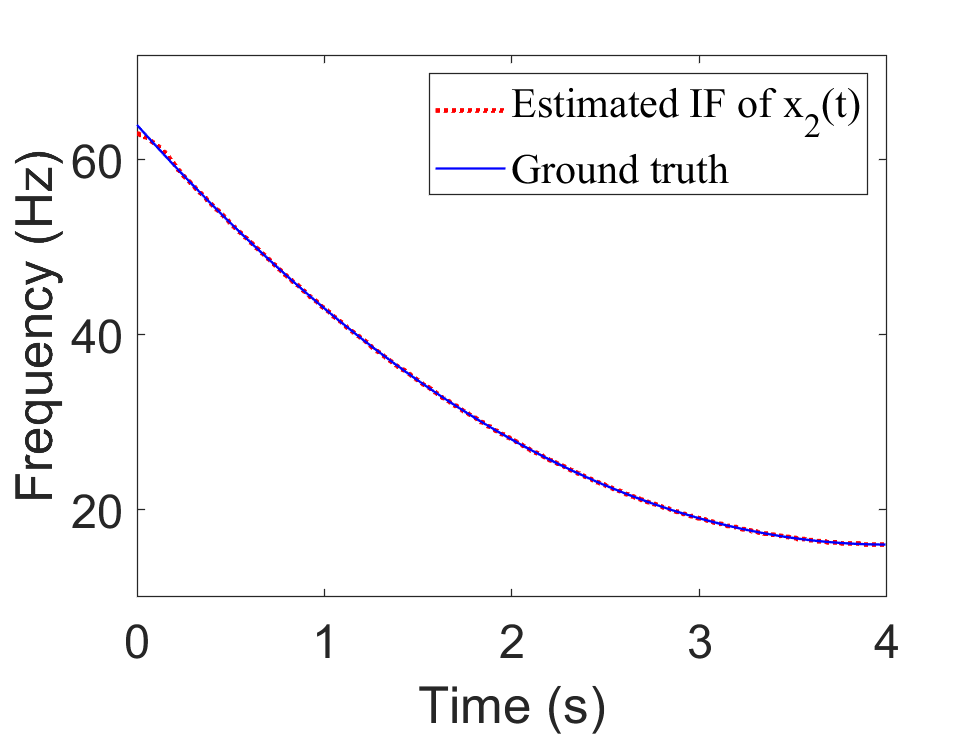}
            \caption{}
        \end{subfigure} &
        \begin{subfigure}[t]{0.22\textwidth}
            \centering
            \includegraphics[width=\linewidth]{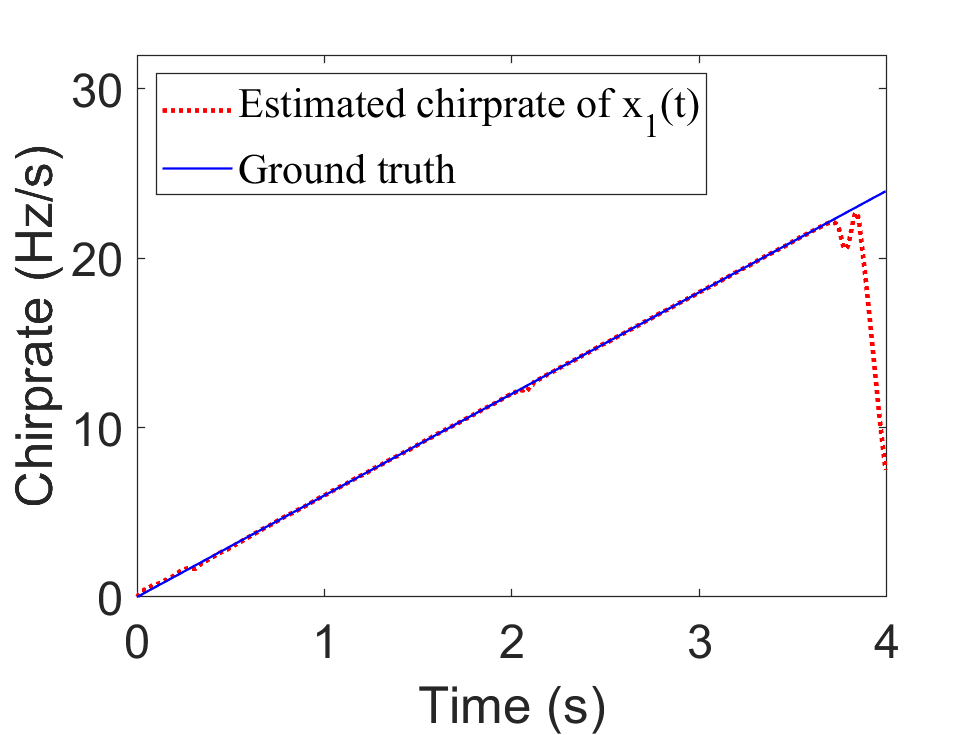}
            \caption{}
        \end{subfigure} &
        \begin{subfigure}[t]{0.22\textwidth}
            \centering
            \includegraphics[width=\linewidth]{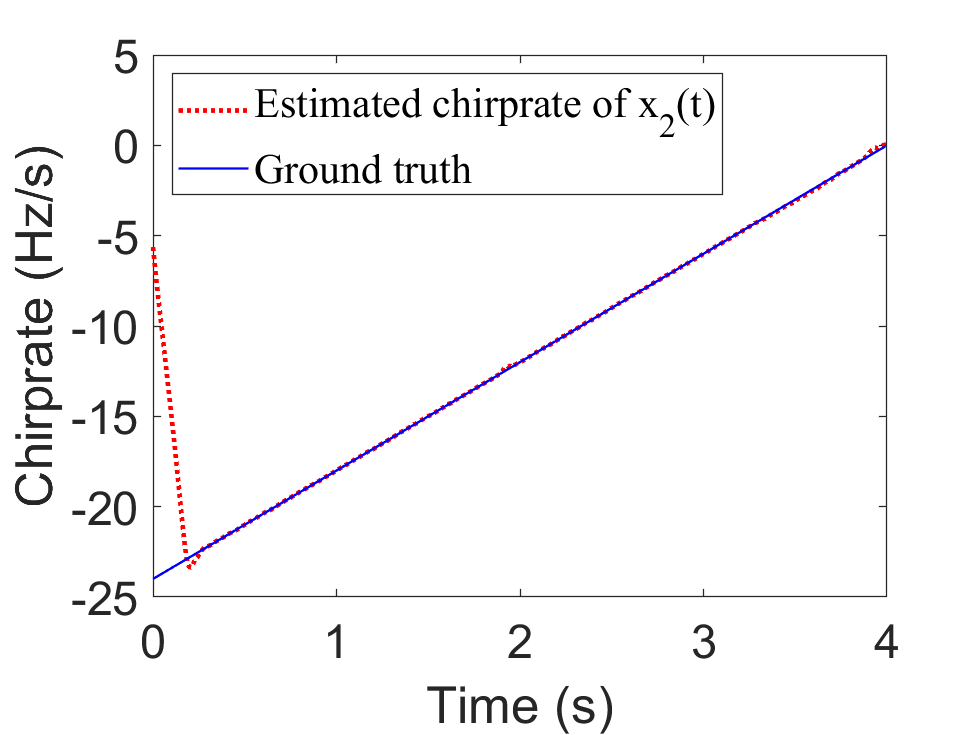}
            \caption{}
        \end{subfigure}
    \end{tabular}
    \caption{Extracted IF and chirprate ridges from: (a)-(d) second-order HSWCT; (e)-(h) third-order HSWCT.}
    \label{figure:IFs_and_CRs_estimation_of_cubic_signal}
\end{figure}

\begin{figure}[H]
    \centering
    \setlength{\tabcolsep}{5pt}
    \begin{tabular}{ccc}
        \begin{subfigure}[t]{0.28\textwidth}  
            \centering
            \includegraphics[width=\linewidth]{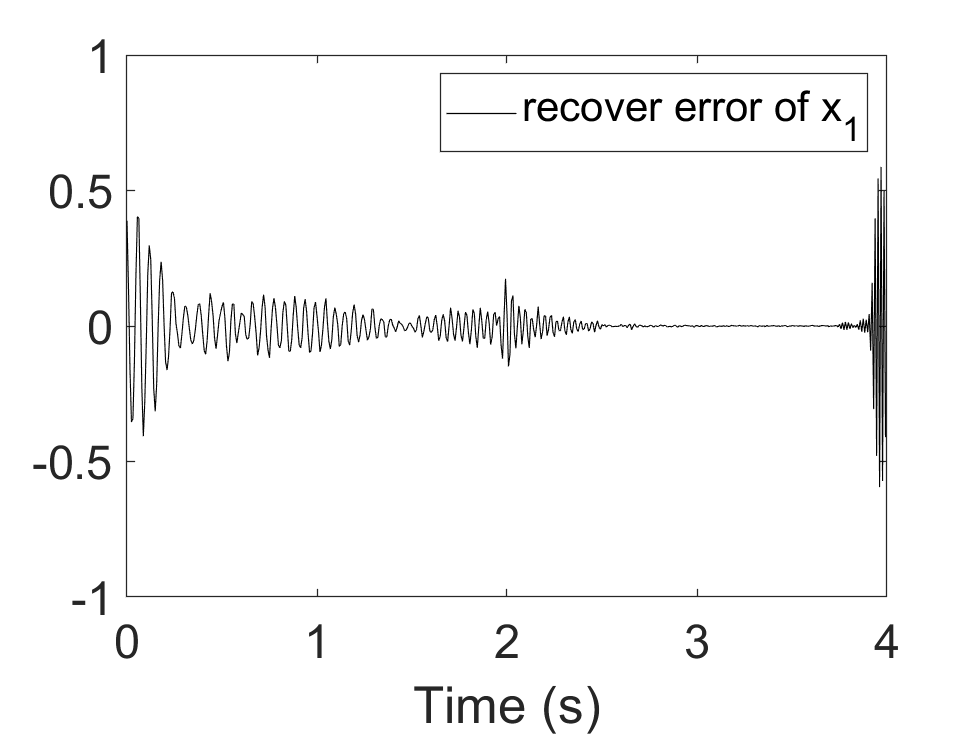}
            \caption{}
        \end{subfigure} &
        \begin{subfigure}[t]{0.28\textwidth}
            \centering
            \includegraphics[width=\linewidth]{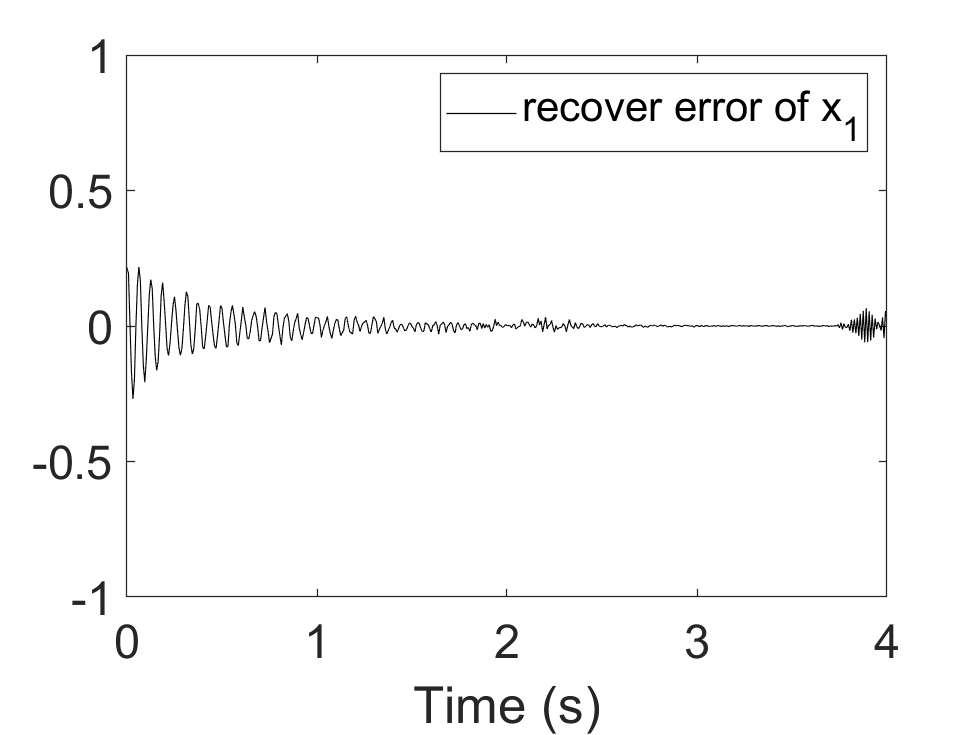}
            \caption{}
        \end{subfigure} &
        \begin{subfigure}[t]{0.28\textwidth}
            \centering
            \includegraphics[width=\linewidth]{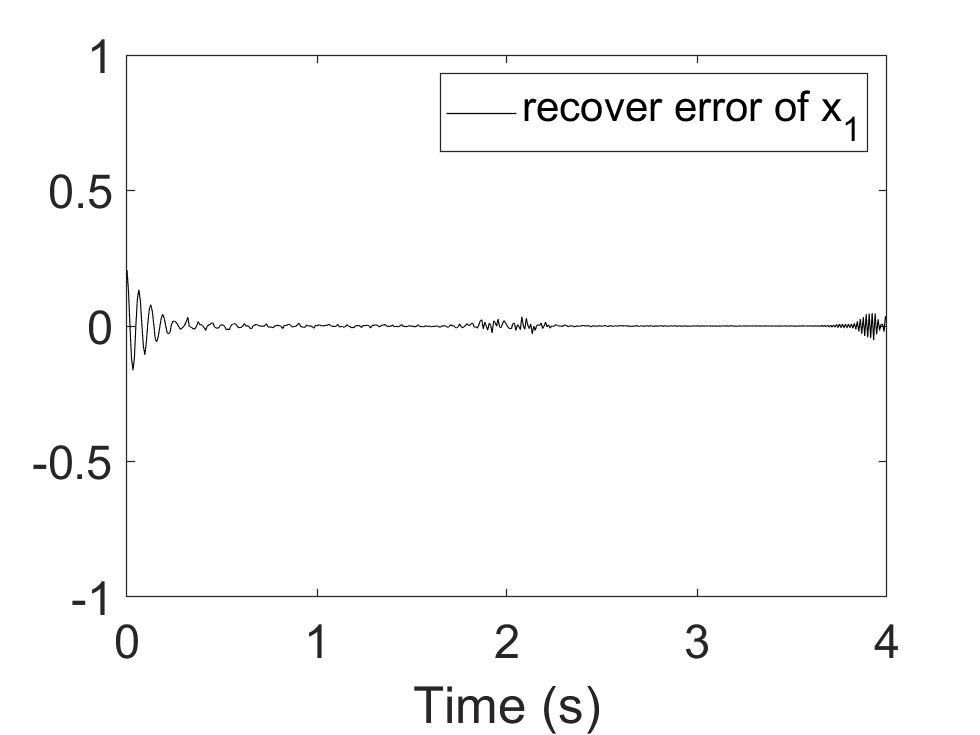}
            \caption{}
        \end{subfigure} \\
        \begin{subfigure}[t]{0.28\textwidth}
            \centering
            \includegraphics[width=\linewidth]{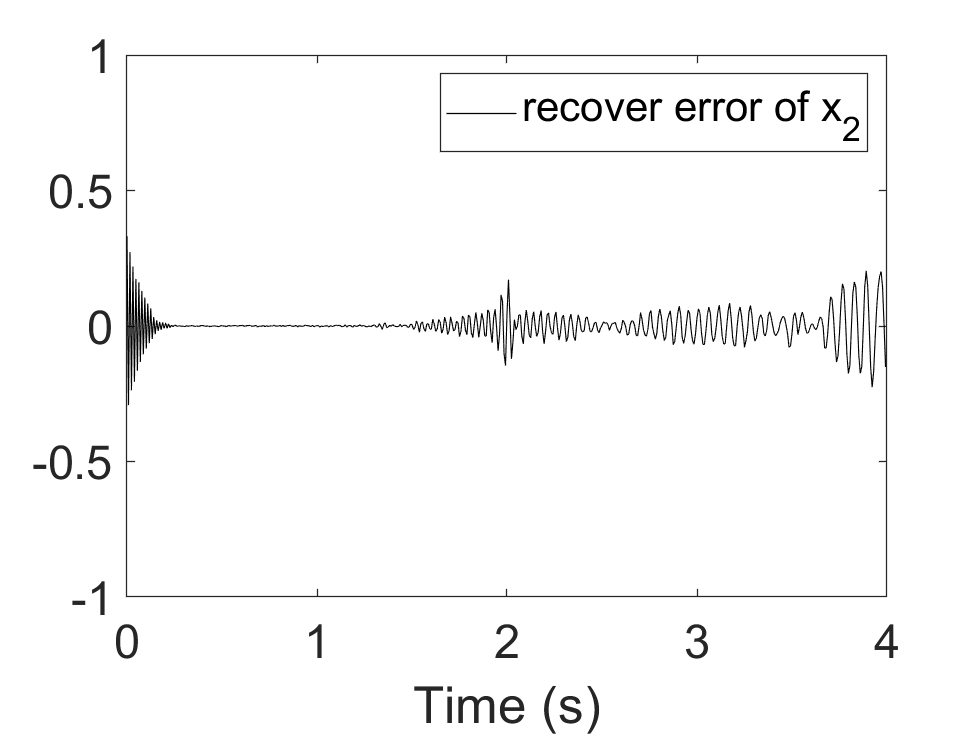}
            \caption{}
        \end{subfigure} &
        \begin{subfigure}[t]{0.28\textwidth}
            \centering
            \includegraphics[width=\linewidth]{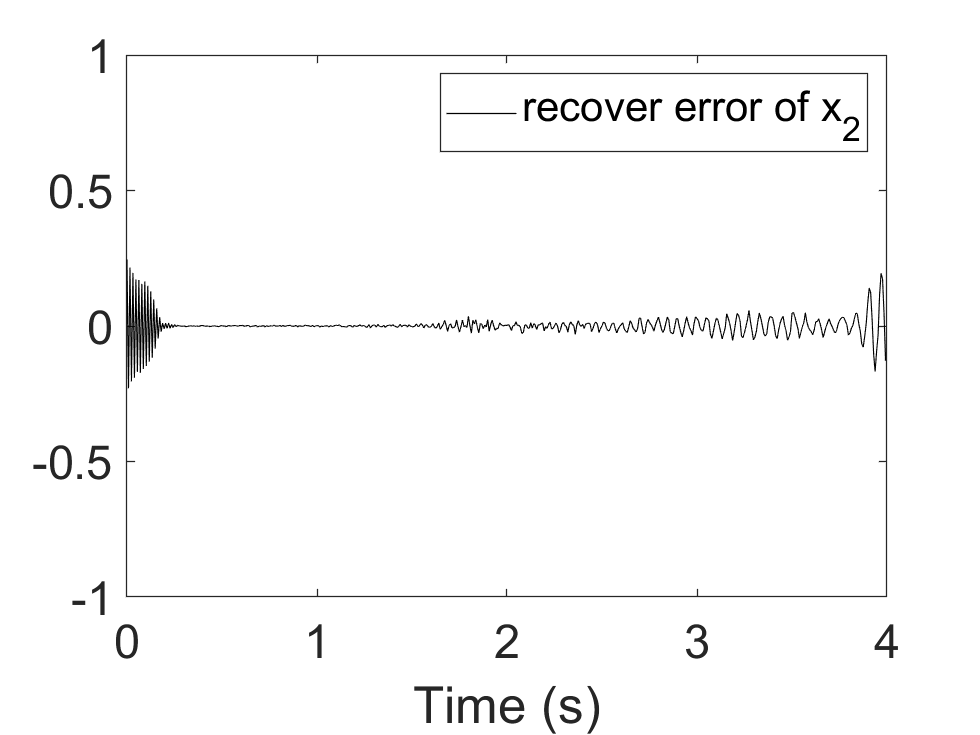}
            \caption{}
        \end{subfigure} &
        \begin{subfigure}[t]{0.28\textwidth}
            \centering
            \includegraphics[width=\linewidth]{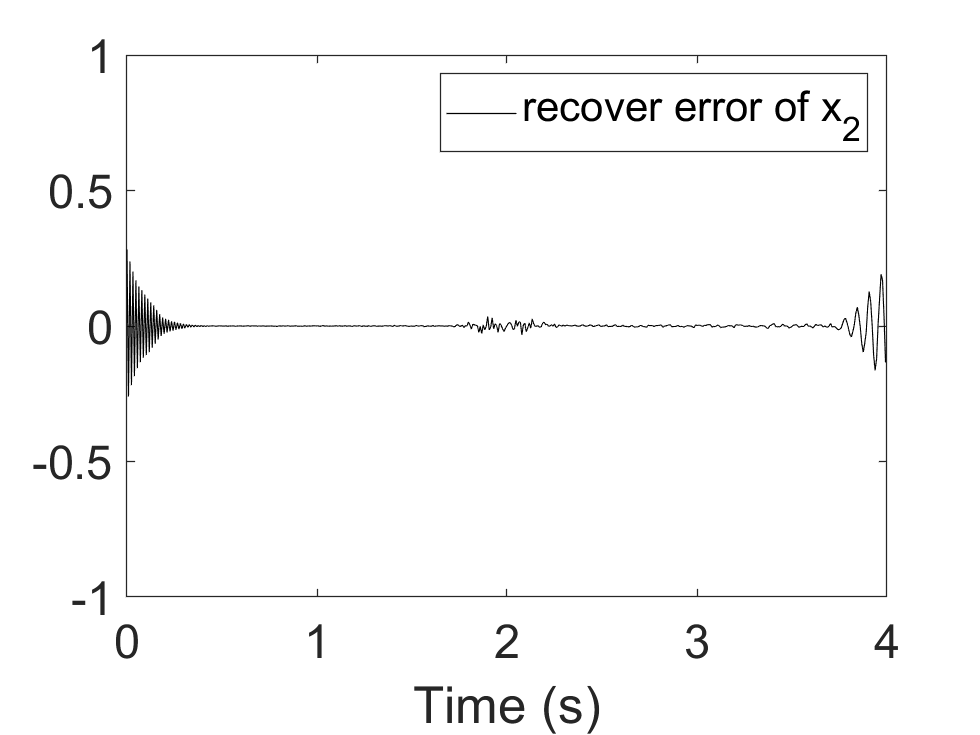}
            \caption{}
        \end{subfigure}
    \end{tabular}
    \caption{Errors between recovered modes and ground truth modes (real part) using extracted IF and chirprate ridges. Top row: error 1; bottom row: error 2. From left to right: second-order HSWCT with \(\sigma_1\), third-order HSWCT with \(\sigma_1\), third-order HSWCT with \(\sigma_2\).}
    \label{figure:Recover errors of cubic signal}
\end{figure}

Let \(\tilde{x}_k(t)\) (\(k=1,2\)) denote the signals \(x_k(t)\) reconstructed by Eq.~\eqref{solution_SSO} using \(\psi_{\sigma_1}(t)\) and \(\psi_{\sigma_2}(t)\), respectively (recall \(\sigma_2 = \frac{1}{3}\sigma_1\)).
Fig.~\ref{figure:Recover errors of cubic signal} displays the reconstruction errors (real part), \(\tilde{x}_k(t) - x_k(t)\), obtained from the estimated IF and chirprate curves. 
Panels (a) and (d) correspond to the second-order HSWCT \(\mathcal{U}_x^{2,\psi_{\sigma_1}}(\xi,b,\gamma)\), while panels (b) and (e) correspond to the third-order HSWCT \(\mathcal{U}_x^{3,\psi_{\sigma_1}}(\xi,b,\gamma)\).
The results show that accurate IF and chirprate estimation leads to more precise mode retrieval, whereas estimation errors compromise reconstruction quality.
Furthermore, panels (c) and (f) reveal that employing \(\psi_{\sigma_2}(t)\) for mode retrieval substantially reduces these errors, validating the use of \(\sigma_2\) for mode retrieval as a practical and effective strategy.

Next, we examine the impact of different estimator orders on a sinusoidal frequency modulation signal defined as

\begin{equation}
\label{def_yt}
y(t) = y_1(t) + y_2(t), \quad y_1(t) = A_1(t) e^{2\pi i \phi_1(t)}, \quad y_2(t) = A_2(t) e^{2\pi i \phi_2(t)}, \quad t \in [0, 4],
\end{equation}
where
\[
A_1(t) = \exp\left(-0.01 t^4 + 0.08 t^3 - 0.26 t^2 + 0.3 t - 0.16\right), \quad \phi_1(t) = 41t - \frac{32}{\pi} \sin\left(\frac{\pi}{2} t\right),
\]
\[
A_2(t) = \exp\left(-0.02 t^4 + 0.15 t^3 - 0.48 t^2 + 0.63 t - 0.32\right), \quad \phi_2(t) = 41t + \frac{32}{\pi} \sin\left(\frac{\pi}{2} t\right).
\]
The corresponding IFs and chirprates are given by
\begin{align*}
\phi_1'(t) &= 41 - 16 \cos\left(\frac{\pi}{2} t\right), \quad \phi_1''(t) = 8\pi \sin\left(\frac{\pi}{2} t\right), \\
\phi_2'(t) &= 41 + 16 \cos\left(\frac{\pi}{2} t\right), \quad \phi_2''(t) = -8\pi \sin\left(\frac{\pi}{2} t\right).
\end{align*}
The signal \(y(t)\) is sampled at 128 Hz, and the optimal parameter \(\sigma_1 = 4.9\) is obtained from Eq.~\eqref{find_sigma}. The IFs of the two components intersect at \(t = 1\) s and \(t = 3\) s, both at frequency \(\omega = 41\) Hz.
Fig.~\ref{fig:IFs_CRs_Doppler} shows the IFs and chirprates of \(y(t)\). Although the IFs and chirprates occasionally overlap, their respective overlapping instants do not coincide, thus always satisfying the separation condition Eq.~\eqref{separation_condition}.
\begin{figure}[H]
    \centering
    \begin{tabular}{cc}
        \subfloat[IFs of signal $y(t)$]{
            \includegraphics[width=0.32\textwidth]{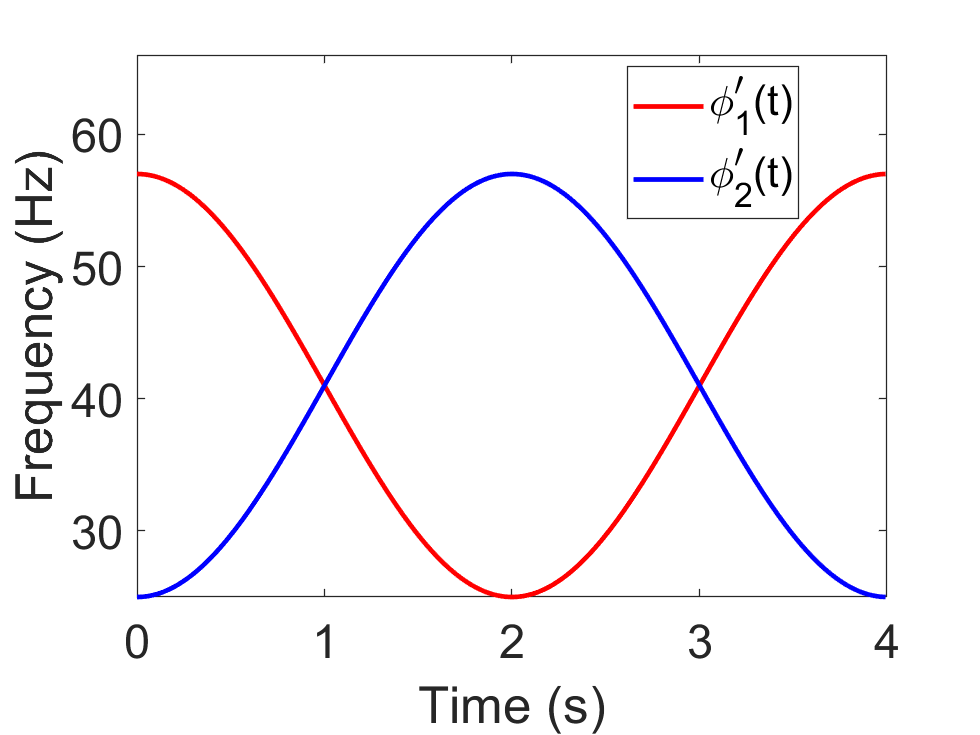}
        } &
        \subfloat[Chirprates of signal $y(t)$]{
\includegraphics[width=0.32\textwidth]{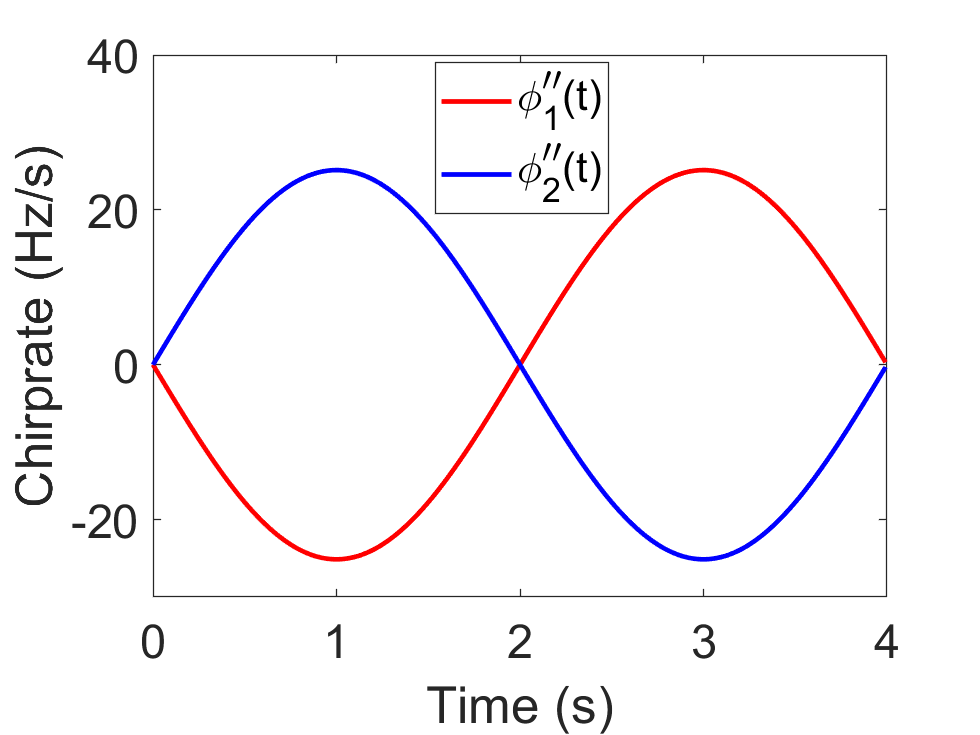}
        } 
    \end{tabular}
    \caption{\small IFs and chirprates of signal $y(t)$}
    \label{fig:IFs_CRs_Doppler}
\end{figure}

To measure the accuracy of the extracted IFs and chirprates, we compute their root mean square error (RMSE). Let \(\mathbf{f}\) denote the discrete ground-truth values of \(\phi'_k(t)\) (for IF) or \(\phi''_k(t)\) (for chirprate), and \(\widetilde{\mathbf{f}}\) the corresponding estimates. The RMSE is defined as
\[
E_f = \frac{1}{\sqrt{n}} \left\lVert \mathbf{f} - \widetilde{\mathbf{f}} \right\rVert_2.
\]
To avoid boundary effects, we compute the error over the central 75\% of the signal, i.e., 
\(
f_{\left\lfloor n/8 \right\rfloor}, \dots, f_{\left\lfloor 7n/8 \right\rfloor}.
\)
Near the signal edges, the time window lacks sufficient data on one side, leading to 
incomplete frequency localization and increased estimation bias. Including 
these boundary points in the RMSE would therefore introduce errors that do 
not accurately reflect the method's intrinsic performance. 
By focusing on the central portion, the RMSE reflects the estimator's performance under reliable data conditions.

\begin{figure}[H]
    \centering
    \begin{tabular}{cccc}
        \begin{subfigure}[t]{0.22\textwidth}
            \centering
            \resizebox{\linewidth}{!}{\includegraphics{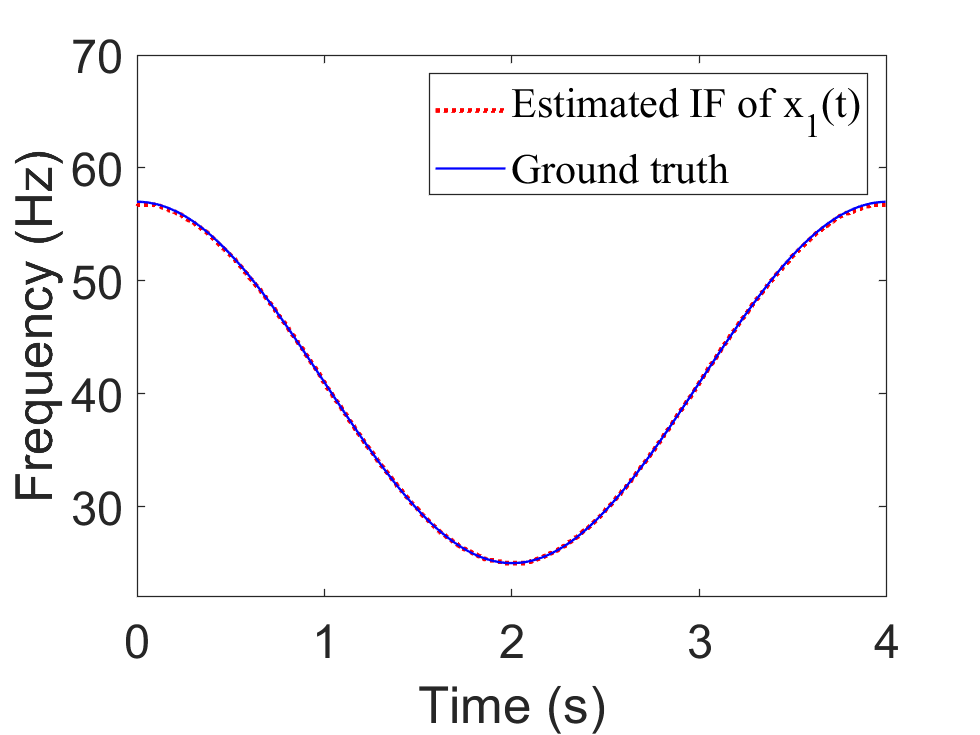}}
        \end{subfigure} &
        \begin{subfigure}[t]{0.22\textwidth}
            \centering
            \resizebox{\linewidth}{!}{\includegraphics{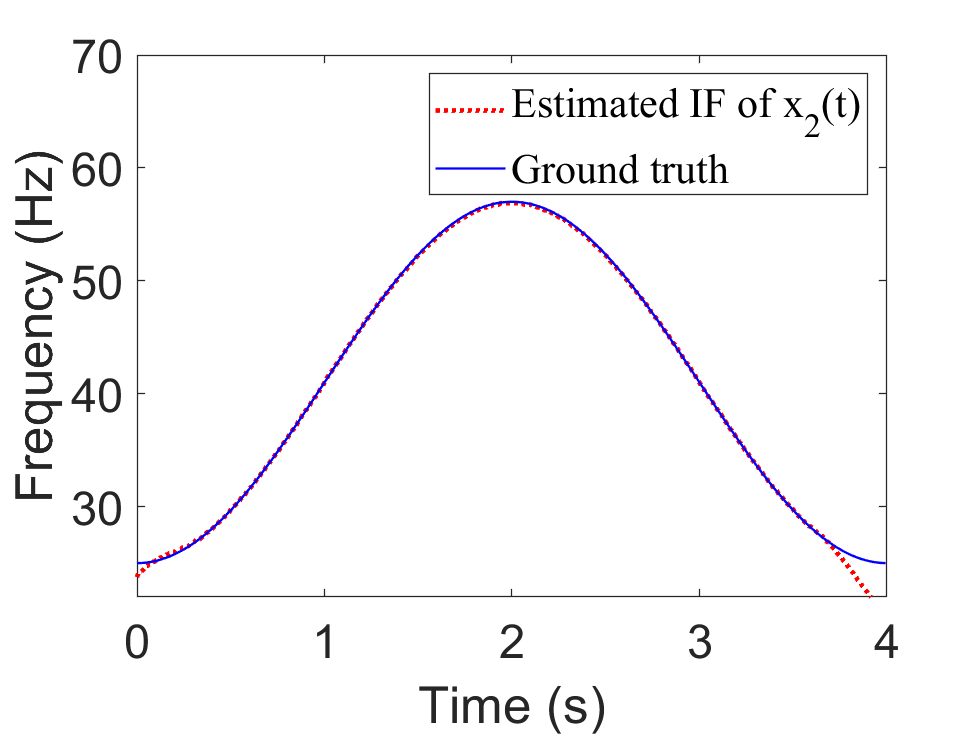}}
        \end{subfigure} &
        \begin{subfigure}[t]{0.22\textwidth}
            \centering
            \resizebox{\linewidth}{!}{\includegraphics{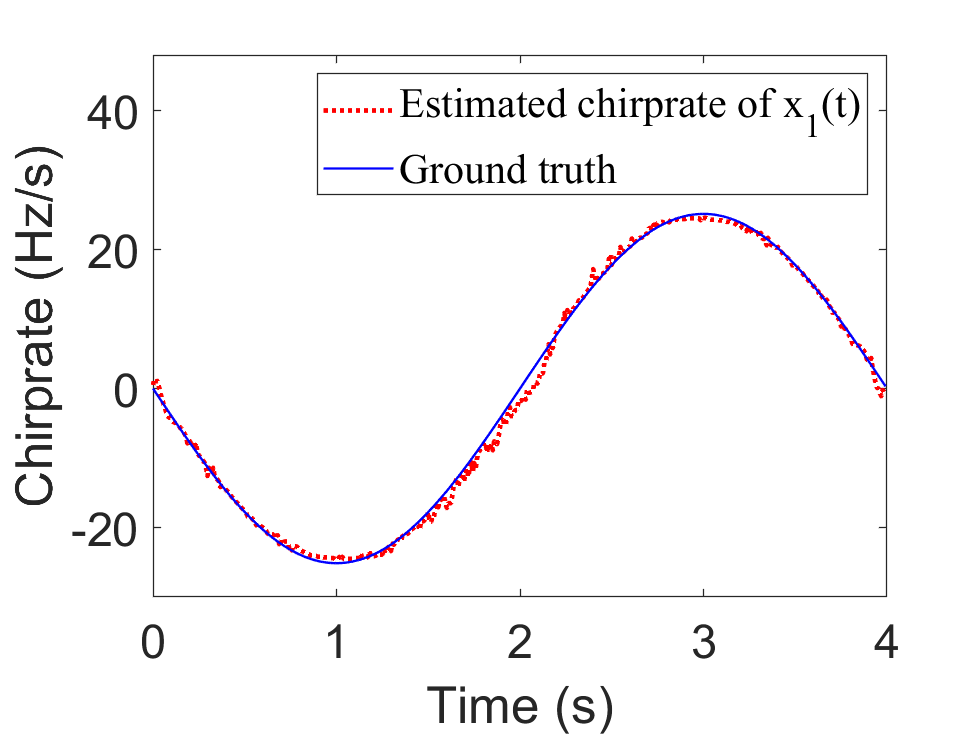}}
        \end{subfigure} &
        \begin{subfigure}[t]{0.22\textwidth}
            \centering
            \resizebox{\linewidth}{!}{\includegraphics{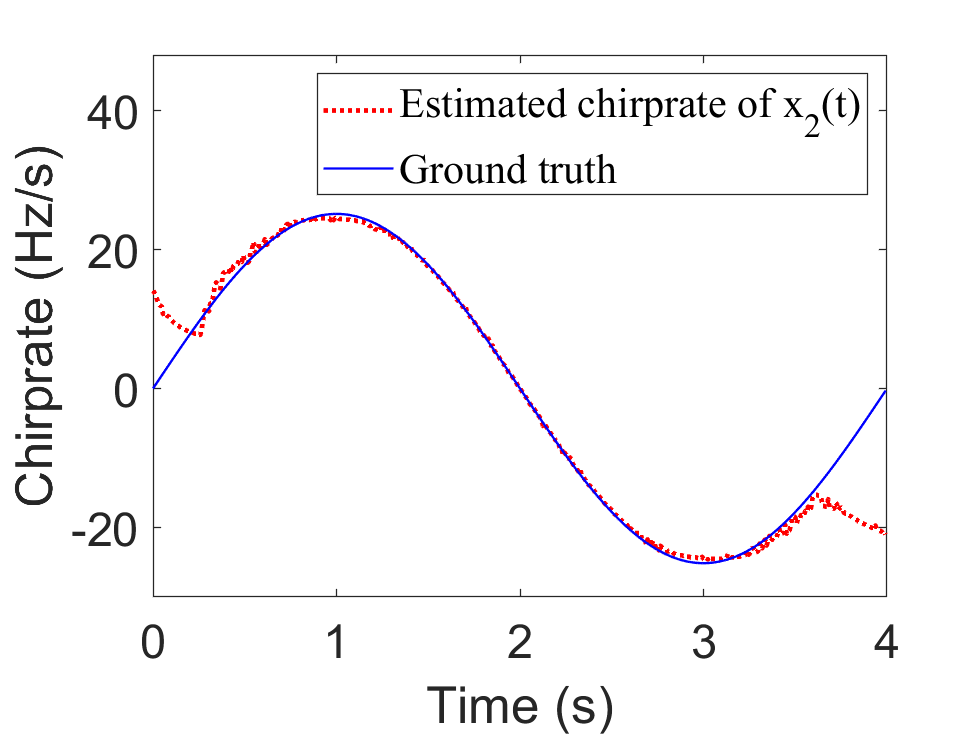}} 
        \end{subfigure} \\
        \begin{subfigure}[t]{0.22\textwidth}
            \centering
            \resizebox{\linewidth}{!}{\includegraphics{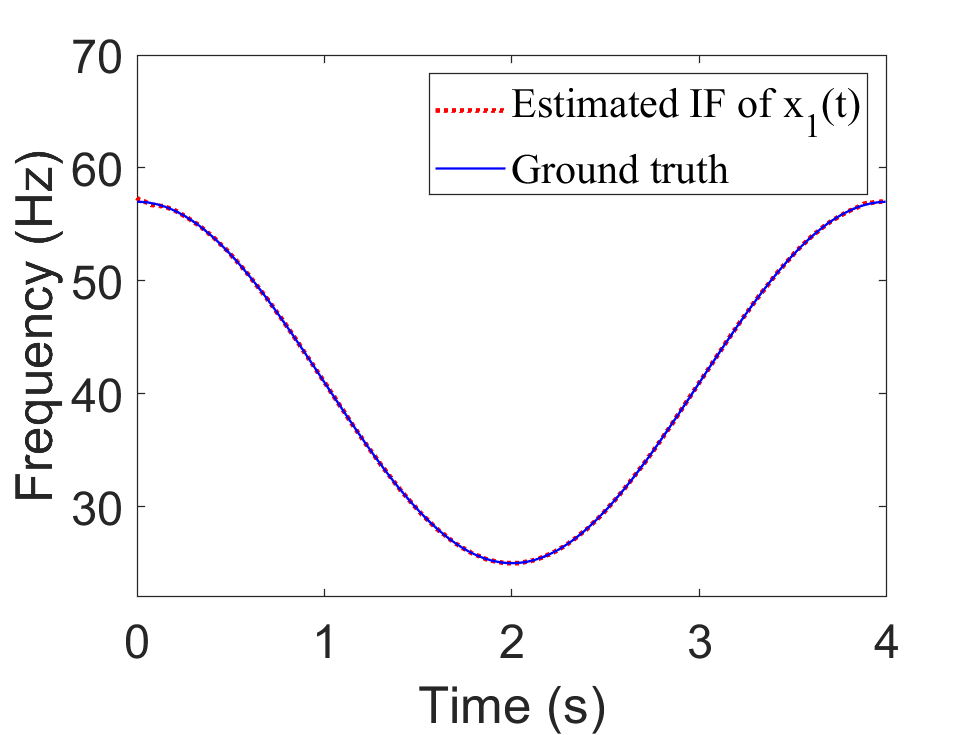}} 
        \end{subfigure} &
        \begin{subfigure}[t]{0.22\textwidth}
            \centering
            \resizebox{\linewidth}{!}{\includegraphics{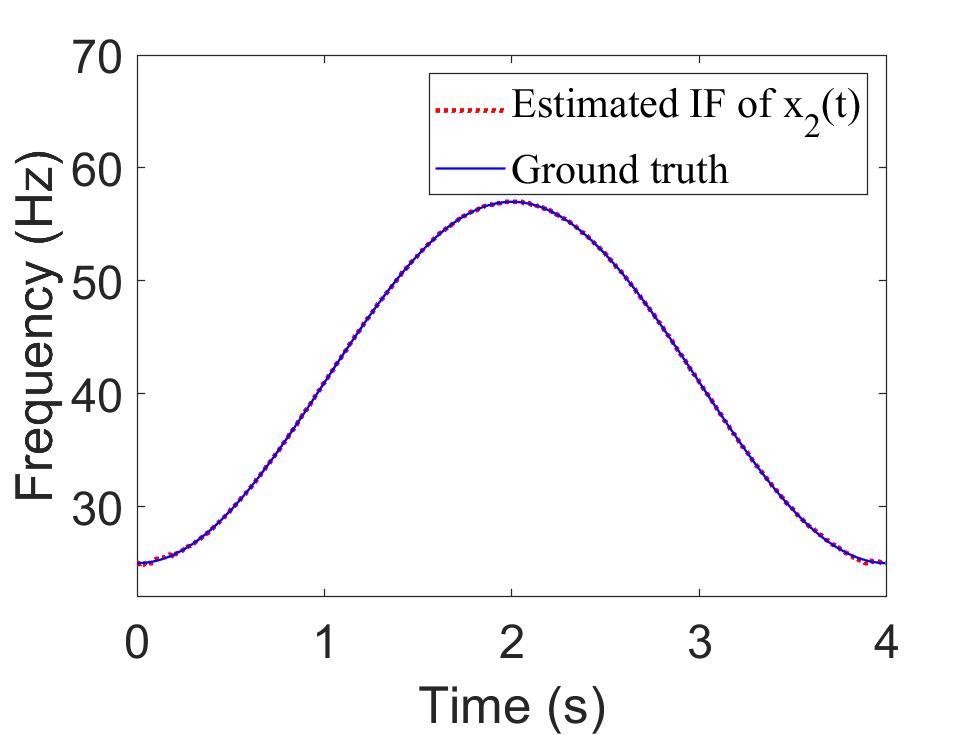}}    
        \end{subfigure} &
        \begin{subfigure}[t]{0.22\textwidth}
            \centering
            \resizebox{\linewidth}{!}{\includegraphics{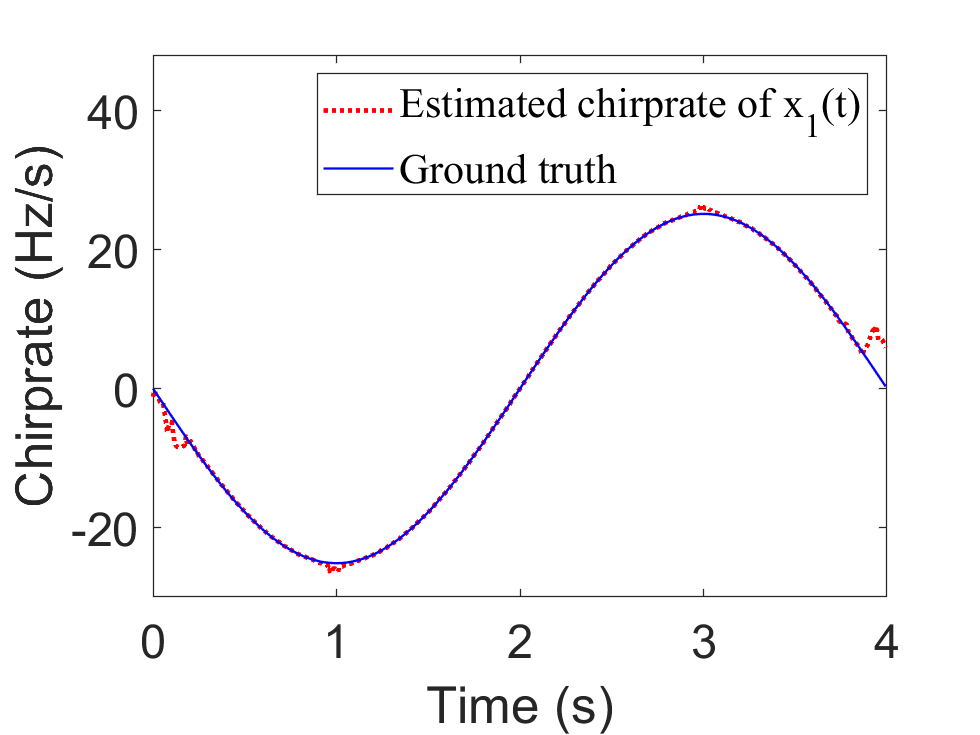}}    
        \end{subfigure} &
        \begin{subfigure}[t]{0.22\textwidth}
            \centering
            \resizebox{\linewidth}{!}{\includegraphics{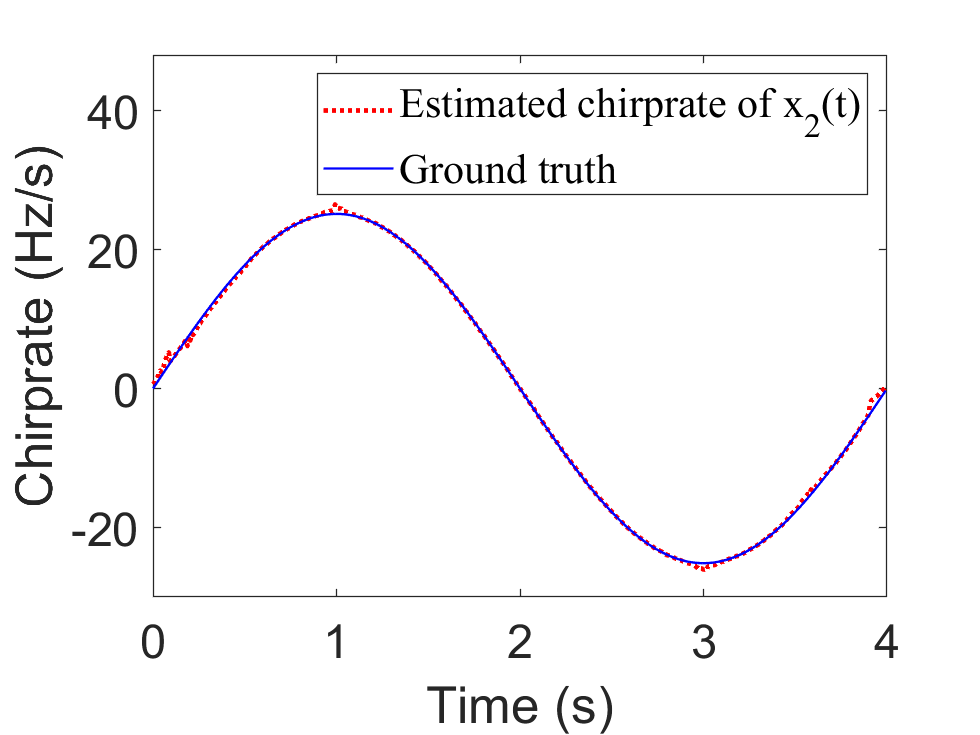}}    
        \end{subfigure} \\
        \begin{subfigure}[t]{0.22\textwidth}
            \centering
            \resizebox{\linewidth}{!}{\includegraphics{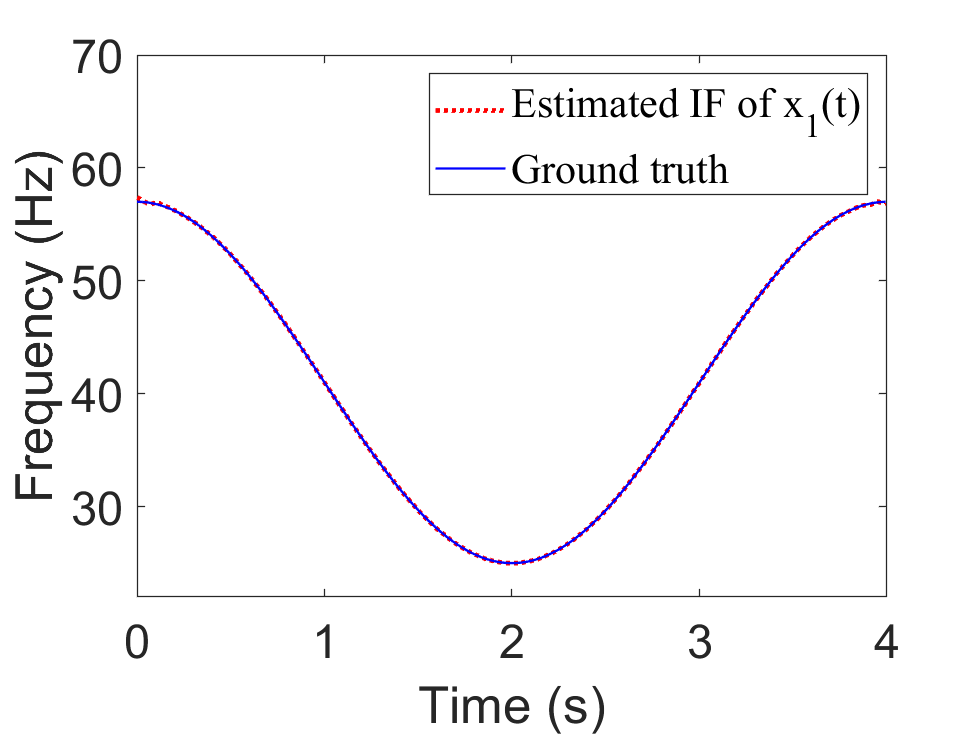}} 
        \end{subfigure} &
        \begin{subfigure}[t]{0.22\textwidth}
            \centering
            \resizebox{\linewidth}{!}{\includegraphics{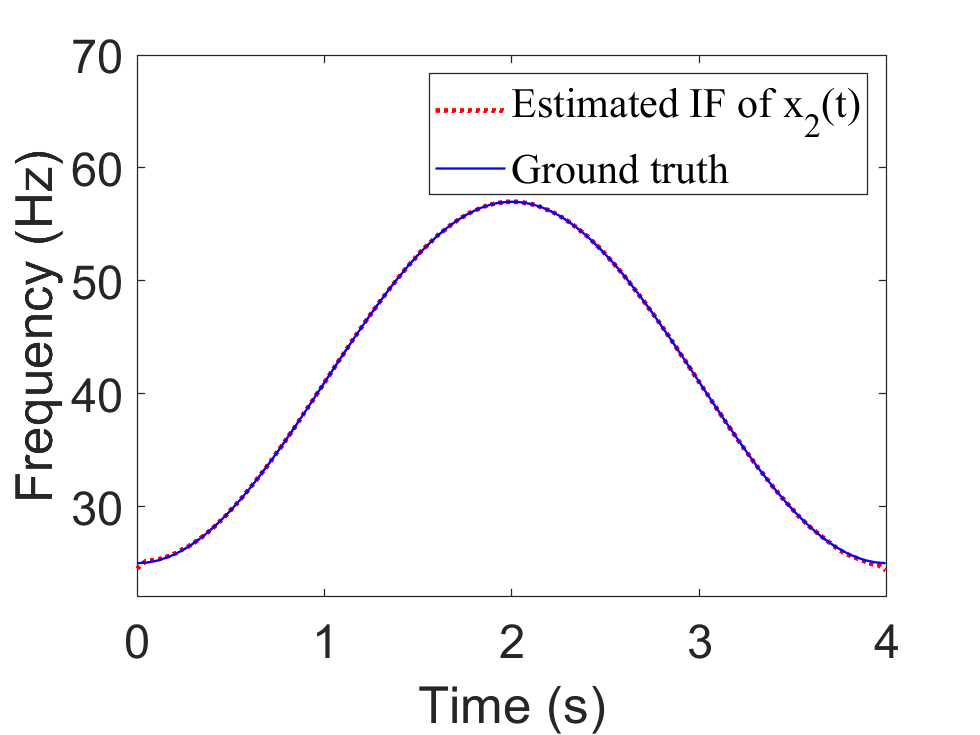}}      
        \end{subfigure} &
        \begin{subfigure}[t]{0.22\textwidth}
            \centering
            \resizebox{\linewidth}{!}{\includegraphics{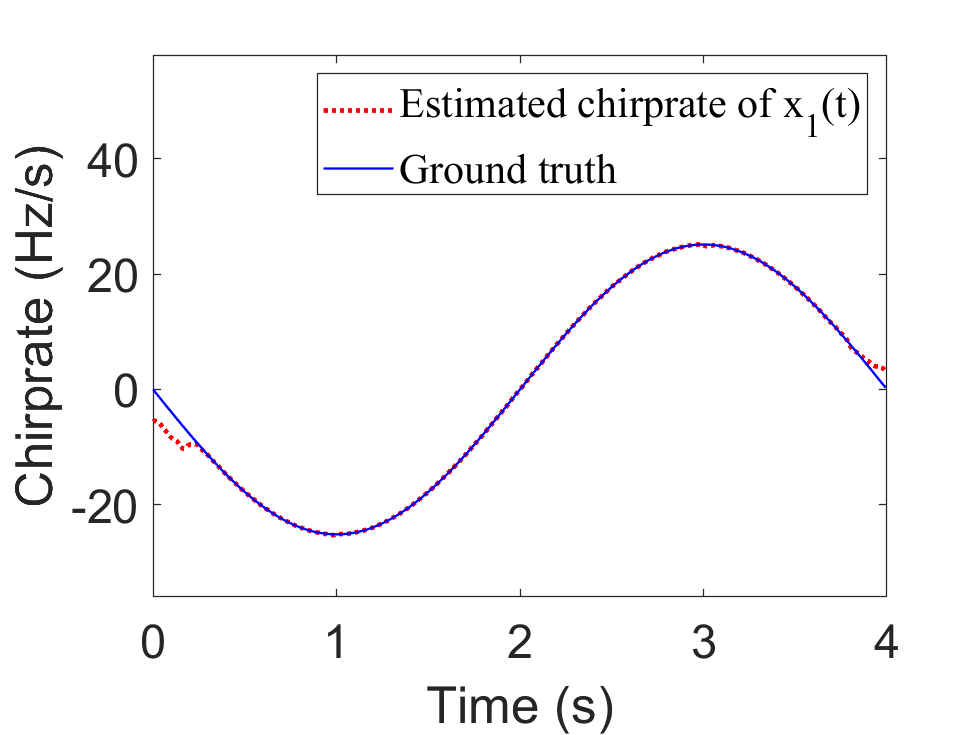}}  
        \end{subfigure} &
        \begin{subfigure}[t]{0.22\textwidth}
            \centering
            \resizebox{\linewidth}{!}{\includegraphics{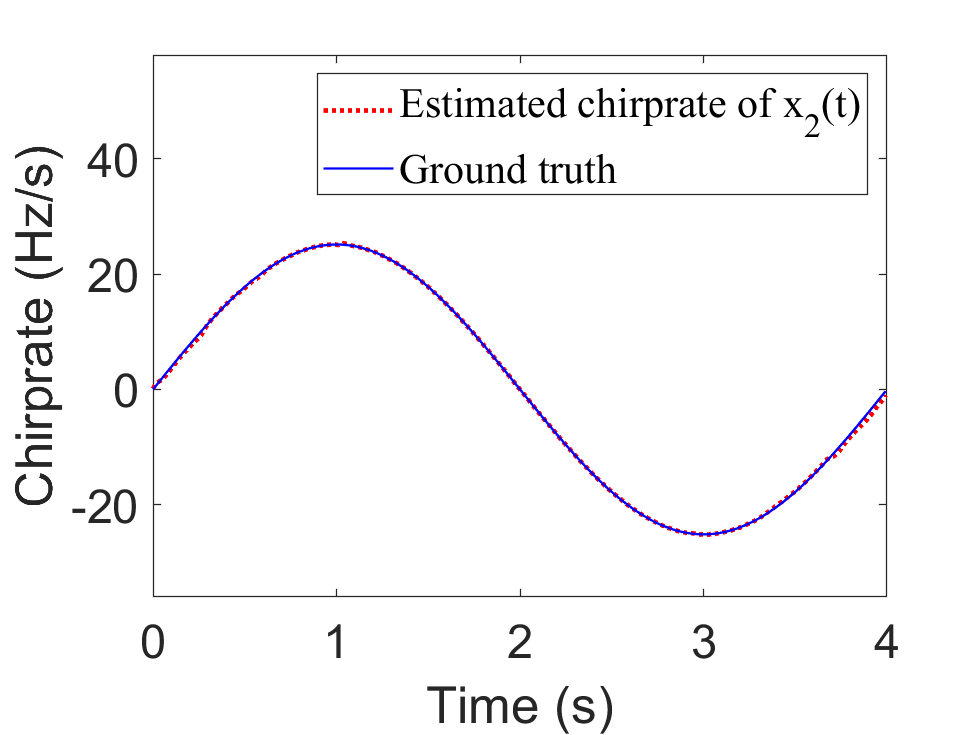}}
        \end{subfigure} \\
         \begin{subfigure}[t]{0.22\textwidth}
            \centering
            \resizebox{\linewidth}{!}{\includegraphics{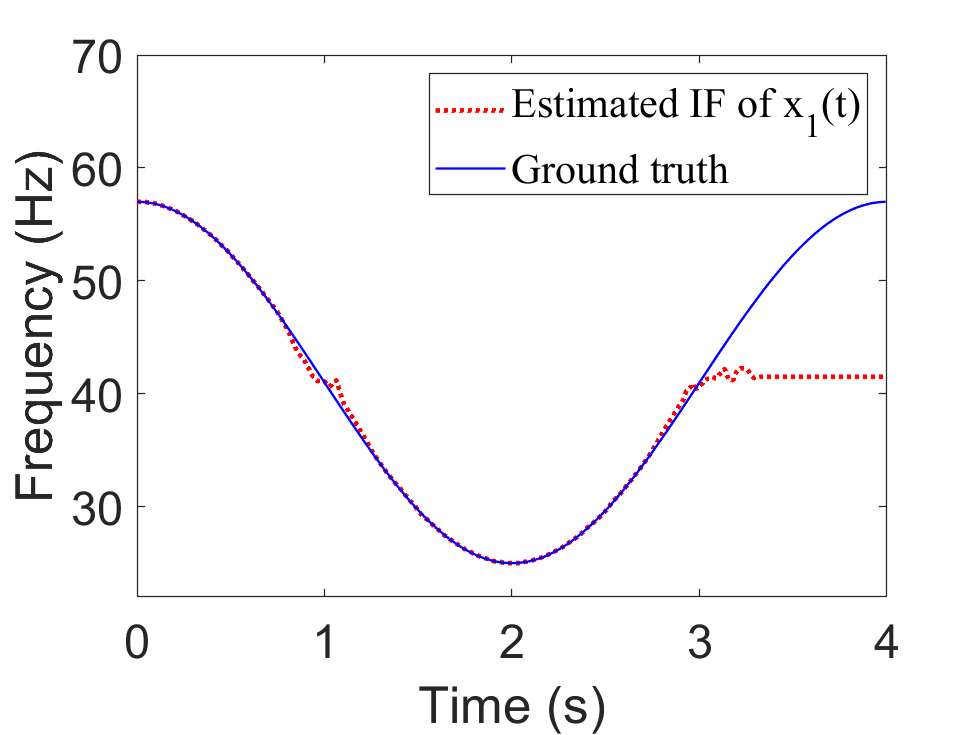}}      
        \end{subfigure} &
         \begin{subfigure}[t]{0.22\textwidth}
            \centering
            \resizebox{\linewidth}{!}{\includegraphics{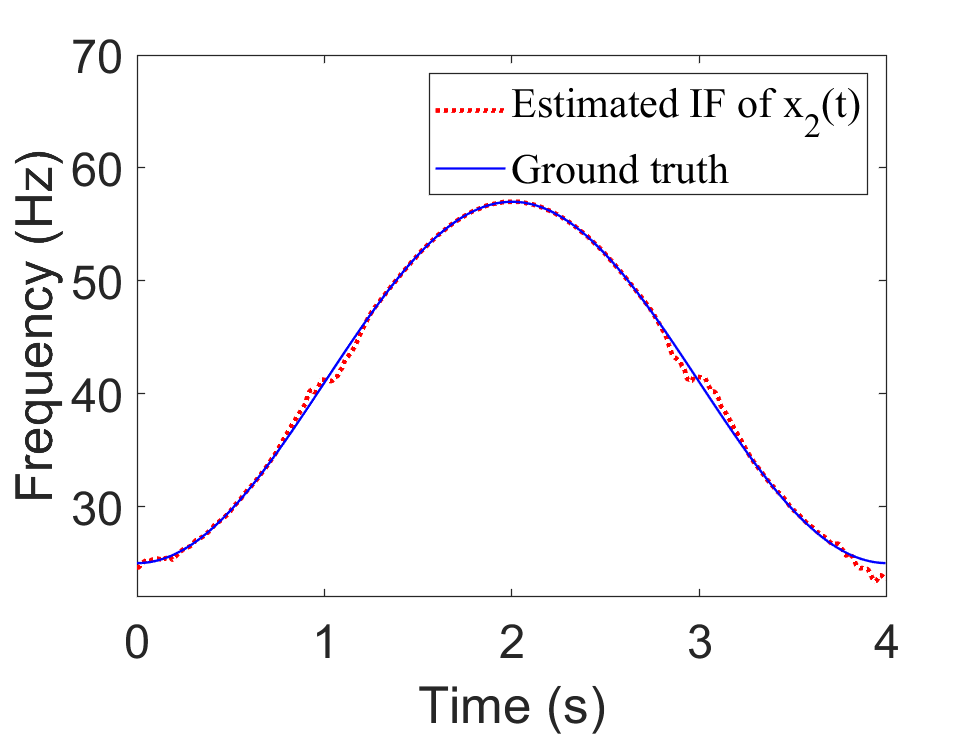}}      
        \end{subfigure} &
        \begin{subfigure}[t]{0.22\textwidth}
            \centering
            \resizebox{\linewidth}{!}{\includegraphics{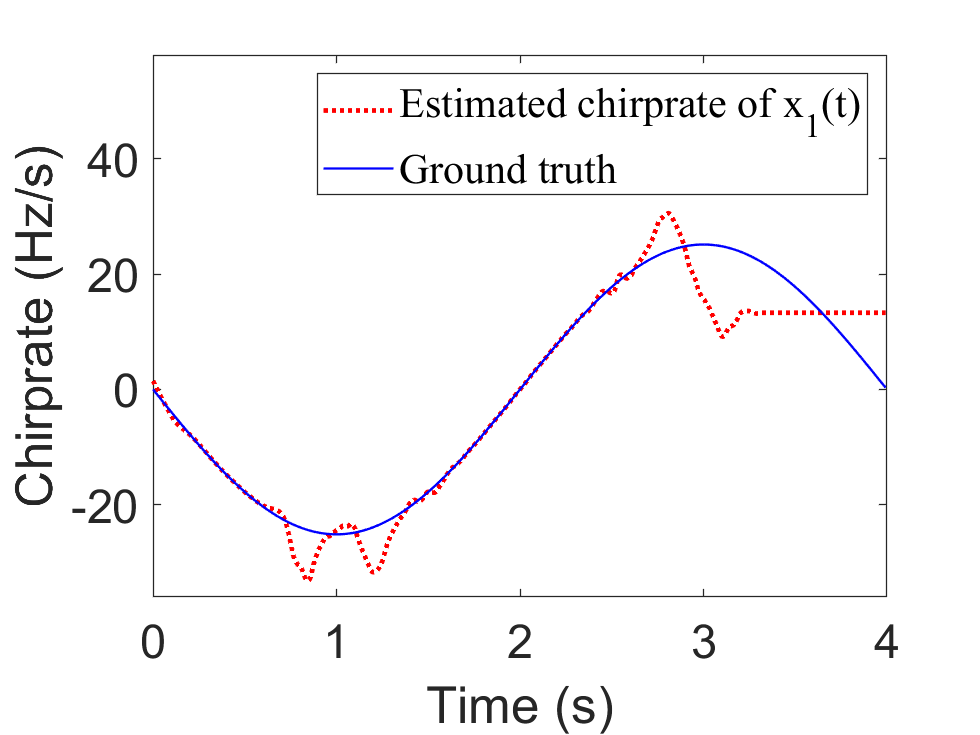}}  
        \end{subfigure} &
        \begin{subfigure}[t]{0.22\textwidth}
            \centering
            \resizebox{\linewidth}{!}{\includegraphics{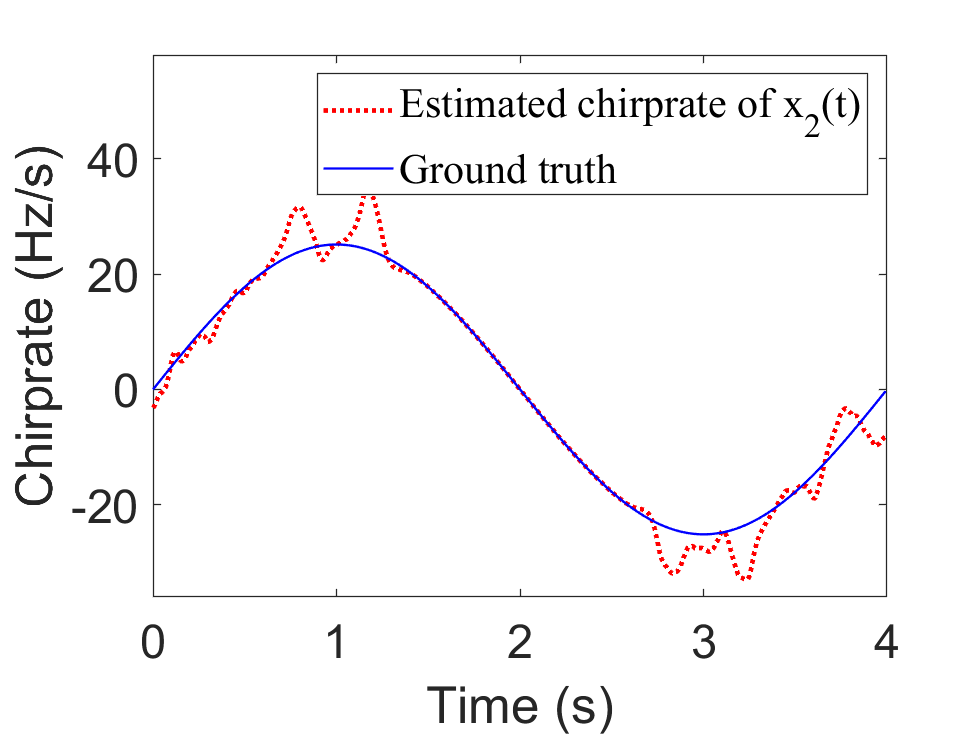}}     
        \end{subfigure} 
    \end{tabular}
\caption{\small IF and chirprate estimation of \( y(t) \) with different orders of HSWCT. 
Rows from top to bottom: second-order \( \mathcal{U}_y^{2,\psi_{\sigma_1}}(\xi,b,\gamma)\); 
third-order \( \mathcal{U}_y^{3,\psi_{\sigma_1}}(\xi,b,\gamma)\); 
fourth-order \( \mathcal{U}_y^{4,\psi_{\sigma_1}}(\xi,b,\gamma)\); 
and multiple synchrosqueezed fourth-order \( \mathbf{\mathcal{U}}_{y,5}^{4,\psi_{\sigma_1}}(\xi,b,\gamma) \) with five iterations. 
In each row, the left two panels show IF estimates and the right two panels show chirprate estimates.}
    \label{fig:IFs_and_CRs_estimation_of_Doppler_signal}
\end{figure}

We employ the second-, third-, and fourth-order HSWCT \(\mathcal{U}_{y}^{n,\psi_{\sigma_1}}(\xi,b,\gamma)\), \(n=2, 3, 4\), to extract the IFs and chirprates of \(y(t)\). For comparison, we also present the results obtained with \(\mathcal{U}_{y,5}^{4,\psi_{\sigma_1}}(\xi,b,\gamma)\), a fourth-order MHSWCT with five iterations of synchrosqueezing.
The extracted IFs and chirprates are shown in Fig.~\ref{fig:IFs_and_CRs_estimation_of_Doppler_signal}. 
The corresponding RMSEs, running times, and post-squeezing R\'enyi entropies are summarized in Table~\ref{table:RMSE_results}.

\begin{table}[H]
    \centering
   \caption{\small  RMSE of IF and chirprate estimation, running time, and R\'enyi entropy for \( y(t) \), obtained using synchrosqueezed transforms of different orders.}
    \label{table:RMSE_results}
    \begin{tabular}{lcccc}
        \toprule
        \textbf{Metric} & \textbf{2nd-order} & \textbf{3rd-order} & \textbf{4th-order} & \textbf{Multiple 4th-order} \\
        \midrule
        IF1   & 0.0570 & 0.0443 & 0.0359 & 2.5682 \\
        CR1   & 1.4048 & 0.2177 & 0.0326 & 5.0317 \\
        IF2   & 0.0736 & 0.0362 & 0.0359 & 0.4174 \\
        CR2   & 0.4996 & 0.2077 & 0.0718 & 3.3107 \\
        \midrule
        Running time (s) & 8.41 & 14.29 & 22.57 & 29.96 \\
        R\'enyi entropy & 10.03 & 6.27 & 6.14 & 6.39 \\
        \bottomrule 
    \end{tabular}
\end{table}
All three transforms achieve effective IF extraction. However, their chirprate estimation performance differs markedly: the second-order HSWCT introduces significant bias; the third-order HSWCT exhibits only slight deviations near the crossing points (\(t = 1\,\text{s}\) and \(t = 3\,\text{s}\)); and the fourth-order HSWCT achieves nearly error-free chirprate extraction. 
For the fourth-order MHSWCT, IF estimation remains relatively accurate across most temporal regions, yet measurable deviations appear at the same instants (\(t = 1\,\text{s}\) and \(t = 3\,\text{s}\)), while chirprate estimation exhibits substantially larger biases at these points.
As expected, higher-order transforms require longer computation times due to their increased complexity. 
The R\'enyi entropy decreases substantially from second- to third-order, indicating improved concentration, but only marginally from third- to fourth-order. Hence, for this example, we deem the third-order HSWCT adequate for the time-frequency analysis of the signal, as the fourth-order counterpart provides only limited gains in concentration at a substantially increased computational cost. 
Interestingly, applying multiple synchrosqueezing iterations leads to an increase in R\'enyi entropy, suggesting a degradation in concentration.

Next, we present the mode retrieval results. Let \(\tilde{y}_k(t)\) (\(k=1,2\)) denote the reconstructed versions of \(y_k(t)\) obtained via Eq.~\eqref{solution_SSO} using the window parameters \(\sigma_1\) and \(\sigma_2\), respectively.  
Fig.~\ref{fig:recovered_errors_sigma12} shows the reconstruction errors \(\text{real}(\tilde{y}_k(t) - y_k(t))\) based on the IFs and chirprates estimated from the fourth-order HSWCT \(\mathcal{U}_y^{4,\psi_{\sigma_1}}(\xi,b,\gamma)\).
Panels (a) and (b) show the reconstruction errors for \(y_1(t)\) and \(y_2(t)\), respectively, obtained by applying the  wavelet \(\psi_{\sigma_1}(t)\) in the reconstruction formula Eq.~\eqref{solution_SSO}. Panels (c) and (d) display the corresponding errors when the narrower wavelet \(\psi_{\sigma_2}(t)\) (with \(\sigma_2 = \sigma_1/3 = 1.63\)) is used instead.
With the wider wavelet \(\psi_{\sigma_1}(t)\), the reconstructed modes exhibit distinct error patterns: for \(y_1(t)\), the error peaks at \(t = 2\) s (panel a), while for \(y_2(t)\), significant errors are mainly concentrated near the boundaries of the time domain (panel b).
In contrast, employing the narrower wavelet \(\psi_{\sigma_2}(t)\) substantially reduces the reconstruction errors in panels (c) and (d), leading to significantly enhanced accuracy. This confirms that using a smaller \(\sigma\) is a viable strategy for most complex signals.

\begin{figure}[H]
    \centering
    \setlength{\tabcolsep}{5pt} 
    \begin{tabular}{cccc}
        \begin{subfigure}[t]{0.22\textwidth}
            \centering
            \includegraphics[width=\linewidth]{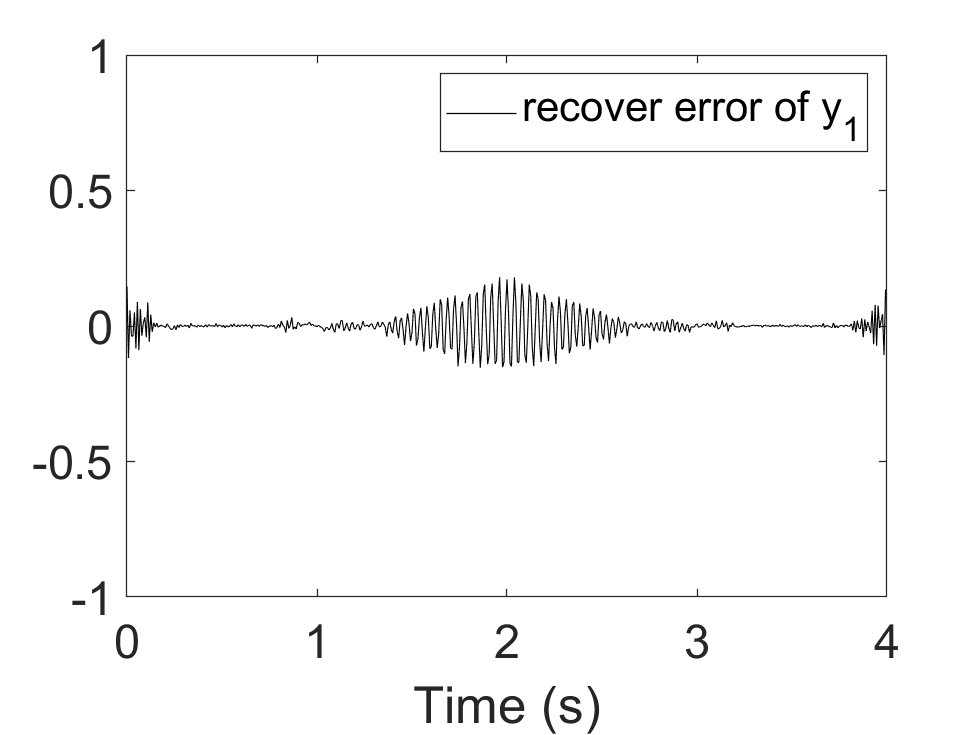}
            \caption{}
        \end{subfigure} &
        \begin{subfigure}[t]{0.22\textwidth}
            \centering
            \includegraphics[width=\linewidth]{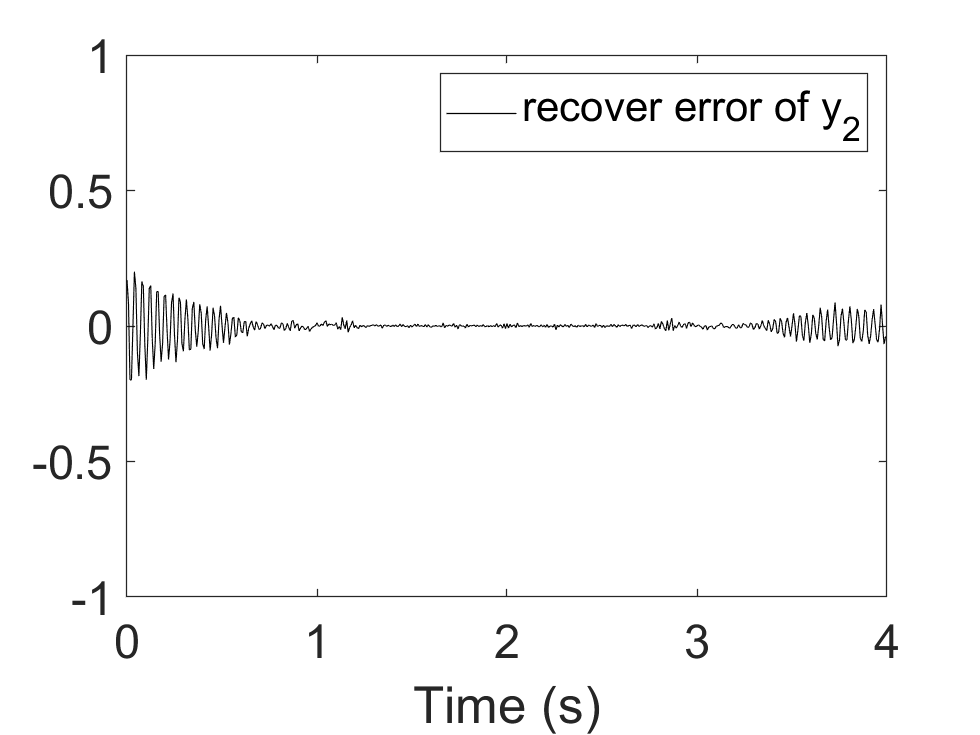}
            \caption{}
        \end{subfigure} &
        \begin{subfigure}[t]{0.22\textwidth}
            \centering
            \includegraphics[width=\linewidth]{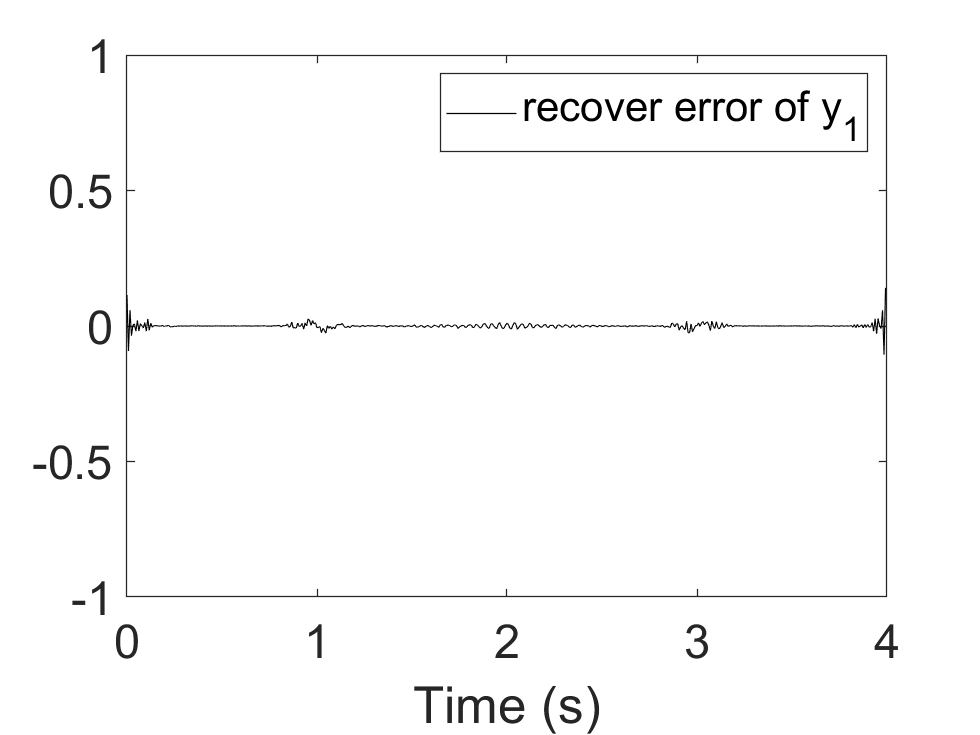}
            \caption{}
        \end{subfigure} &
        \begin{subfigure}[t]{0.22\textwidth}
            \centering
            \includegraphics[width=\linewidth]{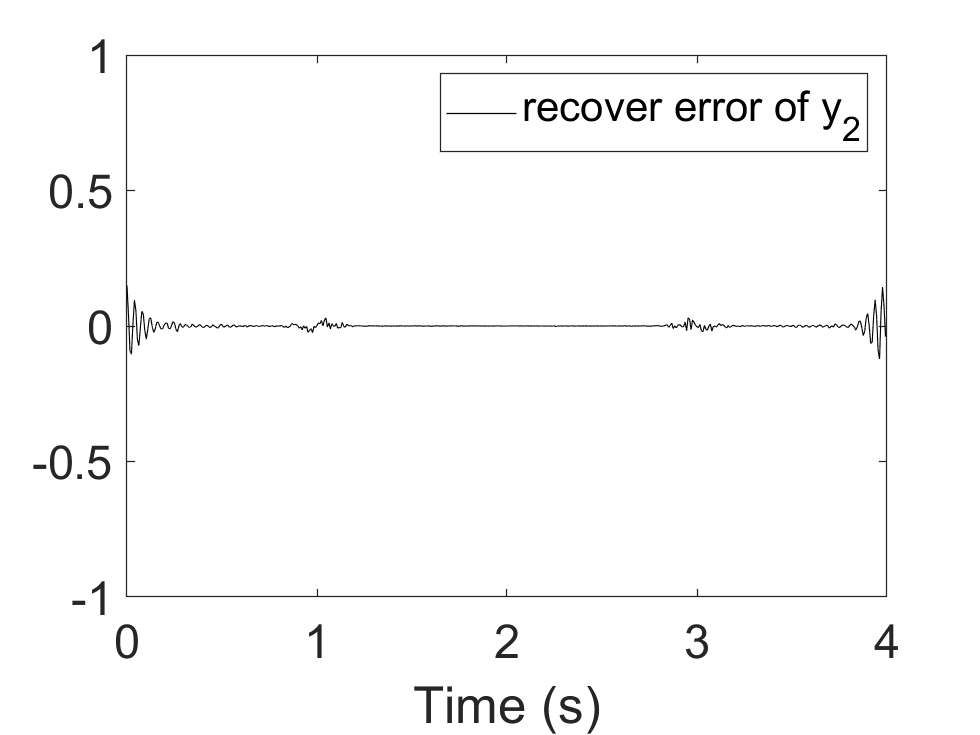}
            \caption{}
        \end{subfigure}
    \end{tabular}
    \vskip -0.3cm
    \caption{Error comparison (real part) in mode recovery using different window parameters in Eq.~\eqref{solution_SSO}: 
    (a)-(b) with \(\psi_{\sigma_1}(t)\); (c)-(d) with \(\psi_{\sigma_2}(t)\). 
    Panels (a) and (c) correspond to recovered \(y_1(t)\), while (b) and (d) correspond to recovered \(y_2(t)\).}
    \label{fig:recovered_errors_sigma12}
\end{figure}
Next, we evaluate the performance of the proposed transforms under noisy conditions by adding Gaussian white noise to the signal Eq.~\eqref{def_yt} at signal-to-noise ratios (SNRs) ranging from 0 to 60 dB. Fig.~\ref{fig:noise_test}(a) shows the behavior of the wavelet parameter $\sigma$ obtained via the R\'enyi entropy Eq.~\eqref{renyi_entropy} with respect to noise level. The optimal parameter $\sigma_1$ obtained via entropy minimization remains stable around 4.9 across all SNR values, indicating that the selection of $\sigma$ is robust to noise.
\begin{figure}[H]
    \centering
    \begin{tabular}{cc}
        \subfloat[$\sigma_1$ values under  noise levels]{
            \includegraphics[width=0.36\textwidth]{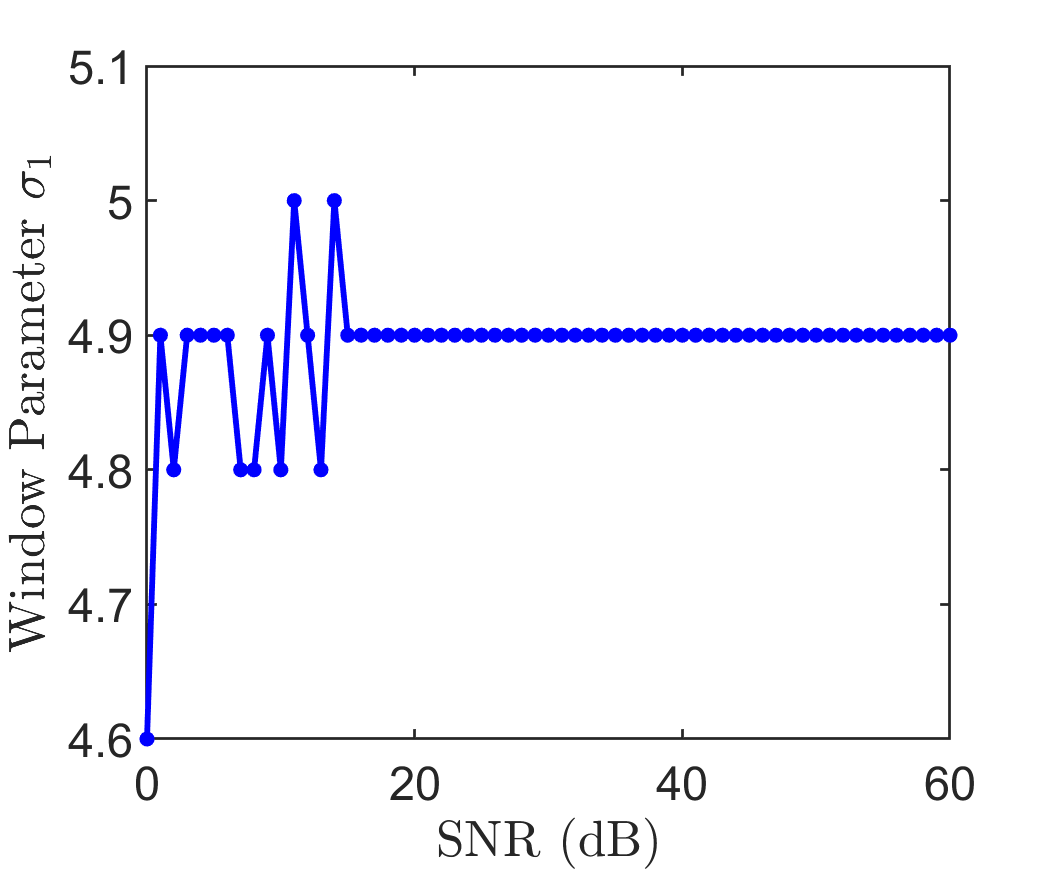}
        } &
        \subfloat[R\'enyi entropy under  noise levels]{
            \includegraphics[width=0.36\textwidth]{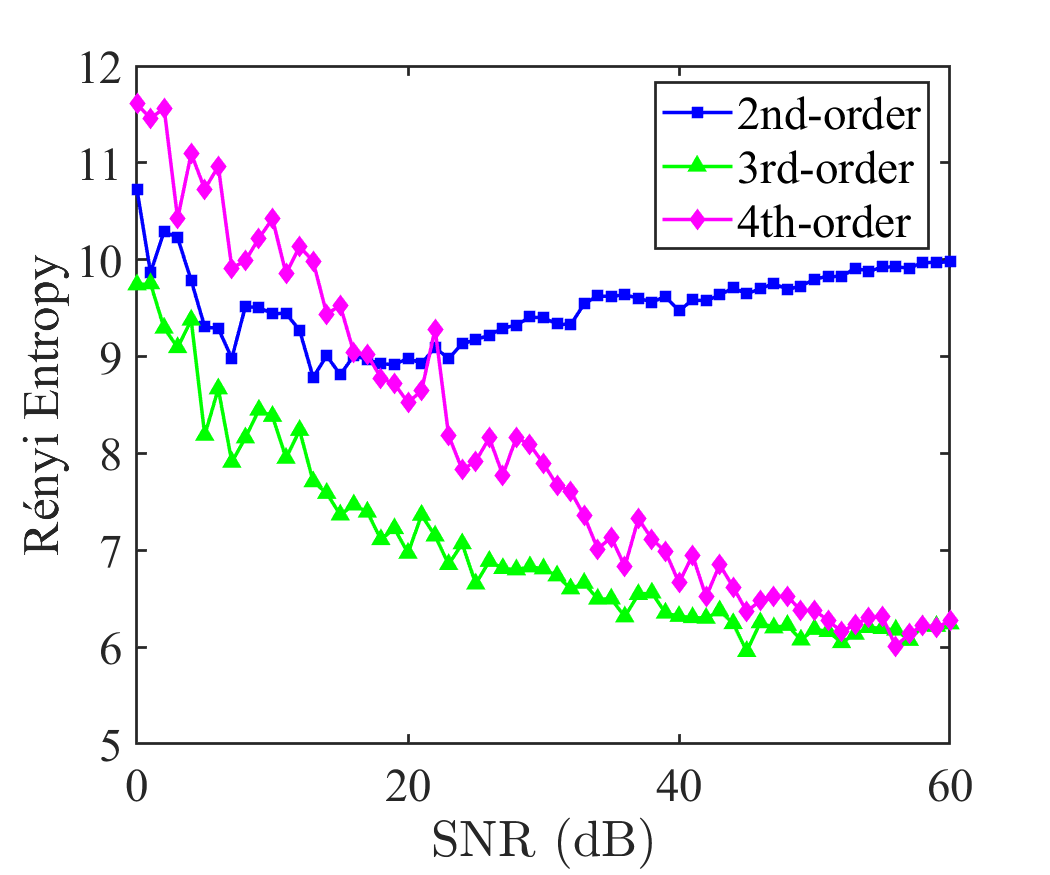}
        } 
    \end{tabular}
 \caption{ Behavior of $\sigma_1$ and post-squeezing TFC R\'enyi entropy under different SNR conditions}
    \label{fig:noise_test}
\end{figure}

We further examine the post-squeezing R\'enyi entropy for different orders of the HSWCT under varying SNR conditions, as shown in Fig.~\ref{fig:noise_test}(b). At low noise levels (SNR > 20 dB), both the third- and fourth-order HSWCT yield lower R\'enyi entropy than their second-order counterpart, indicating better TFC concentration. 
Interestingly, under high-noise conditions, the third-order HSWCT achieves better concentration than both the second- and fourth-order transforms, suggesting that the higher-order transform may be more sensitive to noise.

From the above experimental results, we find that the higher-order HSWCTs have superior performance in capturing  chirprate variations and in mode retrieval compared to their second-order counterparts.
At the end of this section, we validate the effectiveness of the proposed high-order HSWCT using a real-world wolf howling signal\footnote{\text{https://wolfpark.org/Images/Resources/Howls/.}}.
The original signal, sampled at $8$ kHz with a duration of $55$ s, was downsampled by a factor of $8$ to achieve an effective sampling rate of $1$ kHz. This preprocessing step enhances the visibility of crucial low-frequency crossover components below $500$ Hz. 
Such signals have been widely used for evaluating the performance of 3D TFC analysis methods~\cite{chen2023disentangling,chen2024multiple}.

Due to the complexity of directly comparing representations in the 3D TFC space, we present the 2D time-frequency projections of the second-order, third-order, and fourth-order HSWCTs in Fig.~\ref{fig:time-frequency_projections_wolf_howling}. The projected time-frequency representation is defined as:
\[
\mathfrak{T}^N_x(\xi,b) = \int_{\mathbb{R}} \left|\mathcal{U}^{N,\psi_{\sigma_1}}_x(\xi,b,\gamma)\right|^2  d\gamma, \quad N = 2, 3, 4,
\]
and corresponding window parameter $\sigma_1$, obtained from Eq.~\eqref{find_sigma}, is $5.4$.

\begin{figure}[H]
    \centering
    \setlength{\tabcolsep}{3pt}
    \begin{tabular}{ccc}
        \begin{subfigure}[t]{0.28\textwidth} 
            \centering
            \includegraphics[width=\linewidth]{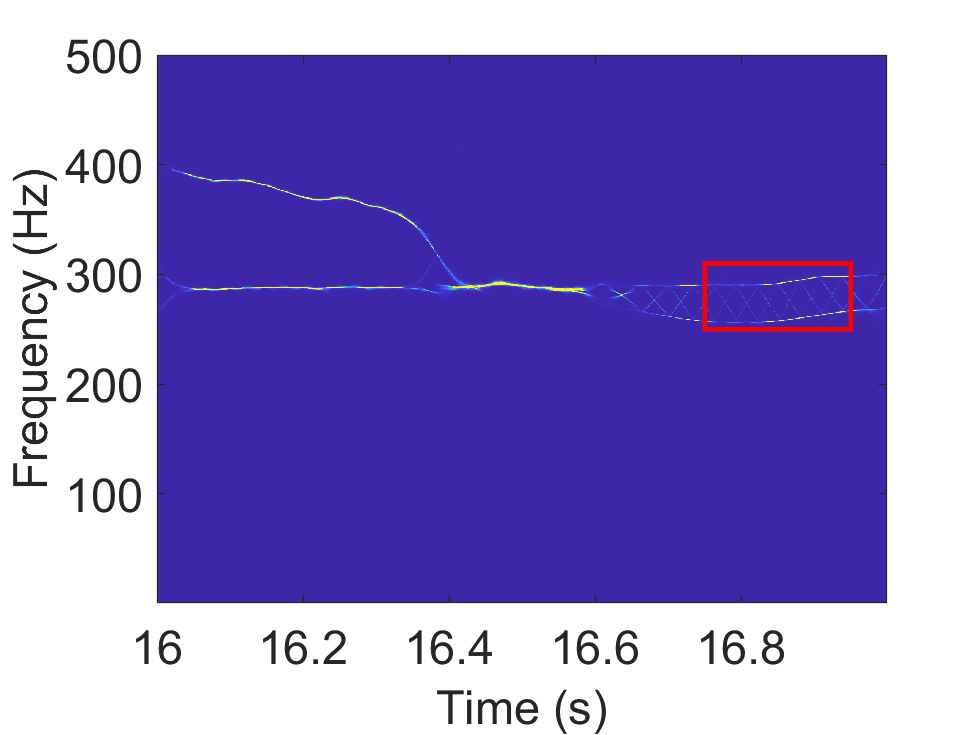}
            \caption{}
        \end{subfigure} &
        \begin{subfigure}[t]{0.28\textwidth}
            \centering
            \includegraphics[width=\linewidth]{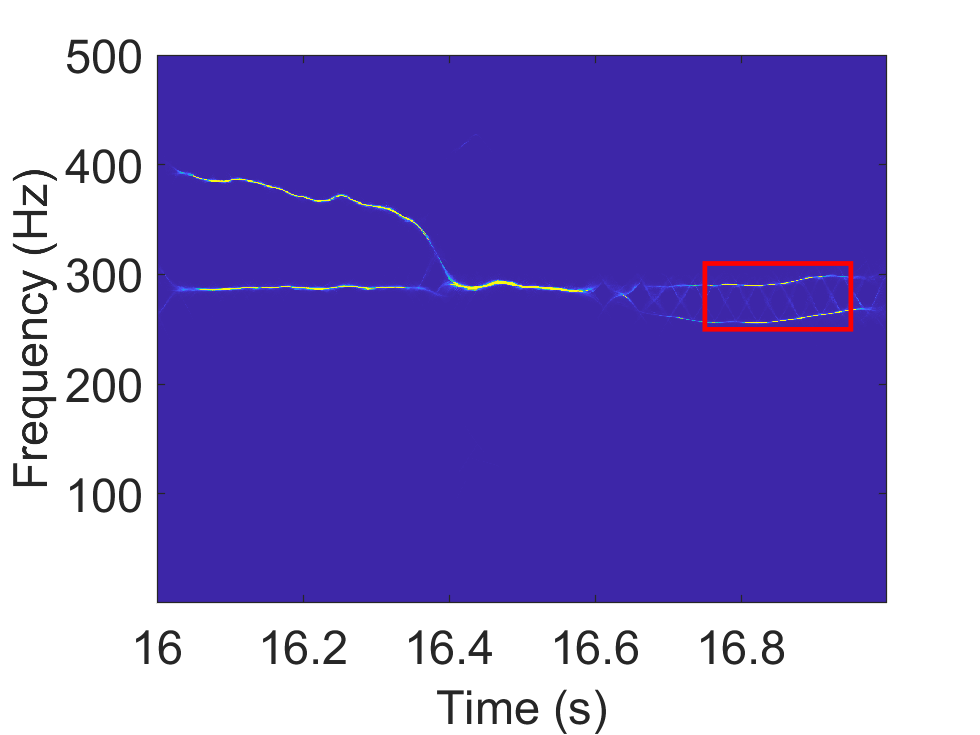}
            \caption{}
        \end{subfigure} &
        \begin{subfigure}[t]{0.28\textwidth}
            \centering
            \includegraphics[width=\linewidth]{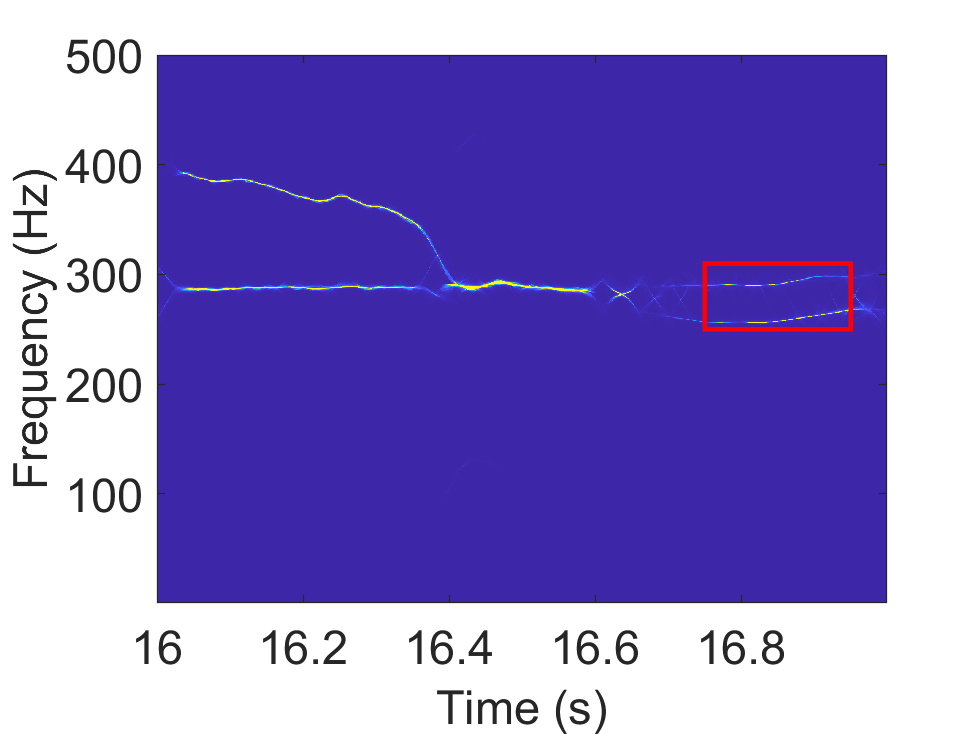}
            \caption{}
        \end{subfigure} \\
        \begin{subfigure}[t]{0.28\textwidth}
            \centering
            \includegraphics[width=\linewidth]{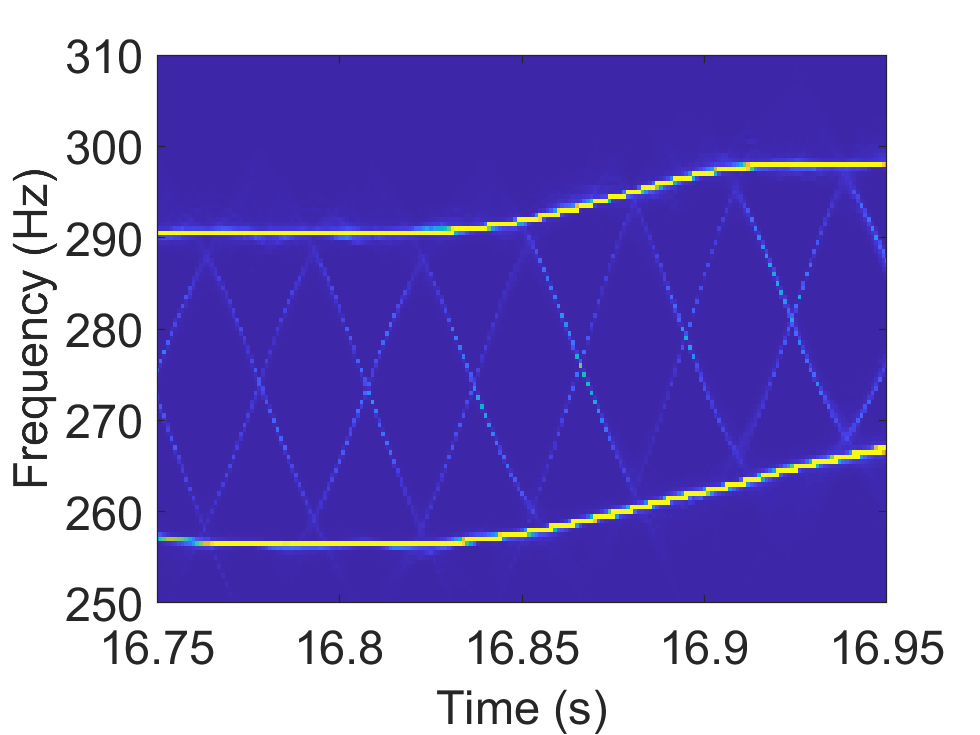}
            \caption{}
        \end{subfigure} &
        \begin{subfigure}[t]{0.28\textwidth}
            \centering
            \includegraphics[width=\linewidth]{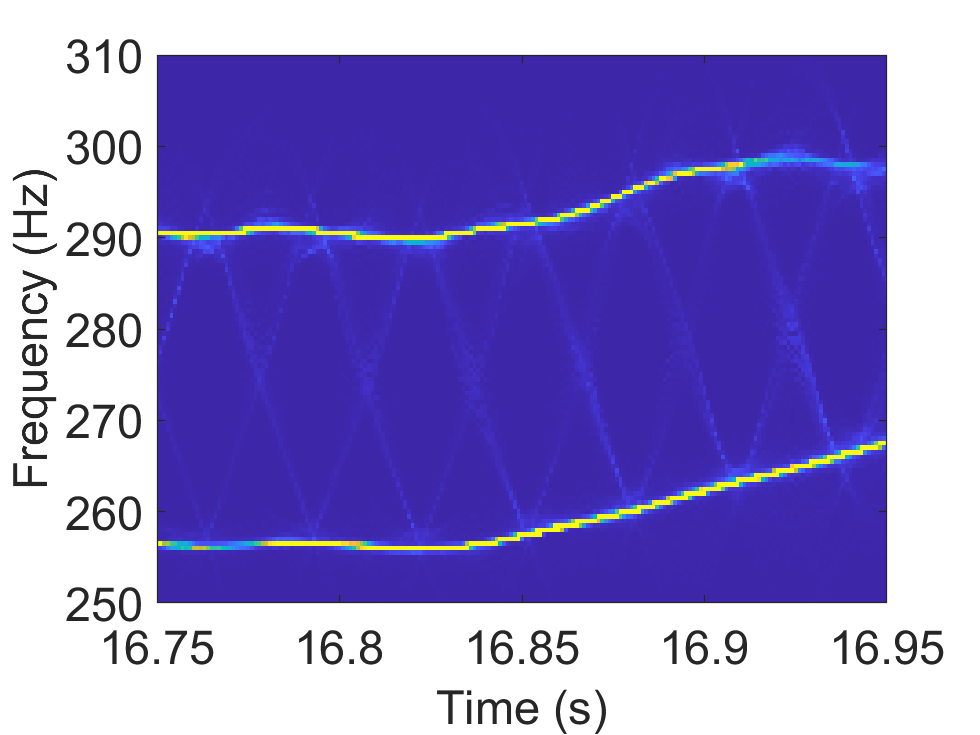}
            \caption{}
        \end{subfigure} &
        \begin{subfigure}[t]{0.28\textwidth}
            \centering
            \includegraphics[width=\linewidth]{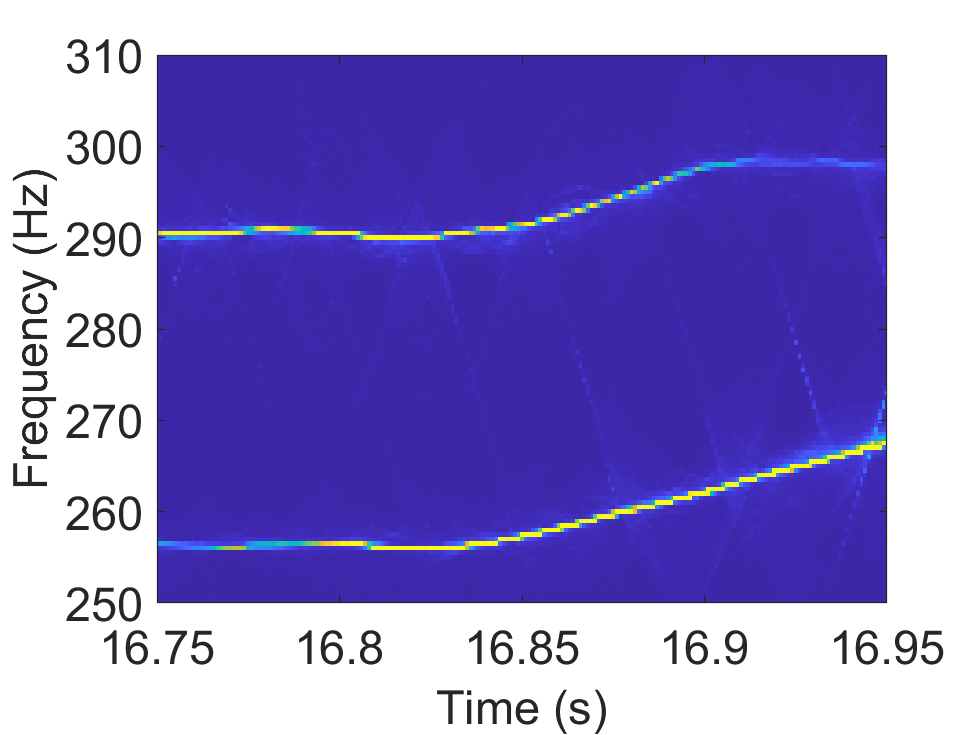}
            \caption{}
        \end{subfigure}
    \end{tabular}
   \caption{Time-frequency projections of the wolf howling signal. (a)-(c) Second-order, third-order, and fourth-order HSWCTs; (d)-(f) corresponding local zooms with red boxes.}
    \label{fig:time-frequency_projections_wolf_howling}
\end{figure}

As shown in Fig.~\ref{fig:time-frequency_projections_wolf_howling}, the signal components are severely merged in the interval \(16.4\)-\(16.6\) s, making separation based solely on frequency or chirprate information infeasible. 
In this region, even higher-order transforms are pushed to their limits, and only a fused frequency curve is observable. However, in the subsequent interval \(16.65\)--\(16.95\) s, as the components begin to separate, the advantage of higher-order HSWCT becomes evident.
Although the separation remains insufficient to fully satisfy condition Eq.~\eqref{separation_condition}, the higher-order transforms still capture the general trend of frequency evolution.
Notably, in the subinterval \(16.75\)-\(16.95\) s, where the signal components are densely distributed in frequency, the second-order projection suffers from significant aliasing, making individual components difficult to distinguish. 
The third-order exhibits only mild aliasing, while the fourth-order successfully suppresses the ambiguity and provides a clear separation of the components.

These observations highlight a key strength of the proposed HSWCT: it maintains superior time-frequency resolution and robust component separability even when the signal operates near the boundaries of the theoretical separation condition. 
This favorable behavior is attributed to the joint estimation of the IF and chirprate. In contrast to conventional time--frequency synchrosqueezing methods, the HSWCT remains capable of revealing the underlying trend of frequency evolution even when the IFs of different signal components are very close.
However, when the separation condition is severely violated, e.g., when both the IFs and chirprates are almost indistinguishable over a long interval, the HSWCT may no longer be capable of reliable separation.
In such extreme scenarios, it may be advantageous to complement the HSWCT with other advanced time-frequency analysis tools.

\section{Conclusion and future work}

In this paper, we have derived a concise general expression for high-order IF and chirprate reassignment operators for multicomponent signals. 
These operators enable a novel HSWCT within the wavelet-chirplet transform framework, extending its capability to handle signals with rapidly varying frequency components and overlapping instantaneous frequency curves.
Besides, rigorous mathematical analyses of the approximation errors for arbitrary-order IF and chirprate reassignment operators are established, providing theoretical insights into how these operators approximate the true IF and chirprate. 
Numerical experiments demonstrate that the proposed  HSWCT achieves superior performance in capturing instantaneous chirprate variations and mode retrieval compared to its second-order counterparts.

Our analysis also reveals certain limitations of the current approach, particularly the degraded performance of the proposed HSWCT under low signal-to-noise ratio  conditions. Future work will therefore aim to enhance its robustness in such environments. 
This goal naturally entails the development of adaptive strategies for window parameter selection and methods to automatically determine the optimal order \(N\) based on signal characteristics. 
Beyond these improvements, improving the computational efficiency of the proposed method is another important direction for future research.

\section*{Acknowledgments} 
This research was partially supported by the National Key Research and Development Program of China  under Grant 2022YFA1005703  and the National Natural Science Foundation of China under Grants  62271230 and and 12571109.

\begin{appendices}
    \section{Proof of Lemma 1} \label{app:proof_lem1}
\begin{proof}
First, we show that the magnitude of the remainder term \(x_{r,k}(a,b,t)\) can be bounded by 
 \begin{align*}
&|x_{r,k}(a,b,t)| = |x_k(b+at) - x_{N,k}(a,b,t)| \\
&\leq  \Big|\left(A_k(b+at) - A_k(b)\right)e^{i 2 \pi \phi_k(b+at)}\Big| + \Big|A_k(b)\Big(e^{i 2 \pi\phi_k(b+at)} - e^{i 2 \pi\phi_k(b)} e^{\sum_{j=1}^{N} \frac{i 2\pi}{j!} \phi^{(j)}_k(b)(at)^j}\Big)\Big|\\
&\leq   \epsilon_1  a |t| +A_k(b) \Big| i 2 \pi\phi_k(b+at)- \sum_{j=0}^{N} 
\frac{i 2\pi}{j!} \phi^{(j)}_k(b)(at)^j\Big|
\\
&\leq   \epsilon_1  a |t| + \epsilon_2 A_k(b) \frac{2\pi}{(N+1)!} a^{N+1} |t|^{N+1}.
\end{align*}
Hence,  
\begin{align*}
      & \left|  U^{t^m \psi}_{x_k} (a,b,\lambda)- \mathcal{R}^{t^m \psi}_{x_k} (a,b,\lambda)\right|=\left|\int_{\mathbb{R}} x_{r,k}(a,b,t) t^m \psi(t)e^{-i\pi \lambda a^2 t^2}dt  \right| \\
     & \leq  \epsilon_1  a I_{m+1}+ \epsilon_2  A_k(b)\frac{2 \pi}{(N+1)!}a^{N+1}I_{N+m+1}, 
\end{align*}
as desired. 
\end{proof}

\section{Proof of Lemma 2 } \label{app:proof_lem2}
\begin{proof} Indeed, we have 
\begin{align*}
    &\Big| \partial_b U_{x_k}^{t^m \psi}(a,b,\lambda)-\sum_{j=1}^{N}\frac{i 2\pi \phi_k^{(j)}(b) }{(j-1)!} a^{j-1}  U_{x_k}^{ t^{m+j-1} \psi }(a,b,\lambda)\Big|\\
  &=\Big|\int_{\mathbb{R}} A'_k(b+at) e^{i 2 \pi \phi_k(b+at)}t^m \psi(t)e^{-i  \pi \lambda a^2 t^2}dt \\
& \quad  +\int_{\mathbb{R}} A_k(b+at) i 2 \pi\phi'_k(b+at)  e^{i 2 \pi \phi_k(b+at)}t^m\psi(t) e^{-i  \pi \lambda a^2 t^2} dt\\
 &\quad  -\sum_{j=1}^{N}\frac{i 2\pi}{(j-1)!} a^{j-1}\phi_k^{(j)}(b)\int_{\mathbb{R}} A_k(b+at) e^{i 2 \pi \phi_k(b+at)}t^{m+j-1} \psi(t)e^{-i  \pi \lambda a^2 t^2} dt\Big|\\
 \end{align*}
 \begin{align*}
  &      
 \le \Big|\int_{\mathbb{R}} A'_k(b+at) e^{i 2 \pi \phi_k(b+at)}t^m \psi(t)e^{-i  \pi \lambda a^2 t^2}dt\Big| \\
        &\quad +\Big|\int_{\mathbb{R}} A_k(b+at) i 2 \pi \Big(
\phi'_k(b+at)- \sum_{j=1}^{N}\frac {(at)^{j-1}}{(j-1)!} \phi_k^{(j)}(b)\Big) e^{i 2 \pi \phi_k(b+at)}t^m\psi(t) e^{-i  \pi \lambda a^2 t^2} dt\Big|\\
  & \le \epsilon_1 I_m +2\pi \int_{\mathbb{R}} A_k(b+at) \frac{\epsilon_2 (|t| a)^N}{N!} |t^m\psi(t)| dt\\
  & \le \epsilon_1 I_m +\epsilon_2\frac{ 2\pi a^N}{N!}\Big(\int_{\mathbb{R}} |A_k(b+at)-A_k(b)| \; |t^{m+N}\psi(t)| dt
+  \int_{\mathbb{R}}  A_k(b) |t^{m+N}\psi(t)| dt\Big)
\\& \le \epsilon_1 I_m +\epsilon_2\frac{ 2\pi a^N}{N!}\big( \epsilon_1 a I_{m+N+1} +A_k(b) I_{m+N}\big). 
    \end{align*} 
\end{proof}

 \section{ Proof of Lemma 3 }\label{app:proof_lem3}
\begin{proof}
By adding and subtracting the terms involving $i 2 \pi \phi_k^{(j)}(b)$, we obtain
\begin{align*}
&\left| \mathrm{Res}_{l, m}(a, b, \lambda) \right| =\Big|\sum_{k=1}^K \Big(\partial_b U_{x_k}^{t^m \psi}(a,b,\lambda) - \sum_{j=1}^{N} \frac{i2\pi \phi_l^{(j)}(b)}{(j-1)!}a^{j-1} U_{x_k}^{t^{m+j-1}\psi}(a,b,\lambda)\Big)\Big|\\
&\le \Big|\sum_{k=1}^K \Big(\partial_b U_{x_k}^{t^m \psi}(a,b,\lambda) - \sum_{j=1}^{N} \frac{i2\pi \phi_k^{(j)}(b)}{(j-1)!}a^{j-1} U_{x_k}^{t^{m+j-1}\psi}(a,b,\lambda)\Big)\Big|\\
&\quad +\Big|\sum_{k=1}^K \Big(\sum_{j=1}^{N} \frac{i2\pi \phi_k^{(j)}(b)}{(j-1)!}a^{j-1} U_{x_k}^{t^{m+j-1}\psi}(a,b,\lambda) - \sum_{j=1}^{N} \frac{i2\pi \phi_l^{(j)}(b)}{(j-1)!}a^{j-1} U_{x_k}^{t^{m+j-1}\psi}(a,b,\lambda)\Big)\Big|\\
&\le \Big|\sum_{k=1}^K \Big(\partial_b U_{x_k}^{t^m \psi}(a,b,\lambda) - \sum_{j=1}^{N} \frac{i2\pi \phi_k^{(j)}(b)}{(j-1)!}a^{j-1} U_{x_k}^{t^{m+j-1}\psi}(a,b,\lambda)\Big)\Big|\\
&\quad +\Big|\sum_{k=1}^K \Big( \sum_{j=1}^{N} i2\pi  \frac{\phi_k^{(j)}(b)-\phi_l^{(j)}(b)}{(j-1)!}a^{j-1} U_{x_k}^{t^{m+j-1}\psi}(a,b,\lambda)\Big)\Big|.
\end{align*}
Then, from Eq.~\eqref{eq:bound_Uxk_proof}, we have
\begin{align*}
\left| \mathrm{Res}_{l, m}(a, b, \lambda) \right|
\le \sum_{k=1}^K C_{k, m}(a,b) + \sum_{k\neq l}\sum_{j=1}^{N} \frac{2 \pi a^{j-1}}{(j-1)!} \Big|\phi_k^{(j)}(b)-\phi_l^{(j)}(b)\Big | \Big( \Pi_{k, m+j-1}(a,b) + A_k(b)  \Upsilon_{k, m+j-1}(b) \Big). 
\end{align*}
This completes the proof.
\end{proof}
\end{appendices}


\bibliographystyle{elsarticle-num-names}
\bibliography{references0402_2026}

@article{amin2016radar,
  title={Radar signal processing for elderly fall detection: The future for in-home monitoring},
  author={Amin, Moeness G and Zhang, Yimin D and Ahmad, Fauzia and Ho, KC Dominic},
  journal={IEEE Signal Processing Magazine},
  volume={33},
  number={2},
  pages={71--80},
  year={2016},
  publisher={IEEE}
}

@article{auger1995improving,
  title={Improving the readability of time-frequency and time-scale representations by the reassignment method},
  author={Auger, Fran{\c{c}}ois and Flandrin, Patrick},
  journal={IEEE Transactions on Signal Processing},
  volume={43},
  number={5},
  pages={1068--1089},
  year={1995},
  publisher={IEEE}
}

@article{baraniuk2002measuring,
  title={Measuring time-frequency information content using the {R}{\'e}nyi entropies},
  author={Baraniuk, Richard G and Flandrin, Patrick and Janssen, Augustus JEM and Michel, Olivier JJ},
  journal={IEEE Transactions on Information Theory},
  volume={47},
  number={4},
  pages={1391--1409},
  year={2001},
  publisher={IEEE}
}

@article{behera2018theoretical,
  title={Theoretical analysis of the second-order synchrosqueezing transform},
  author={Behera, Ratikanta and Meignen, Sylvain and Oberlin, Thomas},
  journal={Applied and Computational Harmonic Analysis},
  volume={45},
  number={2},
  pages={379--404},
  year={2018},
  publisher={Elsevier}
}

@article{chen2023disentangling,
  title={Disentangling modes with crossover instantaneous frequencies by synchrosqueezed chirplet transforms, from theory to application},
  author={Chen, Ziyu and Wu, Hau-Tieng},
  journal={Applied and Computational Harmonic Analysis},
  volume={62},
  pages={84--122},
  year={2023},
  publisher={Elsevier}
}

@article{chen2024multiple,
  title={Multiple enhanced synchrosqueezing in the time--frequency--chirprate space},
  author={Chen, Tao and Xie, Lei and Cui, Mingzhe and Su, Hongye},
  journal={Signal Processing},
  volume={222},
  pages={109541},
  year={2024},
  publisher={Elsevier}
}

@article{chen2024composite,
  title={Composite signal detection using multisynchrosqueezing wavelet transform},
  author={Chen, Xu and Zhang, Zhousuo and Yang, Wenzhan},
  journal={Digital Signal Processing},
  volume={149},
  pages={104482},
  year={2024},
  publisher={Elsevier}
}

@article{chui2016signal,
  title={Signal decomposition and analysis via extraction of frequencies},
  author={Chui, Charles K and Mhaskar, HN},
  journal={Applied and Computational Harmonic Analysis},
  volume={40},
  number={1},
  pages={97--136},
  year={2016},
  publisher={Elsevier}
}

@article{chui2021time,
  title={Time-scale-chirp\_rate operator for recovery of non-stationary signal components with crossover instantaneous frequency curves},
  author={Chui, Charles K and Jiang, Qingtang and Li, Lin and Lu, Jian},
  journal={Applied and Computational Harmonic Analysis},
  volume={54},
  pages={323--344},
  year={2021},
  publisher={Elsevier}
}

@article{chui2023analysis,
  title={Analysis of a direct separation method based on adaptive chirplet transform for signals with crossover instantaneous frequencies},
  author={Chui, Charles K and Jiang, Qingtang and Li, Lin and Lu, Jian},
  journal={Applied and Computational Harmonic Analysis},
  volume={62},
  pages={24--40},
  year={2023},
  publisher={Elsevier}
}

@book{cohen1995time,
  title={Time-frequency Analysis},
  author={Cohen, Leon},
  volume={778},
  year={1995},
  publisher={Prentice Hall PTR New Jersey}
}

@book{daubechies1992ten,
  title={Ten Lectures on Wavelets},
  author={Daubechies, Ingrid},
  year={1992},
  publisher={SIAM}
}

@article{daubechies1996nonlinear,
    author = {Daubechies, Ingrid and Maes, Stéphane},
    title = {A Nonlinear Squeezing of the Continuous Wavelet Transform},
    journal = {Wavelets in Medicine and Biology},
    year = {1996},
    pages = {527--546},
    publisher = {Routledge}
}

@article{daubechies2011synchrosqueezed,
  title={Synchrosqueezed wavelet transforms: An empirical mode decomposition-like tool},
  author={Daubechies, Ingrid and Lu, Jianfeng and Wu, Hau-Tieng},
  journal={Applied and Computational Harmonic Analysis},
  volume={30},
  number={2},
  pages={243--261},
  year={2011},
  publisher={Elsevier}
}

@article{fomel2013seismic,
  title={Seismic data decomposition into spectral components using regularized nonstationary autoregression},
  author={Fomel, Sergey},
  journal={Geophysics},
  volume={78},
  number={6},
  pages={O69--O76},
  year={2013},
  publisher={Society of Exploration Geophysicists}
}

@inproceedings{fourer2019second,
  title={Second-order time-reassigned synchrosqueezing transform: Application to {D}raupner wave analysis},
  author={Fourer, Dominique and Auger, Francois},
 booktitle={2019 27th European Signal Processing Conference (EUSIPCO)},
  pages={1--5},
  year={2019},
  organization={IEEE}
}

@book{grafakos2008classical,
  title={Classical fourier analysis},
  author={Grafakos, Loukas},
  volume={2},
  year={2008},
  publisher={Springer}
}

@article{he2019time,
  title={Time-reassigned synchrosqueezing transform: The algorithm and its applications in mechanical signal processing},
  author={He, Dong and Cao, Hongrui and Wang, Shibin and Chen, Xuefeng},
  journal={Mechanical Systems and Signal Processing},
  volume={117},
  pages={255--279},
  year={2019},
  publisher={Elsevier}
}

@article{he2020gaussian,
  title={Gaussian-modulated linear group delay model: Application to second-order time-reassigned synchrosqueezing transform},
  author={He, Zhoujie and Tu, Xiaotong and Bao, Wenjie and Hu, Yue and Li, Fucai},
  journal={Signal Processing},
  volume={167},
  pages={107275},
  year={2020},
  publisher={Elsevier}
}

@article{hlawatsch1992linear,
  title={Linear and quadratic time-frequency signal representations},
  author={Hlawatsch, Franz and Boudreaux-Bartels, Gloria Faye},
  journal={IEEE Signal Processing Magazine}, 
  volume={9},
  number={2},
  pages={21--67},
  year={1992},
  publisher={IEEE}
}

@article{hu2019high,
  title={High-order synchrosqueezing wavelet transform and application to planetary gearbox fault diagnosis},
  author={Hu, Yue and Tu, Xiaotong and Li, Fucai},
 journal = {Mechanical Systems and Signal Processing},
  volume={131},
  pages={126--151},
  year={2019},
  publisher={Elsevier}
}

@article{jiang2020novel,
  title={A novel parameter estimation for polynomial phase signals using the spectrum phase},
  author={Jiang, Xiaodong and Wu, Siliang},
  journal={IEEE Signal Processing Letters},
  volume={27},
  pages={1919--1923},
  year={2020},
  publisher={IEEE}
}

@article{jiang2025synchrosqueezed,
  title={Synchrosqueezed {X}-ray wavelet--chirplet transform for accurate chirp rate estimation and retrieval of modes from multicomponent signals with crossover instantaneous frequencies},
  author={Jiang, Qingtang and Li, Shuixin and Chen, Jiecheng and Li, Lin},
  journal={Mechanical Systems and Signal Processing},
  volume={238},
  pages={113193},
  year={2025},
  publisher={Elsevier}
}

@article{katkovnik1995new,
  title={A new form of the {F}ourier transform for time-varying frequency estimation},
  author={Katkovnik, Vladimir},
  journal={Signal Processing},
  volume={47},
  number={2},
  pages={187--200},
  year={1995},
  publisher={Elsevier}
}

@article{kinoshita2020sleep,
  title={Sleep spindle detection using {RUSBoost} and synchrosqueezed wavelet transform},
  author={Kinoshita, Takafumi and Fujiwara, Koichi and Kano, Manabu and Ogawa, Keiko and Sumi, Yukiyoshi and Matsuo, Masahiro and Kadotani, Hiroshi},
  journal={IEEE Transactions on Neural Systems and Rehabilitation Engineering},
  volume={28},
  number={2},
  pages={390--398},
  year={2020},
  publisher={IEEE}
}

@article{kodera1978analysis,
  title={Analysis of time-varying signals with small {BT} values},
  author={Kodera, Kunihiko and Gendrin, Roger and Villedary, Claude},
  journal={IEEE Transactions on Acoustics, Speech, and Signal Processing},
  volume={26},
  number={1},
  pages={64--76},
  year={1978},
  publisher={IEEE}
}

@article{li2011local,
  title={Local polynomial {F}ourier transform: A review on recent developments and applications},
  author={Li, Xiumei and Bi, Guoan and Stankovic, Srdjan and Zoubir, Abdelhak M},
  journal={Signal Processing},
  volume={91},
  number={6},
  pages={1370--1393},
  year={2011},
  publisher={Elsevier}
}

@article{li2020adaptive,
  title={Adaptive synchrosqueezing transform with a time-varying parameter for non-stationary signal separation},
  author={Li, Lin and Cai, Haiyan and Jiang, Qingtang},
  journal={Applied and Computational Harmonic Analysis},
  volume={49},
  number={3},
  pages={1075--1106},
  year={2020},
  publisher={Elsevier}
}

@article{li2020adaptivestft,
  title={Adaptive short-time {F}ourier transform and synchrosqueezing transform for non-stationary signal separation},
  author={Li, Lin and Cai, Haiyan and Han, Hongxia and Jiang, Qingtang and Ji, Hongbing},
  journal={Signal Proc.},
  volume={166},
  pages={107231},
  year={2020},
  publisher={Elsevier}
}

@article{li2022self,
  title={Self-matched extracting wavelet transform and signal reconstruction},
  author={Li, Wenting and Auger, Fran{\c{c}}ois and Zhang, Zhuosheng and Zhu, Xiangxiang},
  journal={Digital Signal Processing},
  volume={128},
  pages={103602},
  year={2022},
  publisher={Elsevier}
}

@article{li2022theoretical,
  title={Theoretical analysis of time-reassigned synchrosqueezing wavelet transform},
  author={Li, Wenting and Zhang, Zhuosheng and Auger, Fran{\c{c}}ois and Zhu, Xiangxiang},
  journal={Applied Mathematics Letters},
  volume={132},
  pages={108141},
  year={2022},
  publisher={Elsevier}
}

@article{li2022chirplet,
  title={A chirplet transform-based mode retrieval method for multicomponent signals with crossover instantaneous frequencies},
  author={Li, Lin and Han, Ningning and Jiang, Qingtang and Chui, Charles K},
  journal={Digital Signal Processing},
  volume={120},
  pages={103262},
  year={2022},
  publisher={Elsevier}
}

@article{li2026synchrosqueezed,
title = {Synchrosqueezed windowed linear canonical transform: A method for mode retrieval from multicomponent signals with crossing instantaneous frequencies},
journal = {Signal Processing},
volume = {241},
pages = {110384},
year = {2026},
author = {Li, Shuixin and Chen, Jiecheng  and Jiang, Qingtang  and Lu, Jian },
publisher={Elsevier}
}

@article{liu2023sparse,
  title={Sparse time--frequency analysis of seismic data: Sparse representation to unrolled optimization},
  author={Liu, Naihao and Lei, Youbo and Liu, Rongchang and Yang, Yang and Wei, Tao and Gao, Jinghuai},
  journal={IEEE Transactions on Geoscience and Remote Sensing},
  volume={61},
  pages={1--10},
  year={2023},
  publisher={IEEE}
}

@article{lu2021second,
  title={A second-order synchrosqueezing transform with a simple form of phase transformation},
  author={Lu, Jian and Alzahrani, Jawaher H and Jiang, Qingtang},
  journal={Numerical Mathematics: Theory, Methods and Applications},
  year={2021}
}

@book{mallat1999wavelet,
  title={A Wavelet Tour of Signal Processing},
  author={Mallat, St{\'e}phane},
  year={1999},
  publisher={Elsevier}
}

@article{mann1995chirplet,
  title={The chirplet transform: Physical considerations},
  author={Mann, Steve and Haykin, Simon},
  journal={IEEE Transactions on Signal Processing},
  volume={43},
  number={11},
  pages={2745--2761},
  year={1995},
  publisher={IEEE}
}

@article{mercuri2019vital,
  title={Vital-sign monitoring and spatial tracking of multiple people using a contactless radar-based sensor},
  author={Mercuri, Marco and Lorato, Ilde Rosa and Liu, YaoHong and Wieringa, Fokko and Hoof, Chris Van and Torfs, Tom},
  journal={Nature Electronics},
  volume={2},
  number={6},
  pages={252--262},
  year={2019},
  publisher={Nature Publishing Group}
}

@inproceedings{oberlin2014fourier,
  title={The {F}ourier-based synchrosqueezing transform},
  author={Oberlin, Thomas and Meignen, Sylvain and Perrier, Val{\'e}rie},
 booktitle={2014 IEEE International Conference on Acoustics, Speech and Signal Processing (ICASSP)}, 
  pages={315--319},
  year={2014},
  organization={IEEE}
}

@article{oberlin2015second,
  title={Second-order synchrosqueezing transform or invertible reassignment? {T}owards ideal time-frequency representations},
  author={Oberlin, Thomas and Meignen, Sylvain and Perrier, Val{\'e}rie},
  journal={IEEE Transactions on Signal Processing},
  volume={63},
  number={5},
  pages={1335--1344},
  year={2015},
  publisher={IEEE}
}

@inproceedings{oberlin2017secondwavelet,
  title={The second-order wavelet synchrosqueezing transform},
  author={Oberlin, Thomas and Meignen, Sylvain},
  booktitle={2017 IEEE International Conference on Acoustics, Speech and Signal Processing (ICASSP)}, 
  pages={3994--3998},
  year={2017},
  organization={IEEE}
}

@article{park2011time,
  title={Time-frequency analysis of {EEG} asymmetry using bivariate empirical mode decomposition},
  author={Park, Cheolsoo and Looney, David and Kidmose, Preben and Ungstrup, Michael and Mandic, Danilo P},
  journal={IEEE Transactions on Neural Systems and Rehabilitation Engineering},
  volume={19},
  number={4},
  pages={366--373},
  year={2011},
  publisher={IEEE}
}

@article{pham2017high,
  title={High-order synchrosqueezing transform for multicomponent signals analysis--{W}ith an application to gravitational-wave signal},
  author={Pham, Duong-Hung and Meignen, Sylvain},
  journal={IEEE Transactions on Signal Processing},
  volume={65},
  number={12},
  pages={3168--3178},
  year={2017},
  publisher={IEEE}
}

@article{stankovic2001measure,
  title={A measure of some time--frequency distributions concentration},
  author={Stankovi{\'c}, Ljubi{\v{s}}a},
  journal={Signal Processing},
  volume={81},
  number={3},
  pages={621--631},
  year={2001},
  publisher={Elsevier}
}

@book{stankovic2013time,
  title={Time-frequency signal analysis with applications},
  author={Stankovi{\'c}, Ljubi{\v{s}}a and Dakovi{\'c}, Milo{\v{s}} and Thayaparan, Thayannathan},
  year={2013},
  publisher={Artech House}
}

@article{thakur2013synchrosqueezing,
  title={The synchrosqueezing algorithm for time-varying spectral analysis: Robustness properties and new paleoclimate applications},
  author={Thakur, Gaurav and Brevdo, Eugene and Fu{\v{c}}kar, Neven S and Wu, Hau-Tieng},
  journal={Signal Processing},
  volume={93},
  number={5},
  pages={1079--1094},
  year={2013},
  publisher={Elsevier}
}

@article{wang2008generalized,
  title={Generalized high-order phase function for parameter estimation of polynomial phase signal},
  author={Wang, Pu and Djurovic, Igor and Yang, Jianyu},
  journal={IEEE Transactions on Signal Processing},
  volume={56},
  number={7},
  pages={3023--3028},
  year={2008},
  publisher={IEEE}
}

@article{wang2018matching,
  title={Matching synchrosqueezing transform: A useful tool for characterizing signals with fast varying instantaneous frequency and application to machine fault diagnosis},
  author={Wang, Shibin and Chen, Xuefeng and Selesnick, Ivan W and Guo, Yanjie and Tong, Chaowei and Zhang, Xingwu},
  journal={Mechanical Systems and Signal Processing},
  volume={100},
  pages={242--288},
  year={2018},
  publisher={Elsevier}
}

@article{yu2017synchroextracting,
  title={Synchroextracting transform},
  author={Yu, Gang and Yu, Mingjin and Xu, Chuanyan},
  journal={IEEE Transactions on Industrial Electronics},
  volume={64},
  number={10},
  pages={8042--8054},
  year={2017},
  publisher={IEEE}
}

@article{yu2019local,
title = {Local maximum synchrosqueezing transform: An energy-concentrated time-frequency analysis tool},
author = {Yu, Gang and Wang, Zhonghua  and Zhao, Ping  and  Li, Zhen},
journal = {Mechanical Systems and Signal Processing},
volume = {117},
pages = {537-552},
year = {2019},
publisher={Elsevier}
}

@article{yu2020time,
  title={Time-reassigned multisynchrosqueezing transform for bearing fault diagnosis of rotating machinery},
  author={Yu, Gang and Lin, Tianran and Wang, Zhonghua and Li, Yueyang},
  journal={IEEE Transactions on Industrial Electronics},
  volume={68},
  number={2},
  pages={1486--1496},
  year={2020},
  publisher={IEEE}
}

@article{zhang2022local,
  title={Local maximum frequency-chirp-rate synchrosqueezed chirplet transform},
  author={Zhang, Ran and Wang, Zimeng and Tan, Yu and Yang, Xincheng and Yang, Shenghui},
  journal={Digital Signal Processing},
  volume={130},
  pages={103710},
  year={2022},
  publisher={Elsevier}
}

@article{zhu2020frequency,
  title={Frequency-chirprate reassignment},
  author={Zhu, Xiangxiang and Yang, Haizhao and Zhang, Zhuosheng and Gao, Jinghuai and Liu, Naihao},
  journal={Digital Signal Processing},
  volume={104},
  pages={102783},
  year={2020},
  publisher={Elsevier}
}

@article{zhu2021three,
  title={Three-dimension extracting transform},
  author={Zhu, Xiangxiang and Zhang, Zhuosheng and Gao, Jinghuai},
  journal={Signal Processing},
  volume={179},
  pages={107830},
  year={2021},
  publisher={Elsevier}
}

@article{li2026time,
title = {Time-reassigned synchrosqueezing frequency-domain chirplet transform for multicomponent signals with intersecting group delay curves},
journal = {Mechanical Systems and Signal Processing},
volume = {251},
pages = {114183},
year = {2026},
issn = {0888-3270},
author = {Li, Shuixin and Chen, Jiecheng and Jiang, Qingtang and Li, Lin},
publisher={Elsevier}
}

\end{document}